\documentclass[12pt,a4paper]{article}
\usepackage[utf8]{inputenc}
\usepackage[english]{babel}
\usepackage{graphicx}
\usepackage{amsmath}
\usepackage{amsthm}
\usepackage{amsfonts}
\usepackage{amssymb}
\usepackage{geometry}
\usepackage[longnamesfirst,authoryear]{natbib}
\usepackage{afterpage}
\usepackage{bm}
\usepackage{here}
\usepackage{xcolor}
\usepackage{appendix}

\usepackage[hang,small,bf]{caption}
\usepackage[subrefformat=parens]{subcaption}
\newcommand{\indepe}{\mathop{\perp\!\!\!\perp}}

\newtheorem{Theorem}{Theorem}
\newtheorem{Corollary}{Corollary}

\newtheorem{Lemma}{Lemma}
\newtheorem{Proposition}{Proposition}
\newtheorem{Assumption}{Assumption}

\newtheorem{Remark}{Remark}

\title{Shrinkage Methods for Treatment Choice\footnote{This work was partially supported by JSPS KAKENHI Grant Numbers 22K13373 and 23K12456. We would like to thank the editor, associate editor, and anonymous referees for their careful reading and comments. We also thank Jiaying Gu, Toru Kitagawa, Yuzo Maruyama, Masayuki Sawada, Katsumi Shimotsu, and Kohei Yata for their helpful comments and discussion. We have also benefited from comments made by participants of the Kansai econometrics study group, the 2022 Asia Meeting of the Econometric Society, and the 3rd Tohoku-ISM-UUlm Joint Workshop.}
}
\author{Takuya Ishihara
\thanks{Graduate School of Economics and Management, Tohoku University, 41 Kawauchi, Aoba-ku, Sendai, Miyagi 980-0862, Japan. \textit{E-mail adress}: \texttt{takuya.ishihara.b7@tohoku.ac.jp}}
\and Daisuke Kurisu 
\thanks{Center for Spatial Information Science, The University of Tokyo, 5-1-5 Kashiwanoha, Kashiwa-shi, Chiba 277-8568, Japan. \textit{E-mail adress}: \texttt{daisukekurisu@csis.u-tokyo.ac.jp}}
}
\date{\today}

\begin{document}

\maketitle

\begin{abstract}
This study examines the problem of determining whether to treat individuals based on observed covariates. The most common decision rule is the conditional empirical success (CES) rule proposed by \cite{manski2004statistical}, which assigns individuals to treatments that yield the best experimental outcomes conditional on the observed covariates. Conversely, using shrinkage estimators, which shrink unbiased but noisy preliminary estimates toward the average of these estimates, is a common approach in statistical estimation problems because it is well-known that shrinkage estimators may have smaller mean squared errors than unshrunk estimators. Inspired by this idea, we propose a computationally tractable shrinkage rule that selects the shrinkage factor by minimizing an upper bound of the maximum regret. Then, we compare the maximum regret of the proposed shrinkage rule with those of the CES and pooling rules when the space of conditional average treatment effects (CATEs) is correctly specified or misspecified. Our theoretical results demonstrate that the shrinkage rule performs well in many cases and these findings are further supported by numerical experiments. Specifically, we show that the maximum regret of the shrinkage rule can be strictly smaller than those of the CES and pooling rules in certain cases when the space of CATEs is correctly specified. In addition, we find that the shrinkage rule is robust against misspecification of the space of CATEs. Finally, we apply our method to experimental data from the National Job Training Partnership Act Study.
\end{abstract}

\section{Introduction}
This study examines the problem of determining whether to treat individuals based on observed covariates. The most common decision rule is the conditional empirical success (CES) rule proposed by \cite{manski2004statistical}, which is a rule assigning individuals to treatments that yield the best experimental outcomes conditional on the observed covariates. The CES rule uses only the average treatment effect (ATE) estimate conditional on each covariate value. By contrast, a common method in statistical estimation problems is to shrink unbiased but noisy preliminary estimates toward the average of these estimates. It is well known that shrinkage estimators may have smaller mean squared errors than unshrunk estimators. This study assumes that the dispersion of conditional ATEs (CATEs) is bounded and proposes a shrinkage rule that assigns individuals to treatments based on shrinkage estimators. We also propose a method to select the shrinkage factor by minimizing an upper bound of the maximum regret. By considering the treatment rules for individuals that are based not only on each CATE but also on the CATEs of others, it is possible to incorporate information across individuals. This allows the proposed shrinkage rule to perform as well as or better than existing treatment rules in the sense of maximum regret and to be more flexible to the heterogeneity of individuals. In addition, we compare the shrinkage rule with other rules when the space of CATEs is correctly specified or misspecified.

The contributions of this study are as follows. First, our approach is attractive from a computational perspective. The computation of the exact minimax regret rule is often challenging in the context of statistical treatment choice. Indeed, when the space of CATEs is restricted and the number of possible covariate values is large, it is difficult to obtain a shrinkage rule that minimizes the maximum regret. To overcome this problem, we propose a shrinkage rule that minimizes a tractable upper bound of the maximum regret. In this approach, each shrinkage factor is obtained by optimizing a single parameter and hence the proposed shrinkage rule is easy to compute.

Second, we compare the maximum regret of the shrinkage rule with those of alternative rules when the space of CATEs is correctly specified. As an alternative to the CES and shrinkage rules, one could consider using the pooling rule that determines whether to treat the individuals based on the average of the CATE estimates. As the CES and pooling rules are special cases of shrinkage rules, the proposed shrinkage rule is expected to outperform these two rules. However, because the proposed shrinkage rule does not minimize the exact maximum regret, its maximum regret may be larger than those of the CES and pooling rules. Therefore, it is essential to compare the maximum regrets. In Section \ref{Sec:main}, we derive the conditions under which the proposed shrinkage rule outperforms the CES or pooling rules. If the dispersion of the CATEs is small compared with the standard deviations of the CATE estimates, then the maximum regret of the proposed shrinkage rule is less than that of the CES rule under homoscedasticity. A detailed definition of the dispersion of CATEs is provided in Section \ref{Sec:shrink}. Intuitively, when the space of CATEs is small enough, the CES rule can be improved by shrinking each CATE estimate toward the average of these estimates. We also demonstrate that the shrinkage rule outperforms the pooling rule when the dispersion of the CATEs is sufficiently large. Furthermore, combined with these results, we show that the proposed shrinkage rule outperforms both the CES and pooling rules when the dispersion is moderate.

Third, we evaluate the maximum regret of the shrinkage rule when the space of CATEs is misspecified. The choice of the space of CATEs is important in practice because the minimax decision rule depends on the space of CATEs. For example, \cite{armstrong2018optimal,armstrong2021finite} consider the minimax estimation and inference problem for treatment effects and show that it is not possible to choose the parameter space automatically in a data-driven manner. Hence, it is crucial to analyze the decision rule under the misspecification of the space of CATEs. We investigate the performance of the shrinkage rule and show that our results are robust to the misspecification of the space of CATEs. To the best of our knowledge, this is the first study to consider the misspecification of the parameter space in the treatment choice problem.

Consequently, this study contributes to the growing literature on statistical treatment choice initiated by \cite{manski2000identification, manski2004statistical}. Following \cite{manski2004statistical, manski2007minimax}, \cite{hirano2009asymptotics}, \cite{stoye2009minimax, stoye2012minimax}, and \cite{tetenov2012statistical}, we focus on the maximum regrets of statistical treatment rules. Similar to \cite{stoye2012minimax}, \cite{tetenov2012statistical}, \cite{ishihara2021evidence}, \cite{olea2023decision}, and \cite{yata2021optimal}, we assume that CATE estimates are normally distributed. This assumption can be approximately justified by the asymptotic normality.

The analysis in this study most closely relates to that of \cite{stoye2012minimax}, who also considers Gaussian experiments for CATEs and investigates the properties of the minimax regret treatment rule when CATEs depend on covariates with bounded variations. \cite{stoye2012minimax} shows that the CES rule achieves minimax regret when the dispersion of CATEs is sufficiently large, and the pooling rule achieves minimax regret when the dispersion is sufficiently small. These results imply that the CES rule is optimal when the CATEs vary significantly depending on the values of the covariates, and the pooling rule is optimal when the CATEs take almost the same values. However, we do not know the minimax regret rule when the dispersion is moderate. Given such circumstances, we propose a shrinkage treatment rule that includes both the CES and pooling rules as special cases and compare the shrinkage rule with the CES and pooling rules in terms of maximum regrets for any dispersion value.


The remainder of this paper is organized as follows. Section \ref{Sec:decision-prob} explains the decision problem and introduces the shrinkage rules. Section \ref{Sec:shrink} proposes a shrinkage rule that selects the shrinkage factor by minimizing the maximum regret's upper bound and analyzes the shrinkage rule's properties. Section \ref{Sec:main} compares the maximum regrets of shrinkage, CES, and pooling rules when the space of CATEs is correctly specified and misspecified. Section \ref{Sec:example} presents numerical analyses to compare the shrinkage rule with the CES and pooling rules. As an illustration, we apply our method to the experimental data from the National Job Training Partnership Act (JTPA) Study in Section \ref{Sec:real}. Finally, Section \ref{Sec:conclusion} concludes the paper.

\section{The decision problem}\label{Sec:decision-prob}

\subsection{Settings}\label{SubSec:setting}
Suppose that we have experimental data $\{(Y_i,D_i,X_i)\}_{i=1}^n$, where $X_i \in  \{x_1, \ldots, x_K\}$ is a discrete covariate, $D_i \in \{0,1\}$ is a binary indicator of the treatment, and $Y_i$ is a post-treatment outcome. Suppose that $X_i$ represents a group and we want to determine whether to treat individuals in each group based on the data. For example, the group is determined by an individual's demographics, school, firm, or region. This setting is similar to that of \cite{manski2004statistical}, who proposes the CES rule.

The CES rule assigns individuals to treatments that yield the best experimental outcomes conditional on covariates. For $k \in 1, \ldots, K$, we define
\begin{eqnarray}
\theta_k & \equiv & E\left[ Y_i | D_i=1, X_i = x_k \right] - E\left[ Y_i | D_i=0, X_i = x_k \right], \nonumber \\
\hat{\theta}_k & \equiv & \frac{1}{n_{1,k}} \sum_{i : D_i = 1, X_i = x_k} Y_i - \frac{1}{n_{0,k}} \sum_{i : D_i = 0, X_i = x_k} Y_i, \nonumber
\end{eqnarray}
where $n_{d,k} \equiv \sum_{i=1}^n 1\{D_i = d, X_i = x_k\}$. Letting $Y_i(0), Y_i(1)$ be the potential outcomes, then $\theta_k$ can be interpreted as the ATE conditional on $x_k$, $E[Y_i(1)-Y_i(0)|X_i=x_k]$, under the unconfoundedness assumption $(Y_i(0),Y_i(1)) \indepe D_i | X_i$. In addition, $\hat{\theta}_k$ is a natural estimator of $\theta_k$. The CES rule determines whether to treat individuals with $x_k$ based on the sign of $\hat{\theta}_k$. Then, the treatment rules can be viewed as a map from estimates $\hat{\bm{\theta}} \equiv (\hat{\theta}_1, \ldots, \hat{\theta}_K)'$ to the binary decisions of the treatment choice. Hence, the CES rule can be expressed as follows:
\begin{equation}
\hat{\bm{\delta}}^{\text{CES}}(\hat{\bm{\theta}}) \ \equiv \ \left( \hat{\delta}_1^{\text{CES}}(\hat{\bm{\theta}}), \ldots, \hat{\delta}_K^{\text{CES}}(\hat{\bm{\theta}}) \right)', \ \text{where $\hat{\delta}_k^{\text{CES}}(\hat{\bm{\theta}}) \equiv 1 \{ \hat{\theta}_k \geq 0 \}$.} \label{CES_rule}
\end{equation}

Because $\hat{\theta}_k$ is consistent and asymptotically normal under some weak conditions, we assume that $\hat{\theta}_1, \ldots, \hat{\theta}_K$ are independently distributed and
\begin{equation}
\hat{\theta}_k \ \sim \ N (\theta_k, \sigma_k^2), \ \ \ k = 1, \ldots , K, \label{normality}
\end{equation}
where $\sigma_k$ is the standard deviation of $\hat{\theta}_k$. We assume that $\sigma_k$ is known. In practice, we can only construct a consistent estimator for $\sigma_k$. This assumption can be approximately justified by the asymptotic normality. If $\hat{\theta}_k$ has the asymptotic normality, the distribution of $\hat{\theta}_k$ is approximated by $N(\theta_k,\sigma_k^2)$. Because the treatment effect can vary with observable individual characteristics, we allow $\theta_k$ to vary across the covariates.

\subsection{Welfare and regret}\label{SubSec:welfare}
Given a treatment choice action $\bm{\delta} \equiv (\delta_1, \cdots, \delta_K)' \in \{0,1\}^K$, we define the welfare attained at $\bm{\delta}$ as follows:
\begin{equation}
W(\bm{\theta}, \bm{\delta}) \ \equiv \ \sum_{k=1}^K p_k \cdot \left\{ \theta_k \cdot \delta_k + \mu_{0,k} \right\}, \label{welfare}
\end{equation}
where $\bm{\theta}=(\theta_1,\dots,\theta_K)'$, $p_k \equiv P(X = x_k)$ and $\mu_{d,k} \equiv E[Y| D=d, X = x_k]$ for $d \in \{0,1\}$ and $k=1, \ldots, K$. Note that $\theta_k$ is written as $\theta_k = \mu_{1,k} - \mu_{0,k}$. If we know the true value of $\bm{\theta}$, then the optimal treatment choice action is given by
$$
\bm{\delta}^{\ast} \ \equiv \ \left( \delta_1^{\ast}, \cdots, \delta_K^{\ast} \right)' \ \equiv \ \left( 1\{\theta_1 \geq 0\}, \ldots, 1\{\theta_K \geq 0\} \right)'.
$$
However, the treatment choice action $\bm{\delta}^{\ast}$ is infeasible because the true value of $\bm{\theta}$ is unknown.

Let $\hat{\bm{\delta}} : \mathbb{R}^K \to \{0,1\}^K$ be a treatment rule that maps the estimates $\hat{\bm{\theta}}$ to the binary decisions of treatment choice. The welfare regret of $\hat{\bm{\delta}}(\hat{\bm{\theta}}) \equiv \left( \hat{\delta}_1(\hat{\bm{\theta}}), \ldots, \hat{\delta}_K(\hat{\bm{\theta}}) \right)'$ is defined as
\begin{eqnarray}
R(\bm{\theta},\hat{\bm{\delta}}) & \equiv & E_{\bm{\theta}} \left[ W(\bm{\theta}, \bm{\delta}^*) - W(\bm{\theta}, \hat{\bm{\delta}}(\hat{\bm{\theta}})) \right] \nonumber \\
&=& \sum_{k=1}^K p_k \cdot \left[ \theta_k \cdot \left\{ \delta_k^* - E_{\bm{\theta}}[\hat{\delta}_k(\hat{\bm{\theta}})] \right\} \right], \label{regret}
\end{eqnarray}
where $E_{\bm{\theta}}$ is the expectation with respect to the sampling distribution of estimates $\hat{\bm{\theta}}$ given the parameters $\bm{\theta}$. Following existing studies, we evaluate the treatment rule $\hat{\bm{\delta}}$ using the maximum regret
\[
\max_{\bm{\theta} \in \Theta} R(\bm{\theta},\hat{\bm{\delta}}),
\]
where $\Theta$ is the space of $\bm{\theta}$. The minimax regret criterion selects the statistical treatment rule that minimizes the maximum regret.

\subsection{Shrinkage rules}\label{SubSec:shrink}
The CES rule does not use $\hat{\theta}_l$ for $l \neq k$ to determine whether or not to treat individuals with $x_k$. However, in the problem of estimating $\bm{\theta} \equiv ( \theta_1, \ldots, \theta_K)'$, a common method is to shrink $\hat{\theta}_k$ toward the average of estimates $\mathrm{ave}(\hat{\bm{\theta}}) \equiv \frac{1}{K} \sum_{k=1}^K \hat{\theta}_k$ and it is well known that shrinkage estimators may have smaller mean squared errors than unshrunk estimators. Hence, we propose the following shrinkage rules $\hat{\bm{\delta}}^{\bm{w}} : \mathbb{R}^K \to \{0,1\}^K$ for $\bm{w} \equiv (w_1, \ldots, w_K)' \in [0,1]^K$.
\begin{equation}
\hat{\bm{\delta}}^{\bm{w}}(\hat{\bm{\theta}}) \ \equiv \ \left( \hat{\delta}_{1}^{w_1}(\hat{\bm{\theta}}), \ldots, \hat{\delta}_{K}^{w_K}(\hat{\bm{\theta}})  \right)', \label{shrinkage_rule}
\end{equation}
where
\begin{equation}
\hat{\delta}_{k}^{w_k}(\hat{\bm{\theta}}) \ \equiv \ 1 \left\{ w_k \cdot \hat{\theta}_k + (1-w_k) \cdot \mathrm{ave}(\hat{\bm{\theta}}) \geq 0 \right\}. \nonumber
\end{equation}
When the vector of shrinkage factors $\bm{w}$ is $\bm{1} \equiv (1,\ldots,1)'$, the shrinkage rule $\hat{\bm{\delta}}^{\bm{w}}$ becomes the CES rule $\hat{\bm{\delta}}^{\text{CES}}$ defined in (\ref{CES_rule}). Hence, the class of shrinkage rules contains the CES rule as a special case. Furthermore, when $\bm{w}$ is $\bm{0} \equiv (0,\ldots,0)'$, this rule becomes the pooling rule $\hat{\bm{\delta}}^{\text{pool}}(\hat{\bm{\theta}}) \equiv \hat{\bm{\delta}}^{\bm{0}}(\hat{\bm{\theta}})$.

From (\ref{normality}), we observe that
\begin{eqnarray*}
& w_k \cdot \hat{\theta}_k + (1-w_k) \cdot \mathrm{ave}(\hat{\bm{\theta}}) \ \sim \ N \left( w_k \cdot \theta_k + (1-w_k) \cdot \overline{\theta}, s_k^2(w_k) \right),
\end{eqnarray*}
where $\overline{\theta} \equiv K^{-1}\sum_{k=1}^K \theta_k$ and $s_k^2(w_k)$ is the variance of $w_k \cdot \hat{\theta}_k + (1-w_k) \cdot \mathrm{ave}(\hat{\bm{\theta}})$, that is,
$$
s_k^2(w_k) \ = \ \left\{ w_k^2 + 2w_k(1-w_k)/K \right\} \sigma_k^2 + (1-w_k)^2 \left\{ K^{-2} \sum_{k=1}^K \sigma_k^2 \right\}.
$$
Hence, from (\ref{regret}), the welfare regret of shrinkage treatment rule $\hat{\bm{\delta}}^{\bm{w}}(\hat{\bm{\theta}})$ can be written as follows:
\begin{eqnarray}
R(\bm{\theta},\hat{\bm{\delta}}^{\bm{w}}) &=& \sum_{k=1}^K p_k \cdot \left[ \theta_k \cdot \left\{ 1\{\theta_k \geq 0\} - \Phi \left( \frac{w_k \cdot \theta_k + (1-w_k) \cdot \overline{\theta}}{s_k(w_k)} \right) \right\} \right] \nonumber \\
&=& \sum_{k=1}^K p_k \cdot \left\{ |\theta_k| \cdot  \Phi \left( - \text{sgn}(\theta_k) \cdot \frac{\theta_k - (1-w_k) (\theta_k - \overline{\theta})}{s_k(w_k)} \right) \right\} \nonumber \\
&=& \sum_{k=1}^K p_k \cdot \left\{ |\theta_k| \cdot  \Phi \left( -   \frac{|\theta_k| - (1-w_k) \cdot \text{sgn}(\theta_k) (\theta_k - \overline{\theta})}{s_k(w_k)} \right) \right\}, \label{regret_shrinkage}
\end{eqnarray}
where $\text{sgn}(x) \equiv 1\{x>0\} -1\{x<0\}$, $\Phi(\cdot)$ is the distribution function of $N(0,1)$, and the second equality follows from the symmetry of the normal distributions.

\section{Shrinkage methods}\label{Sec:shrink}

\subsection{Choice of the shrinkage factors}\label{SubSec:choice}
In this section, we consider how to choose the shrinkage factors $\bm{w}$ under the following assumption.

\begin{Assumption}\label{ass:parameter_space}
For a positive constant $\kappa > 0$, the parameter $\bm{\theta}$ satisfies the following condition:
$$
\left| \theta_k - \overline{\theta} \right| \ \leq \ \kappa, \ \ \ k = 1, \ldots, K.
$$
\end{Assumption}

Under Assumption \ref{ass:parameter_space}, the space of CATEs $\bm{\theta}$ becomes
\begin{equation}
\Theta(\kappa) \ \equiv \ \left\{ \bm{\theta} = (\theta_1, \ldots, \theta_K)' \in \mathbb{R}^K : \left| \theta_k - \overline{\theta} \right| \ \leq \ \kappa \right\}, \label{Theta}
\end{equation}
where the constant $\kappa$ can be interpreted as controlling the dispersion of parameters $\theta_{k}$ around the mean $\overline{\theta}$. This assumption is similar to Assumption 1 in \cite{stoye2012minimax}. \cite{stoye2012minimax} assumes that $| \mu_{d,k} - \mu_{d,l} | \leq \kappa$ for all $d \in \{0,1\}$ and $k,l \in \{1, \ldots, K\}$, where $\mu_{d,k} = E[Y|D=d,X=x_k]$. Because $\theta_k = \mu_{1,k} - \mu_{0,k}$, this assumption implies that
\begin{equation}
\left| \theta_k - \theta_l \right| \ \leq \ 2\kappa, \ \ \ k,l \in \{1, \ldots, K\}. \label{ass_Stoye}
\end{equation}
If Assumption \ref{ass:parameter_space} holds, then we have $| \theta_k - \theta_l | \leq | \theta_k - \overline{\theta} | + | \theta_l - \overline{\theta} | \leq 2 \kappa$; thus (\ref{ass_Stoye}) is satisfied. Conversely, if (\ref{ass_Stoye}) holds, then Assumption \ref{ass:parameter_space} is satisfied by replacing $\kappa$ with $2\kappa$. Because the shrinkage location is the average of $\hat{\bm{\theta}}$, we restrict the difference between $\theta_k$ and the average of $\bm{\theta}$. In Section \ref{SubSec:reg}, we consider other shrinkage locations and impose the assumption that corresponds to the location.

The minimax regret criterion selects the shrinkage factors that minimize the maximum regret. From (\ref{regret_shrinkage}), the optimal shrinkage factors are obtained by minimizing the following:
\[
\max_{\bm{\theta} \in \Theta(\kappa)} \sum_{k=1}^K p_k \cdot \left\{ |\theta_k| \cdot  \Phi \left( -   \frac{|\theta_k| - (1-w_k) \cdot \text{sgn}(\theta_k) (\theta_k - \overline{\theta})}{s_k(w_k)} \right) \right\}.
\]
However, obtaining the optimal shrinkage factors becomes computationally challenging when $K$ is large. To overcome this problem, we propose selecting shrinkage factors that minimize an upper bound of the maximum regret. Because we have $\text{sgn}(\theta_k) (\theta_k - \overline{\theta}) \leq \kappa$ for any $\bm{\theta} \in \Theta(\kappa)$ and $\Phi(\cdot)$ is an increasing function, the regret of shrinkage rule $\hat{\bm{\delta}}^{\bm{w}}(\hat{\bm{\theta}})$ is bounded above by
\begin{equation}
\sum_{k=1}^K p_k \cdot \left\{ |\theta_k| \cdot  \Phi \left( -   \frac{|\theta_k| - (1-w_k) \cdot \kappa}{s_k(w_k)} \right) \right\}. \label{upper_bound_regret}
\end{equation}
Using this upper bound, we obtain
\begin{eqnarray}
\max_{\bm{\theta} \in \Theta(\kappa)} R(\bm{\theta},\hat{\bm{\delta}}^{\bm{w}}) & \leq & \max_{\bm{\theta} \in \Theta(\kappa)} \sum_{k=1}^K p_k \cdot \left\{ |\theta_k| \cdot  \Phi \left( -   \frac{|\theta_k| - (1-w_k) \cdot \kappa}{s_k(w_k)} \right) \right\} \nonumber \\
&\leq & \sum_{k=1}^K p_k \cdot \max_{\theta_k \in \mathbb{R}} \left\{ |\theta_k| \cdot  \Phi \left( -   \frac{|\theta_k| - (1-w_k) \cdot \kappa}{s_k(w_k)} \right) \right\} \nonumber \\
&=& \sum_{k=1}^K p_k \cdot s_k(w_k) \eta \left( \frac{(1-w_k) \cdot \kappa}{s_k(w_k)} \right), \label{max_regret}
\end{eqnarray}
where $\eta(a) \equiv \max_{t \geq 0} \left\{ t \cdot \Phi(-t+a) \right\} = \max_{t \in \mathbb{R}} \left\{ |t| \cdot \Phi(-|t|+a) \right\}$ for $a \in \mathbb{R}$. \cite{tetenov2012statistical} and \cite{ishihara2021evidence} show that function $\eta(\cdot)$ is strictly increasing and convex. Figure \ref{fig:eta} displays the shape of $\eta(a)$. Using function $\eta(\cdot)$, we propose the following shrinkage factors:
\begin{equation}
\bm{w}^{\ast}(\kappa) \equiv \left(w_1^{\ast}(\kappa), \ldots, w_K^{\ast}(\kappa) \right)', \ \  w_k^{\ast}(\kappa) \ \equiv \ \text{arg} \min_{w_k \in [0,1]} \left\{ \psi_k(w_k;\kappa) \right\} \ \text{for} \ k =1, \ldots, K, \label{optimal_shrinkage_factor}
\end{equation}
where $\psi_k(w_k;\kappa) \equiv s_k(w_k) \eta \left( (1-w_k) \cdot \kappa/s_k(w_k) \right)$ for $k = 1, \ldots, K$. If the right-hand side of (\ref{max_regret}) is a good approximation of the maximum regret, this rule is expected to be close to the minimax regret rule.

\begin{figure}[h] 
\centering
\includegraphics[width=12cm]{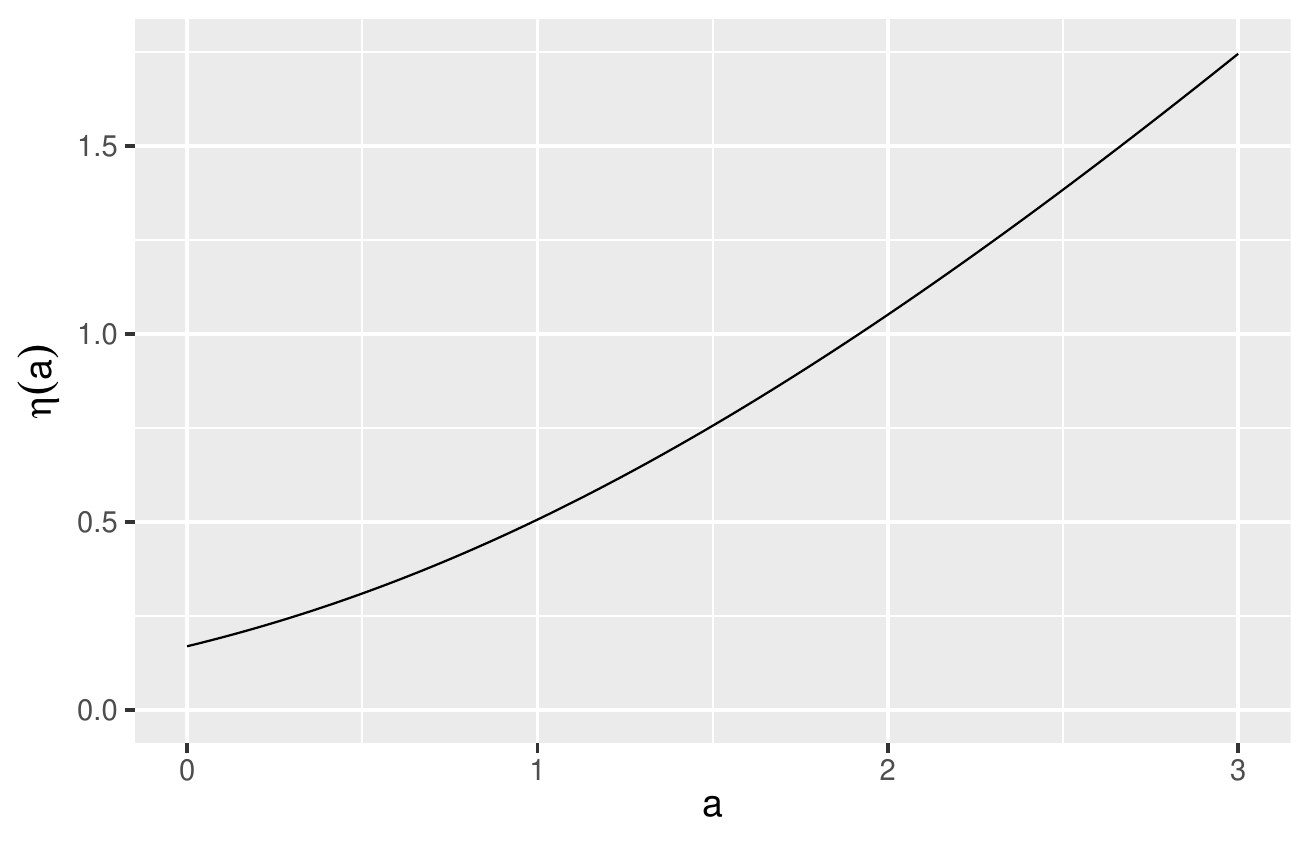}
\caption{Functional form of $\eta(a)$. The function $\eta(a)$ is strictly increasing and convex and $\eta(0)$ is approximately equal to 0.17.} \label{fig:eta}
\end{figure}

Our approach is attractive from a computational perspective. Indeed, the proposed shrinkage factors are easy to compute because (\ref{optimal_shrinkage_factor}) is obtained by optimizing the objective function over a single parameter while it is difficult to obtain the shrinkage factors that minimize the exact maximum regret when the number of possible covariate values is large. 

\begin{Remark}\label{rem:Randomized_rule}
Our results can be extended to the following randomized statistical treatment rules:
\begin{eqnarray*}
\hat{\delta}_{k}^{w_k, v_k}(\hat{\bm{\theta}}) \ \equiv \ 1 \left\{ w_k \cdot \hat{\theta}_k + (1-w_k) \cdot \mathrm{ave}(\hat{\bm{\theta}}) + v_k Z_k \geq 0 \right\},
\end{eqnarray*}
where $Z_k \sim N(0,1)$ is independent of $\hat{\bm{\theta}}$ and $v_k \geq 0$ is the randomization factor. Then, conditional on $\hat{\bm{\theta}}$, we obtain
\[
\hat{\delta}_{k}^{w_k, v_k}(\hat{\bm{\theta}}) \ = \ \begin{cases}
    1, & \text{with probability $\Phi_{v_k}\left( w_k \cdot \hat{\theta}_k + (1-w_k) \cdot \mathrm{ave}(\hat{\bm{\theta}}) \right)$} \\
    0, & \text{with probability $1 - \Phi_{v_k}\left( w_k \cdot \hat{\theta}_k + (1-w_k) \cdot \mathrm{ave}(\hat{\bm{\theta}}) \right)$} 
\end{cases},
\]
where $\Phi_{v}$ denotes the distribution function of $N(0,v^2)$. When $v_k = 0$, this rule becomes a non-randomized rule.

Let $\bm{v} \equiv (v_1, \ldots , v_K)'$ and $\hat{\bm{\delta}}^{\bm{w}, \bm{v}}(\hat{\bm{\theta}}) \equiv \left( \hat{\delta}_{1}^{w_1, v_1}(\hat{\bm{\theta}}), \ldots, \hat{\delta}_{K}^{w_K, v_K}(\hat{\bm{\theta}})  \right)'$. Because we have
\begin{eqnarray*}
& w_k \cdot \hat{\theta}_k + (1-w_k) \cdot \mathrm{ave}(\hat{\bm{\theta}}) + v_k Z_k \ \sim \ N \left( w_k \cdot \theta_k + (1-w_k) \cdot \overline{\theta}, s_k^2(w_k) + v_k^2 \right),
\end{eqnarray*}
we obtain the following upper bound of the maximum regret:
\begin{equation}
\max_{\bm{\theta} \in \Theta(\kappa)} R(\bm{\theta}, \hat{\bm{\delta}}^{\bm{w},\bm{v}}) \ \leq \ \sum_{k=1}^K p_k \cdot \sqrt{s_k^2(w_k)+v_k^2} \cdot \eta \left( \frac{(1-w_k) \cdot \kappa}{\sqrt{s_k^2(w_k)+v_k^2}} \right). \nonumber
\end{equation}
Similar to the non-randomized shrinkage rule, the shrinkage and randomization factors can be easily obtained. \cite{olea2023decision} and \cite{yata2021optimal} show that the minimax regret rule can be a randomized rule under partial identification. However, in our setting, $\bm{\theta}$ is point-identified. Hence, in this study, we focus on non-randomized treatment rules.
\end{Remark}

\begin{Remark}\label{rem:EWM}
\cite{kitagawa2018should} propose the empirical welfare maximization (EWM) method in a similar setting. Letting $\delta(X_i) \equiv \sum_{k=1}^K \delta_k \cdot 1\{X_i = x_k\}$ and $\bm{\delta} \equiv (\delta_1, \ldots, \delta_K)' \in \{0,1\}^{K}$, then the EWM rule in our setting is obtained by solving the following problem:
\begin{equation}
\max_{\bm{\delta} \in \Pi} \, \sum_{i=1}^n \left\{ \frac{Y_i D_i}{\hat{e}(X_i)} -  \frac{Y_i (1-D_i)}{1-\hat{e}(X_i)} \right\} \cdot \delta(X_i), \label{EWM_object}
\end{equation}
where $\hat{e}(x)$ is an estimator of the propensity score $P(D_i=1|X_i=x)$ and $\Pi \subset \{0,1\}^K$ is the class of candidate treatment rules. If we set $\hat{e}(x_k) = \frac{n_{1,k}}{n_{0,k}+n_{1,k}}$ and $p_k = \frac{n_{0,k}+n_{1,k}}{n}$, then the problem (\ref{EWM_object}) can be written as
\[
\max_{\bm{\delta} \in \Pi} \, \sum_{k=1}^K p_k \hat{\theta}_k \cdot \delta_k.
\]
Therefore, if there are no restrictions on the class of candidate treatment rules, that is, $\Pi = \{0,1\}^K$ holds, then the EWM rule is equivalent to the CES rule. In contrast, if there are some restrictions on $\Pi$, the EWM rule diverges from the CES rule.

\cite{kitagawa2018should} consider a more general setting, and their EWM approach imposes fewer restrictions than ours. In particular, when $X_i$ is continuous, the EWM approach remains applicable, whereas our method cannot be directly applied. However, if the policy maker restricts treatment rules to rely solely on predetermined categories of $X_i$, our method can be applied by forming subgroups according to these categories.
\end{Remark}

\begin{Remark}\label{rem:choice_kappa}
The choice of $\kappa$ is important in practice. However, it is difficult to obtain a theoretically justifiable data-driven way to select $\kappa$. Existing studies discuss this problem. For example, \cite{armstrong2018optimal,armstrong2021finite} consider the minimax estimation and inference problem for treatment effects and show that it is not possible to choose the parameter space automatically in a data-driven manner. Hence, this paper does not consider the adaptive method for unknown $\kappa$. In Section \ref{SubSec:misspecified}, we evaluate the proposed method in a situation where $\kappa$ is misidentified but close to the true value.

While it is challenging to develop a theoretically justified data-driven method to select $\kappa$, given a specific value of $\kappa$, we can formulate a testing procedure for the following null and alternative hypotheses:
\[
H_0: \bm{\theta} \in \Theta(\kappa) \ \ \ \text{vs} \ \ \ H_1: \bm{\theta} \not\in \Theta(\kappa).
\]
Since $\kappa < \kappa'$ implies $\Theta(\kappa) \subset \Theta(\kappa')$, this test allows us to assess whether the given value of $\kappa$ is too small, but it does not allow us to determine whether $\kappa$ is too large. If $\bm{\theta} \in \Theta(\kappa)$, then the inequality $\max_{1 \leq k \leq K} |\theta_k - \overline{\theta}| \leq \kappa$ holds. Hence, if $\kappa$ is small, then the quantity $\max_{1 \leq k \leq K} \left| \hat{\theta}_k - \mathrm{ave}(\hat{\bm{\theta}}) \right|$ is expected to be sufficiently small as well. Moreover, under the assumption that $\bm{\theta} \in \Theta(\kappa)$, the following bound can be derived:
$$
\left| \hat{\theta}_k - \mathrm{ave}(\hat{\bm{\theta}}) \right| \ \leq \ \kappa + \left| (\hat{\theta}_k-\theta_k) - (\mathrm{ave}(\hat{\bm{\theta}}) -\overline{\theta} )\right|,
$$
where $\left| (\hat{\theta}_k-\theta_k) - (\mathrm{ave}(\hat{\bm{\theta}}) -\overline{\theta} )\right|$ has the same distribution as $\left| Z_k - \overline{Z} \right|$ with $Z_k \sim N(0,\sigma_k^2)$. This leads to the following testing procedure:
\[
\max_{1 \leq k \leq K} \left| \hat{\theta}_k - \mathrm{ave}(\hat{\bm{\theta}}) \right| \ > \ \kappa + c_0 \ \ \ \Rightarrow \ \ \ \text{reject $H_0$},
\]
where $c_0$ is the $1-\alpha$ quantile of the distribution of $\max_{1 \leq k \leq K} \left| Z_k - \overline{Z} \right|$ and $\alpha$ is a prespecified significance level. Therefore, if a realized value of $\max_{1 \leq k \leq K} \left| \hat{\theta}_k - \mathrm{ave}(\hat{\bm{\theta}}) \right|$ is less than $c_0$, all candidate values of $\kappa$ are considered acceptable by construction.
\end{Remark}


\subsection{Comparison between the true maximum regret and the upper bound}\label{SubSec:tightness}
In this section, we illustrate the difference between the true maximum regret and the upper bound proposed in the previous section. For illustration, we consider the simple case where $\sigma_1 = \cdots = \sigma_K$ and $p_1 = \cdots = p_K$. In this case, the regret function $R(\bm{\theta},\hat{\bm{\delta}}^{\bm{w}})$ does not change when $(\theta_k,w_k)$ is replaced with $(\theta_j, w_j)$. This implies that a function $\bm{w} \mapsto \max_{\bm{\theta} \in \Theta(\kappa)} R(\bm{\theta},\hat{\bm{\delta}}^{\bm{w}})$ is permutation invariant because $\Theta(\kappa)$ is permutation invariant. In this setting, our upper bound is also permutation invariant and the proposed shrinkage rule satisfies $w_1^{\ast}(\kappa) = \cdots = w_K^{\ast}(\kappa)$. Hence, we focus on the shrinkage rules $\hat{\delta}^{\bm{w}}$ with $w_1 = \cdots = w_K$. Although we do not know that the minimax shrinkage rule satisfies $w_1 = \cdots = w_K$, it is natural to consider the shrinkage rules with $w_1 = \cdots = w_K$ because the maximum regret is permutation invariant.
If we consider the shrinkage rule $\hat{\delta}^{(w,\ldots,w)}$, the upper bound (\ref{max_regret}) can be written as
\[
\overline{R}_{\mathrm{upper}}(w) \ \equiv \ \sum_{k=1}^K p_k \cdot \psi_k(w;\kappa) \ = \ \psi(w;\kappa), \ \ \ \text{for $w \in [0,1]$.}
\]
where $\psi(w;\kappa) \equiv s(w) \eta \left( (1-w)\kappa / s(w) \right)$ and $s(w) \equiv s_1(w) = \cdots = s_K(w)$. Similarly, we define $\overline{R}_{\mathrm{true}}(w) \equiv \max_{\bm{\theta} \in \Theta(\kappa)} R(\bm{\theta},\hat{\delta}^{(w,\ldots,w)})$ for $w \in [0,1]$. In the following, we compare the upper bound $\overline{R}_{\mathrm{upper}}(w)$ and the true maximum regret $\overline{R}_{\mathrm{true}}(w)$ both analytically and numerically.

In the following proposition, we show that the upper bound $\overline{R}_{\mathrm{upper}}(w)$ is less than almost twice the true maximum regret $\overline{R}_{\mathrm{true}}(w)$.

\begin{Proposition}\label{prop:tightness}
Suppose that $\sigma_1 = \cdots = \sigma_K$ and $p_1 = \cdots = p_K$. For all $w \in [0,1]$, we obtain
\begin{equation}
1 \ \leq \ \frac{\overline{R}_{\mathrm{upper}}(w)}{\overline{R}_{\mathrm{true}}(w)} \ \leq \ \frac{\overline{R}_{\mathrm{upper}}(w)}{L_{\mathrm{true}}(w)} \ \leq \ \frac{K}{\lfloor K/2 \rfloor}, \label{tight_upper}
\end{equation}
where $L_{\mathrm{true}}(w)$ is defined in Lemma \ref{lem:prop1} of Appendix A. Furthermore, $L_{\mathrm{true}}(w) = \overline{R}_{\mathrm{upper}}(w)$ holds when $\kappa = 0$, that is, $\overline{R}_{\mathrm{upper}}(w)$ is equal to $\overline{R}_{\mathrm{true}}(w)$ when $\kappa = 0$.
\end{Proposition}

Proposition \ref{prop:tightness} shows that $\overline{R}_{\mathrm{true}}(w) \leq \overline{R}_{\mathrm{upper}}(w) \leq 2 \overline{R}_{\mathrm{true}}(w)$ holds for all $w \in [0,1]$ when $K$ is even. In addition, we can obtain a tighter upper bound that depends on $w$ and $\kappa$. Although the upper bound $\frac{\overline{R}_{\mathrm{upper}}(w)}{L_{\mathrm{true}}(w)}$ is complex and may be difficult to interpret, we can show that it is equal to $1$ when $\kappa=0$. Hence, these results imply that our bound is relatively tight compared to the true maximum regret, especially when $\kappa$ is close to zero. As demonstrated later, the numerical evaluation also shows that the functional form of $\overline{R}_{\mathrm{upper}}(w)$ is similar to that of $\overline{R}_{\mathrm{true}}(w)$ and $\overline{R}_{\mathrm{upper}}(w)$ is exactly the same as $\overline{R}_{\mathrm{true}}(w)$ when $\kappa = 0$.

To obtain Proposition \ref{prop:tightness}, we need to derive a lower bound of $\overline{R}_{\mathrm{true}}(w)$. When $K$ is even, we have
\begin{eqnarray*}
\overline{R}_{\mathrm{true}}(w) \ = \ \max_{\bm{\theta} \in \Theta(\kappa)} R(\bm{\theta},\hat{\bm{\delta}}^{(w,\ldots,w)}) & \geq & R(\bm{\theta}_{t,\kappa},\hat{\bm{\delta}}^{(w,\ldots,w)}) \ \ \ \text{for any $t \in \mathbb{R}$,}
\end{eqnarray*}
where
$$
\bm{\theta}_{t,\kappa} \ \equiv \ (\underbrace{t+\kappa, \ldots, t+\kappa}_{\text{$K/2$ elements}}, \underbrace{t-\kappa, \ldots, t-\kappa}_{\text{$K/2$ elements}})'.
$$
Using this lower bound, we obtain the results of Proposition \ref{prop:tightness}. We expect that this parameterization provides a tighter lower bound because the difference between each element of $\bm{\theta}_{t,\kappa}$ and the average becomes $\pm \kappa$. In fact, in Section \ref{SubSec:choice}, we obtain the upper bound (\ref{upper_bound_regret}) by replacing $\theta_k - \overline{\theta}$ with $\pm \kappa$.

Figure \ref{fig:true_K20} shows the functional forms of $\overline{R}_{\mathrm{true}}(w)$ and $\overline{R}_{\mathrm{upper}}(w)$ for $K=20$ and $\kappa = 0, \, 0.25, \, 0.5, \, 0.75$.\footnote{We calculate $\overline{R}_{\mathrm{true}}(w)$ using the Monte Carlo approximation. We generate $\bm{\theta}_1, \ldots, \bm{\theta}_{n_{\mathrm{sim}}}$ from a distribution on $\Theta(\kappa)$ and approximate $\overline{R}_{\mathrm{true}}(w)$ as $\max_{i} R(\bm{\theta}_i,\hat{\delta}^{(w,\ldots,w)})$. Hence, $\overline{R}_{\mathrm{true}}(w)$ in Figure \ref{fig:true_K20} may be less than the true maximum regret function.} The solid and dashed lines denote $\overline{R}_{\mathrm{true}}(w)$ and $\overline{R}_{\mathrm{upper}}(w)$, respectively. In all settings, $\overline{R}_{\mathrm{true}}(w)$ and $\overline{R}_{\mathrm{upper}}(w)$ have similar functional forms and $\overline{R}_{\mathrm{true}}(1)=\overline{R}_{\mathrm{upper}}(1)$ holds. In particular, Figure \ref{fig:true_K20_kappa0} shows that the upper bound $\overline{R}_{\mathrm{upper}}(w)$ is exactly equal to the true maximum regret $\overline{R}_{\mathrm{true}}(w)$ when $\kappa=0$.

\begin{figure}[htbp]
    \begin{tabular}{cc}
      \begin{minipage}[t]{0.45\hsize}
        \centering
        \includegraphics[width=7cm]{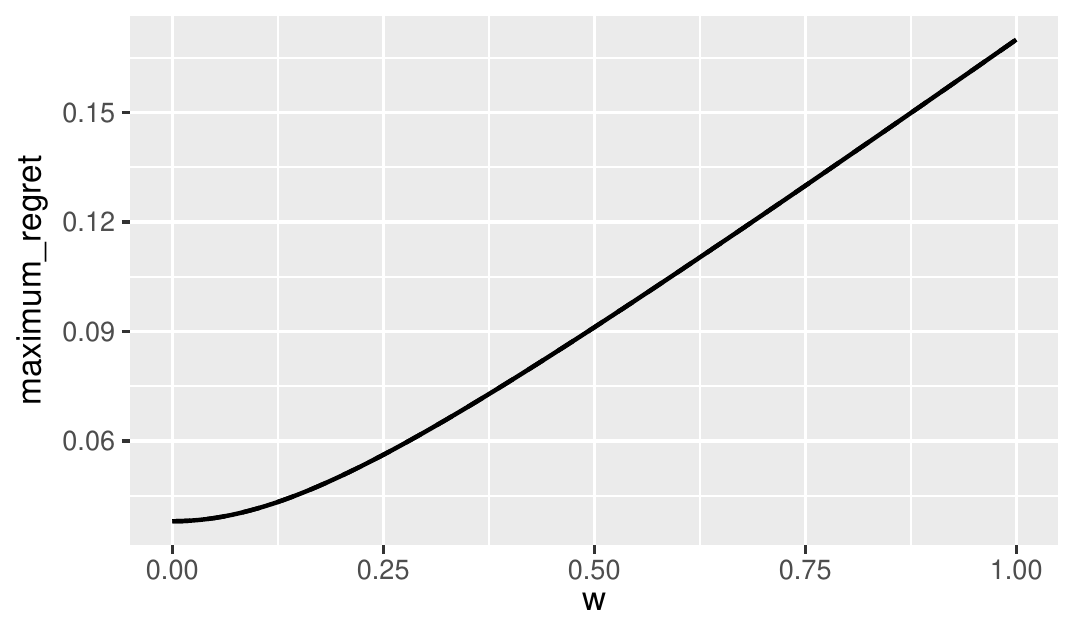}
        \subcaption{$K=20$ and $\kappa = 0$. The minimum values of $\overline{R}_{\mathrm{true}}(w)$ and $\overline{R}_{\mathrm{upper}}(w)$ are $0.038$ and $0.038$, respectively.}
        \label{fig:true_K20_kappa0}
      \end{minipage} &
      \begin{minipage}[t]{0.45\hsize}
        \centering
        \includegraphics[width=7cm]{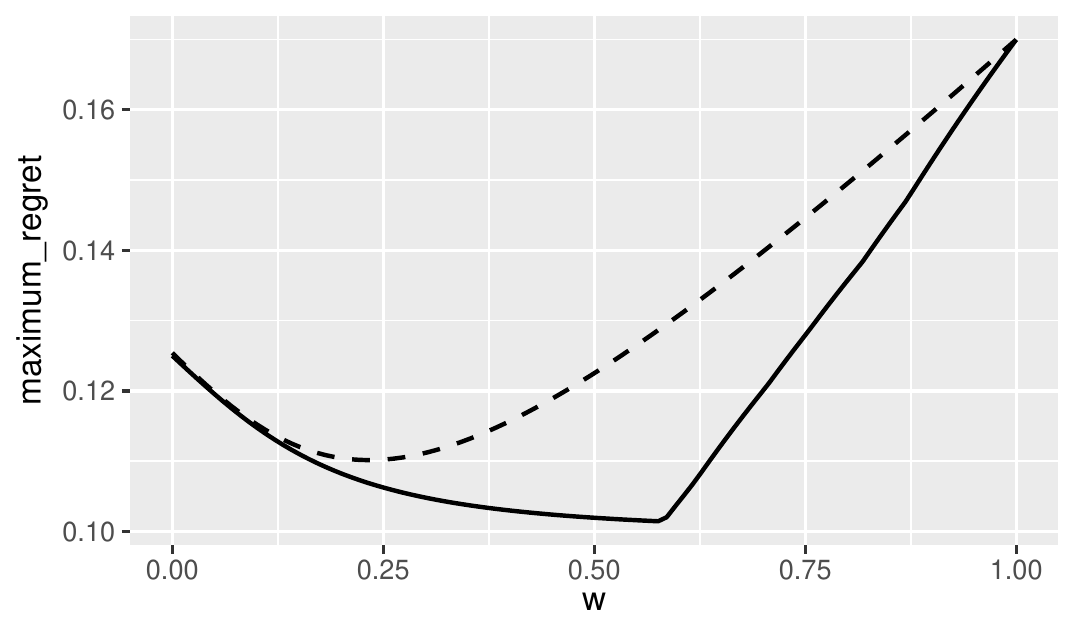}
        \subcaption{$K=20$ and $\kappa = 0.25$. The minimum values of $\overline{R}_{\mathrm{true}}(w)$ and $\overline{R}_{\mathrm{upper}}(w)$ are $0.101$ and $0.110$, respectively.}
        \label{fig:true_K20_kappa025}
      \end{minipage} \\
   
      \begin{minipage}[t]{0.45\hsize}
        \centering
        \includegraphics[width=7cm]{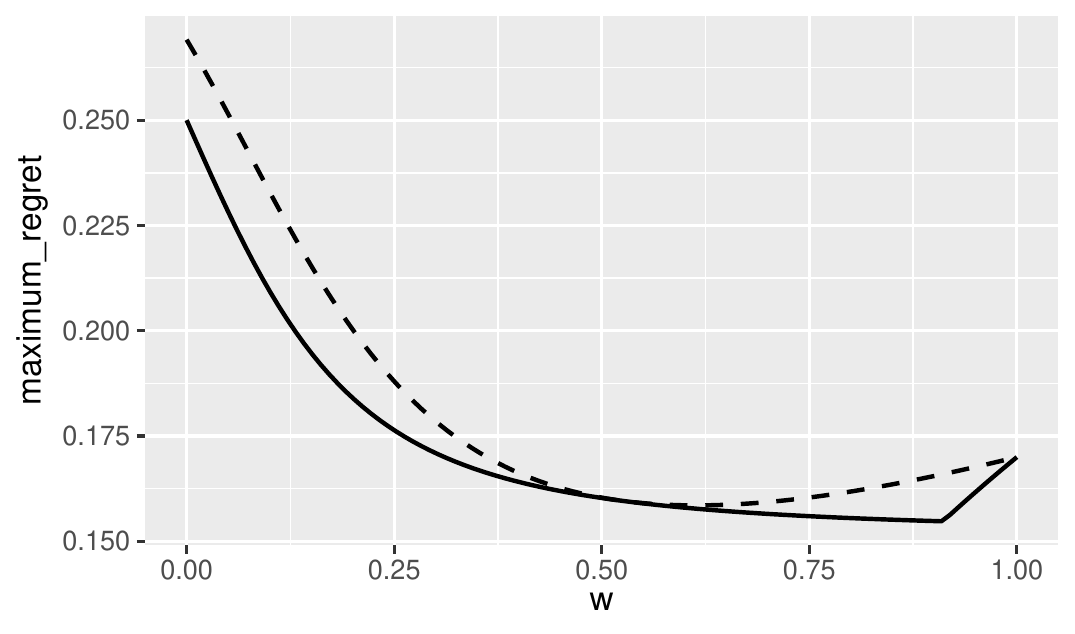}
        \subcaption{$K=20$ and $\kappa = 0.5$. The minimum values of $\overline{R}_{\mathrm{true}}(w)$ and $\overline{R}_{\mathrm{upper}}(w)$ are $0.155$ and $0.159$, respectively.}
        \label{fig:true_K20_kappa050}
      \end{minipage} &
      \begin{minipage}[t]{0.45\hsize}
        \centering
        \includegraphics[width=7cm]{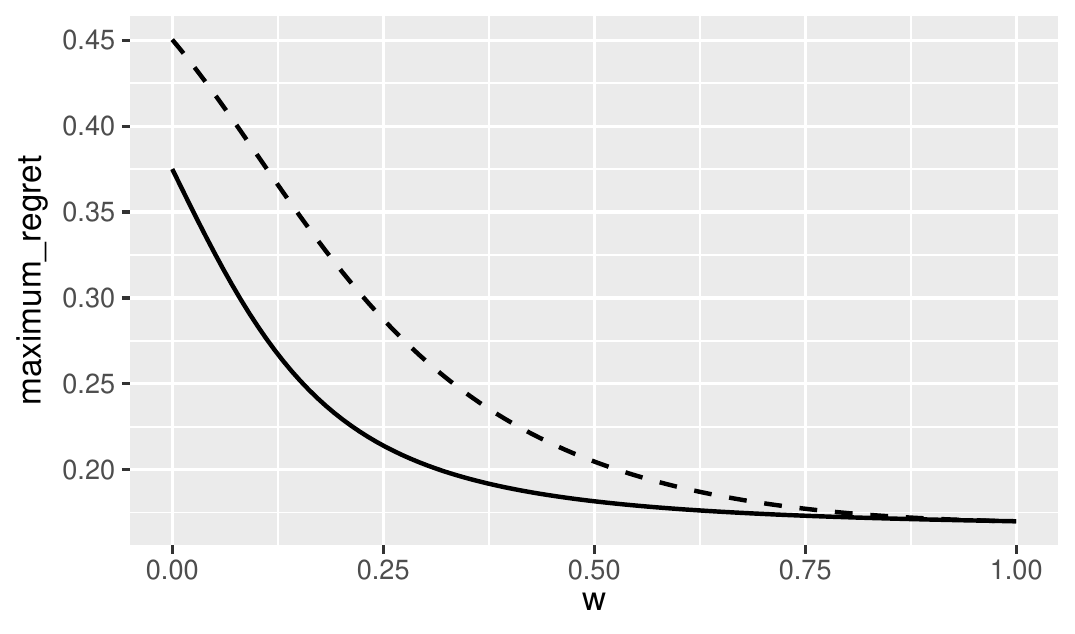}
        \subcaption{$K=20$ and $\kappa = 0.75$. The minimum values of $\overline{R}_{\mathrm{true}}(w)$ and $\overline{R}_{\mathrm{upper}}(w)$ are $0.170$ and $0.170$, respectively.}
        \label{fig:true_K20_kappa075}
      \end{minipage} 
    \end{tabular}
    \caption{The solid and dashed lines denote $\overline{R}_{\mathrm{true}}(w)$ and $\overline{R}_{\mathrm{upper}}(w)$, respectively.}\label{fig:true_K20}
\end{figure}

Figure \ref{fig:ratio_K20} shows the ratio of $\overline{R}_{\mathrm{upper}}(w)$ to $\overline{R}_{\mathrm{true}}(w)$ for $K=20$ and $\kappa = 0, \, 0.25, \, 0.5, \, 0.75$. In all cases, the ratio $\overline{R}_{\mathrm{upper}}(w) / \overline{R}_{\mathrm{true}}(w)$ is less than 2, and these results are consistent with Proposition \ref{prop:tightness}. Specifically, even in the worst case, the ratio does not exceed 1.4 under our settings. We also find that the ratio tends to increase when the shrinkage factor $w$ falls within a certain range that varies depending on $\kappa$, however, Figure \ref{fig:true_K20} shows that the minimum value of $\overline{R}_{\mathrm{upper}}(w)$ is close to that of $\overline{R}_{\mathrm{true}}(w)$ for all $\kappa$.
These results imply that the optimal rule based on the upper bound achieves near-optimal maximum regret. In Appendix B, we provide additional results on $\overline{R}_{\mathrm{true}}(w)$, $\overline{R}_{\mathrm{upper}}(w)$, and $\overline{R}_{\mathrm{upper}}(w) / \overline{R}_{\mathrm{true}}(w)$ for $K=4, \, 100$.

\begin{figure}[h] 
\centering
\includegraphics[width=12cm]{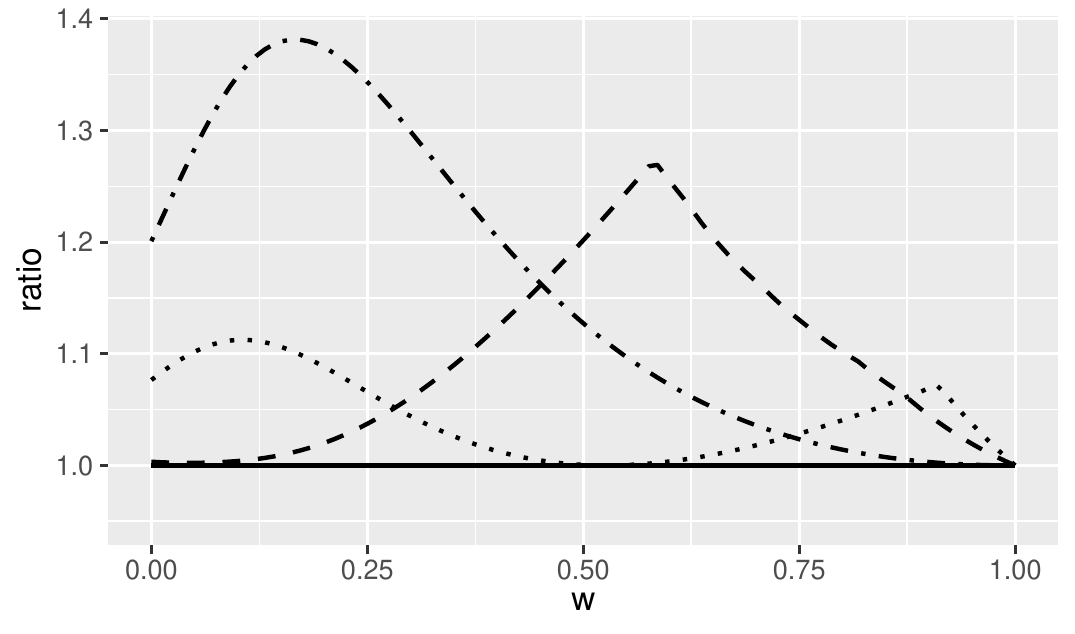}
\caption{The solid, dashed, dotted, and dot-dashed lines denote $\overline{R}_{\mathrm{upper}}(w) / \overline{R}_{\mathrm{true}}(w)$ when $K=20$ and $\kappa = 0, \, 0.25, \, 0.5$, and $0.75$.} \label{fig:ratio_K20}
\end{figure}

In this section, we focus on the balanced case where $\sigma_1 = \cdots = \sigma_K$ and $p_1 = \cdots = p_K$. Although this setting is unrealistic, it is an unfavorable setting for our shrinkage method. When constructing the upper bound (\ref{upper_bound_regret}), we replace $\text{sgn}(\theta_k) (\theta_k - \overline{\theta})$ with its maximum value $\kappa$. Hence, our upper bound is equivalent to
\[
\sum_{k=1}^K p_k \cdot \max_{\bm{\theta} \in \Theta(\kappa)}  \left\{ |\theta_k| \cdot  \Phi \left( -   \frac{|\theta_k| - (1-w_k) \cdot \text{sgn}(\theta_k) (\theta_k - \overline{\theta})}{s_k(w_k)} \right) \right\}.
\]
Therefore, the difference between the true maximum regret and the upper bound is small when a small number of groups account for a large proportion of the population. Especially, when $p_k = 1$ for some $k$, the upper bound is equal to the true maximum regret.

\subsection{Asymptotic behavior of $\bm{w}^{\ast}(\kappa)$}\label{SubSec:asymptotic}
In this section, we discuss the asymptotic behavior of the proposed shrinkage rule. We consider the following three asymptotic situations: (i) the dispersion of the CATEs becomes larger, that is, $\kappa \to \infty$, (ii) the dispersion of the CATEs becomes smaller, that is, $\kappa \to 0$, and (iii) the number of subgroups increases, that is, $K \rightarrow \infty$. As an example, suppose that a discrete covariate represents the region in which the experiment is conducted and we want to determine whether to treat individuals in each region. Since replacing $\sigma_1, \ldots, \sigma_K, \kappa$ with $c\sigma_1, \ldots, c\sigma_K, c\kappa$ does not change the value of $\psi_k(w_k;\kappa)$ for any $c > 0$, $\kappa \to \infty$ is equivalent to $\sigma_k \to 0$ for all $k$. Hence, situation (i) corresponds to the case where the sample size in each region goes to infinity but the dispersion of the region-specific treatment effects is fixed. Situation (ii) corresponds to the case where the dispersion of the region-specific treatment effects approaches zero. As discussed above, this situation is equivalent to the case where the standard errors $\sigma_1, \ldots, \sigma_K$ go to infinity. Situation (iii) corresponds to the case where the number of regions in which the experiment is conducted goes to infinity but the sample size in each region is fixed. In the shrinkage estimation literature, many studies focus on this type of asymptotic scenario.

First, we consider situation (i). Because $\eta(a)$ is convex, we have $\eta(a) \geq \eta(0) + \eta'(0) a$ for $a \geq 0$. This implies
\[
\psi_k(w_k;\kappa) \ \geq \ \eta(0) \cdot s_k(w_k) + (1-w_k) \cdot \eta'(0) \kappa,
\]
where $\eta'(0) \simeq 0.226$ and the equality holds when $w_k = 1$. Hence, if the right-hand side is minimized at $w_k = 1$, we obtain $w_k^{\ast}(\kappa)=1$. Because the derivative of the right-hand side becomes
\[
\eta(0) \cdot s_k'(w_k) - \eta'(0) \kappa,
\]
the shrinkage factor $w_k^{\ast}(\kappa)$ becomes one when $\eta(0) \cdot s_k'(w_k) \leq \eta'(0) \kappa$ holds for all $w_k \in [0,1]$. For $w_k \in [0,1]$, we obtain
\begin{eqnarray*}
s_k'(w_k) &=& \left\{ \sqrt{s_k^2(w_k)} \right\}' \ = \ \frac{1}{2} \frac{\{s_k^2(w_k)\}'}{s_k(w_k)} \\
&=& \frac{\left\{ w_k + \frac{1}{K}(1-2w_k) \right\}\sigma_k^2 - (1-w_k) \left( \frac{1}{K^2} \sum_{k=1}^K \sigma_k^2 \right) }{s_k(w_k)} \\
&\leq & \frac{\left\{ (1-2/K)w_k + 1/K \right\} \sigma_k^2}{\min_{w_k \in [0,1]} s_k(w_k)} \ \leq \ \frac{\sigma_k^2}{\min_{w_k \in [0,1]} s_k(w_k)}.
\end{eqnarray*}
As $s_k(w_k)$ is bounded away from zero, we obtain $w_k^{\ast}(\kappa)=1$ for a sufficiently large $\kappa$. Hence, the proposed shrinkage rule becomes the CES rule when $\kappa$ is sufficiently large.

Second, we consider situation (ii). As $\kappa \to 0$, we have that
$$
\psi_k(w_k;\kappa) \ \to \ s_k(w_k) \eta(0), \ \ \text{for $w_k \in [0,1]$.}
$$
Thus, the limit of $\psi_k(w_k; \kappa)$ is minimized at $w_k = \text{arg} \min_{w \in [0,1]}  s_k^2(w)$. Hence, if the homoscedasticity assumption holds, that is, $\sigma_k = \sigma$ for all $k$, then the limit of $\psi_k(w_k;\kappa)$ is minimized at $w_k=0$. Hence, if the dispersion of the parameters decreases, the proposed shrinkage rule approaches the pooling rule.

If we consider $w_k \cdot \hat{\theta}_k + (1-w_k) \cdot \mathrm{ave}(\hat{\bm{\theta}})$ to be an estimator of $\theta_k$, then this estimator becomes unbiased when $w_k = 1$. Hence, these two asymptotic situations imply that the shrinkage factor $w_k^{\ast}(\kappa)$ chooses a less biased estimator when $\kappa$ is large and a small variance estimator when $\kappa$ is small. This result can be seen as a type of bias-variance trade-off.

Finally, we consider situation (iii). We assume $\frac{1}{K^2} \sum_{k=1}^K \sigma_k^2  \to  0$ as $K \to \infty$. Under this condition, $s_k(w_k)$ can be approximated as $w_k \sigma_k$. Hence, in this situation, we have
$$
\psi_k(w_k;\kappa) \ \to \ \tilde{\psi}_k(w_k;\kappa) \ \equiv \ \sigma_k w_k \eta \left( (w_k^{-1} - 1) \cdot (\kappa / \sigma_k) \right).
$$
By letting $\tilde{w}_k^{\ast}(\kappa) \equiv \text{arg} \min_{w_k \in [0,1]} \tilde{\psi}_k(w_k;\kappa)$, $w_k^{\ast}(\kappa)$ can be approximated by $\tilde{w}_k^{\ast}(\kappa)$ when $K$ is large. Hence, when $K$ is sufficiently large, $\tilde{w}_k^{\ast}(\kappa)$ is useful to understand the properties of $w_k^{\ast}(\kappa)$.

\begin{Proposition}\label{prop:K_asymptotics}
Let $t^{\ast}(a) \equiv \mathrm{arg} \max_{t \geq 0} \left\{ t \Phi(-t+a) \right\}$. We have $\tilde{w}_k^{\ast}(\kappa) = 1$ when  $\kappa/\sigma_k > t^{\ast}(0) \simeq 0.75$ and $\tilde{w}_k^{\ast}(\kappa) = 0$ when $\kappa = 0$.
\end{Proposition}

Proposition \ref{prop:K_asymptotics} implies that we obtain results similar to above two situations even when $K$ is large. Specifically, the proposed shrinkage rule becomes the CES rule when $\kappa / \sigma_k$ is larger than approximately $3/4$ and $K$ is large, and the proposed shrinkage rule becomes the pooling rule when $\kappa = 0$ and $K$ is large. 

\begin{Remark}\label{rem:shrinkage_simple}
To illustrate the proposed shrinkage rule, we consider a simple case in which $\sigma_1 = \cdots = \sigma_K = \sigma$. Then, we observe that
\begin{eqnarray*}
\psi_k'(w_k;\kappa) &=& s_k'(w_k) \eta \left( \frac{(1-w_k) \cdot \kappa}{s_k(w_k)} \right) + s_k(w_k) \eta'\left( \frac{(1-w_k) \cdot \kappa}{s_k(w_k)} \right) \left\{ \frac{(1-w_k) \cdot \kappa}{s_k(w_k)} \right\}' \\
&=& s_k'(w_k) \eta \left( \frac{(1-w_k) \cdot \kappa}{s_k(w_k)} \right) - \left\{ \frac{\kappa s_k(w_k) + \kappa (1-w_k) s_k'(w_k)}{s_k(w_k)} \right\} \eta'\left( \frac{(1-w_k) \cdot \kappa}{s_k(w_k)} \right),
\end{eqnarray*}
where $s_k(w_k) = \sqrt{\{w_k^2 + 2w_k(1-w_k)/K \} \sigma^2 + (1-w_k^2)\sigma^2/K}$ and $s_k'(w_k) = \frac{\sigma^2 (1-1/K) w_k}{s_k(w_k)}$. This implies that
\begin{eqnarray*}
\psi_k'(1;\kappa) &=& \left( 1- \frac{1}{K} \right) \sigma \eta(0) - \kappa \eta'(0), \\
\psi_k'(0;\kappa) &=&  - \kappa \cdot \eta'\left( \frac{\kappa \sqrt{K}}{\sigma} \right).
\end{eqnarray*}
Therefore, in this setting, we have $\hat{\bm{\delta}}^{\bm{w}^{\ast}(\kappa)} \neq \hat{\bm{\delta}}^{\mathrm{CES}}$ when $\kappa < \frac{(1-1/K) \sigma \eta(0)}{\eta'(0)} \simeq 0.75 \left( 1-1/K \right) \sigma$, and $\hat{\bm{\delta}}^{\bm{w}^{\ast}(\kappa)} \neq \hat{\bm{\delta}}^{\mathrm{pool}}$ for any $\kappa > 0$.
\end{Remark}

\begin{Remark}\label{rem:shrinkage_est}
Our shrinkage method is based on the idea of defining a class of treatment rules by shrinking each estimator $\hat{\theta}_k$ toward $\mathrm{ave} (\hat{\bm{\theta}})$ to obtain a better treatment rule than existing unshrunk rules under the maximum regret criterion. The idea of using shrinkage rules is inspired by the construction of shrinkage estimators developed by \cite{JaSt61}. Here, we discuss the differences in the construction of the shrinkage rule when using our shrinkage factors versus James-Stein-type shrinkage factors. For simplicity, we focus on the case where $\sigma_1 = \cdots = \sigma_K = \sigma$ with $K > 3$. A James–Stein-type estimator that shrinks toward $\mathrm{ave} (\hat{\bm{\theta}})$ is given by
\begin{eqnarray}
    \hat{\bm{\theta}}^{\mathrm{JS}} & \equiv & \mathrm{ave} (\hat{\bm{\theta}}) + \hat{w}_{\mathrm{JS}} \cdot \left( \hat{\bm{\theta}} - \mathrm{ave} (\hat{\bm{\theta}}) \right) \ = \ \hat{w}_{\mathrm{JS}} \cdot \hat{\bm{\theta}} + (1-\hat{w}_{\mathrm{JS}}) \cdot \mathrm{ave} (\hat{\bm{\theta}}), \label{JS-estimator}
\end{eqnarray}
where
\[
\hat{w}_{\mathrm{JS}} \ \equiv \ 1 - \frac{(K-3) \sigma^2}{\sum_{k=1}^K \left( \hat{\theta}_k - \mathrm{ave}(\hat{\bm{\theta}}) \right)^2} \ = \ 1 - \frac{\sigma^2}{\frac{1}{K-3} \sum_{k=1}^K \left( \hat{\theta}_k - \mathrm{ave}(\hat{\bm{\theta}}) \right)^2}.
\]
When $K$ is large, $E \left[ \frac{1}{K-3} \sum_{k=1}^K \left( \hat{\theta}_k - \mathrm{ave}(\hat{\bm{\theta}}) \right)^2 \right]$ is approximated by $\sigma^2 + \frac{1}{K}\sum_{k=1}^K (\theta_k - \overline{\theta})^2$. Because $\bm{\theta} \in \Theta(\kappa)$ satisfies $\frac{1}{K}\sum_{k=1}^K (\theta_k - \overline{\theta})^2 \leq \kappa^2$, it is expected that $\hat{w}_{\mathrm{JS}} \leq \frac{\kappa^2}{\kappa^2 + \sigma^2}$ occurs with high probability when $K$ is large. Hence, similar to our shrinkage rule, the shrinkage factor $\hat{w}_{\mathrm{JS}}$ likely to take a value close to $0$ when $\kappa$ is small. On the contrary, unlike our shrinkage rule, $\hat{w}_{\mathrm{JS}}$ never takes on a value of exactly 1. The numerical simulation in Section \ref{Sec:example} shows that $\hat{w}_{\mathrm{JS}}$ tends to be smaller than $w_k^{\ast}(\kappa)$. This implies that the proposed shrinkage rule places greater emphasis on the bias than on the variance compared with the James–Stein-type estimator. We also note that the minimax regret criterion places more emphasis on the bias than the variance compared with the minimax MSE criterion \citep{ishihara2021evidence}.
\end{Remark}

\begin{Remark}\label{rem:continuous_case}
We analyze the asymptotic behavior of $w^{\ast}_k (\kappa)$ when the number of groups $K$ and the standard error $\sigma_k$ are related. If we have a continuous covariate $X_i$ and construct groups by discretizing $X_i$, the standard error of each group is expected to increase as the number of groups increases. For simplicity, we assume that the sample size of each group is $n/K$ and the standard error is inversely proportional to $n/K$, that is, $\sigma_k^2 = \frac{v_k^2}{n/K} = \frac{K v_k^2}{n}$ for some $v_k > 0$. Then, as discussed above, if the sample size grows faster than the number of groups, that is, $K/n \to 0$, we have $w_k^{\ast}(\kappa) = 1$. However, if $K/n \to c \in (0,1)$, $\psi_k(w_k;\kappa)$ can be approximated by
\[
\bar{\psi}_k(w_k;\kappa) \ \equiv \ \sqrt{c v_k^2} w_k \eta \left( (w_k^{-1} - 1) \cdot \left( \kappa /  \sqrt{c v_k^2} \right) \right).
\]

When applying our shrinkage approach to a continuous covariate, it is important to consider how to discretize $X_i$. As the number of groups increases, the standard error of each group is expected to increase. In contrast, this also allows for the implementation of more complex treatment rules. Hence, to determine an optimal discretization strategy, it is necessary to analyze this trade-off. However, in our setting, investigating this trade-off is challenging because the regret function is evaluated based on the optimal treatment given a fixed number of groups. This makes it difficult to accurately assess the improvement in optimal welfare that results from increasing $K$. Therefore, we do not pursue this extension further in this paper.
\end{Remark}

\subsection{Shrinkage toward a general location}\label{SubSec:reg}
In the previous sections, we consider the shrinkage rules that shrink toward the average of $\hat{\bm{\theta}}$. However, there are other options for shrinkage locations. In this section, we consider shrinkage rules that shrink toward other estimates such as a regression estimate, a weighted average estimate, and so on. Concretely, we consider the following shrinkage rules:
\begin{equation}
\tilde{\bm{\delta}}^{\bm{w}}(\hat{\bm{\theta}}) \ \equiv \ \left( \tilde{\delta}^{w_1}_{1}(\hat{\bm{\theta}}), \ldots, \tilde{\delta}^{w_K}_{K}(\hat{\bm{\theta}}) \right)', \label{shrinkage_rule_reg}
\end{equation}
where $\tilde{\delta}^{w_k}_{k}(\hat{\bm{\theta}}) \equiv 1 \left\{ w_k \cdot \hat{\theta}_k + (1-w_k) \cdot \hat{\xi}_k \right\}$, $\hat{\xi}_k \equiv \sum_{j=1}^K \omega_{k,j} \hat{\theta}_j$, and $(\omega_{k,1}, \ldots, \omega_{k,K})'$ is a known weight vector. For example, if $\hat{\xi}_k$ is the regression estimate of $\hat{\theta}_k$ on a vector of regression variables $z_k$, then $\hat{\xi}_k$ can be written as
\[
\hat{\xi}_k \ = \ z_k' \left( \bm{Z}'\bm{Z} \right)^{-1} \bm{Z}' \hat{\bm{\theta}} \ = \ \sum_{j=1}^K z_k' \left( \bm{Z}'\bm{Z} \right)^{-1} z_j \hat{\theta}_j,
\]
where $z_k$ can be different from $x_k$ and $\bm{Z} \equiv \left( z_1, \ldots, z_K \right)'$. If $X_i$ represents the region in which the experiment is conducted, then we can use region-specific characteristics as $z_k$. Similarly, we can also consider the shrinkage rule that shrinks toward the weighted average estimate $\hat{\xi}_k = \sum_{j=1}^K p_j \hat{\theta}_j$ or the weighted regression estimate $\hat{\xi}_k = z_k' \left( \sum_{j=1}^K p_j \cdot z_j z_j' \right)^{-1}  \sum_{j=1}^K p_j \cdot z_j \hat{\theta}_j$.

Instead of Assumption \ref{ass:parameter_space}, we assume that the difference between $\theta_k$ and the estimand of $\hat{\xi}$ is bounded by $\kappa$.

\begin{Assumption}\label{ass:parameter_space_reg}
For a positive constant $\kappa > 0$, the parameter $\bm{\theta}$ satisfies the following condition:
$$
\left| \theta_k - \xi_k \right| \ \leq \ \kappa, \ \ \ k = 1, \ldots, K,
$$
where $\xi_k \equiv \sum_{j=1}^K \omega_{k,j} \theta_j$.
\end{Assumption}

If we consider shrinkage rules that shrink toward the regression estimate, then we obtain
\[
\xi_k \ = \ z_k' \left( \bm{Z}'\bm{Z} \right)^{-1} \bm{Z}' \bm{\theta},
\]
where $\xi_k$ is interpreted as a linear projection of $\theta_k$ on $z_k$. Hence, Assumption \ref{ass:parameter_space_reg} implies that residuals from regressing $\bm{\theta}$ on $\bm{Z}$ are bounded by $\kappa$. Assumption \ref{ass:parameter_space_reg} generalizes Assumption \ref{ass:parameter_space} because we have $\xi_k = \sum_{j=1}^K \omega_{k,j} \theta_j = \overline{\bm{\theta}}$ when $\omega_{k,j} = 1/K$.

Similar to Section \ref{SubSec:choice}, we propose selecting shrinkage factors that minimize an upper bound of the maximum regret under Assumption \ref{ass:parameter_space_reg}. We observe that
\[
 w_k \cdot \hat{\theta}_k + (1-w_k) \cdot \hat{\xi}_k \ \sim \ N \left( w_k \cdot \theta_k + (1-w_k) \cdot \xi_k, \,\tilde{s}_k^2(w_k) \right),
\]
where $\tilde{s}_k^2(w_k)$ denotes the variance of $w_k \cdot \hat{\theta}_k + (1-w_k) \cdot \hat{\xi}_k$. Let $\Theta_{g}(\kappa)$ be the space of $\bm{\theta}$ satisfying Assumption \ref{ass:parameter_space_reg}. Then, as in (\ref{max_regret}), the maximum regret of the shrinkage rule (\ref{shrinkage_rule_reg}) is bounded by
$$
\max_{\bm{\theta} \in \Theta_{g}(\kappa)} R(\bm{\theta}, \tilde{\bm{\delta}}^{\bm{w}}) \ \leq \ \sum_{k=1}^K p_k \cdot \tilde{s}_k(w_k) \eta \left( \frac{(1-w_k) \cdot \kappa}{\tilde{s}_k(w_k)} \right).
$$
Hence, similar to $\bm{w}^{\ast}(\kappa)$, we propose the following shrinkage factors:
\begin{eqnarray*}
& & \bm{w}^{\ast}_{g}(\kappa) \equiv \left( w_{g,1}^{\ast}(\kappa), \ldots, w_{g,K}^{\ast}(\kappa) \right)', \\
& & \hspace{0.5in} \text{where} \ w_{g,k}^{\ast}(\kappa) \ \equiv \ \text{arg} \min_{w_k \in [0,1]} \left\{ \tilde{s}_k(w_k) \eta \left( \frac{(1-w_k) \cdot \kappa}{\tilde{s}_k(w_k)} \right) \right\} \ \text{for all $k$}.
\end{eqnarray*}

As discussed above, we can consider shrinkage rules that shrink toward a general location such as a regression estimate or a weighted average estimate. However, because analyzing the maximum regrets of such shrinkage rules is computationally extensive, we focus on the shrinkage rule proposed in Section \ref{SubSec:choice} in subsequent sections.

\section{Main results}\label{Sec:main}

The proposed shrinkage rule does not minimize the maximum regret because $\bm{w}^{\ast}(\kappa)$ minimizes an upper bound of the maximum regret. Hence, it remains unclear whether the maximum regret of $\hat{\bm{\delta}}^{\bm{w}^{\ast}(\kappa)}(\hat{\bm{\theta}})$ is smaller than that of $\hat{\bm{\delta}}^{\text{CES}}(\hat{\bm{\theta}})$ and $\hat{\bm{\delta}}^{\mathrm{pool}}(\hat{\bm{\theta}})$. This section compares the maximum regret of the proposed shrinkage rule with the CES and pooling rules when $\kappa$ is correctly specified or misspecified.

\subsection{Comparison with the CES and pooling rules when $\kappa$ is correctly specified}\label{SubSec:correct}
We compare the maximum regret of the proposed shrinkage rule with that of the CES and pooling rules when the true space of CATEs $\Theta(\kappa)$ is known. First, we compare the maximum regrets of $\hat{\bm{\delta}}^{\bm{w}^{\ast}(\kappa)}(\hat{\bm{\theta}})$ and $\hat{\bm{\delta}}^{\text{CES}}(\hat{\bm{\theta}})$.

\begin{Theorem}\label{thm:CES_correct}
Suppose that Assumption \ref{ass:parameter_space} holds. Let $\underline{\sigma} \equiv \min_k \{ \sigma_k \}$, $\overline{\sigma} \equiv \max_k \{ \sigma_k \}$, and $s_0 \equiv \sqrt{\frac{1}{K^2} \sum_{k=1}^K \sigma_k^2}$. If $t^{\ast}(0) \cdot (\overline{\sigma} - \underline{\sigma}) \leq \kappa$ or $\kappa \leq \eta(0) \cdot (\underline{\sigma} - s_0)$ holds, then we obtain
$$
\max_{\bm{\theta} \in \Theta(\kappa)} R(\bm{\theta},\hat{\bm{\delta}}^{\bm{w}^{\ast}(\kappa)}) \ \leq \ \max_{\bm{\theta} \in \Theta(\kappa)} R(\bm{\theta},\hat{\bm{\delta}}^{\mathrm{CES}}).
$$
In addition, if $t^{\ast}(0) \cdot (\overline{\sigma} - \underline{\sigma}) \leq \kappa < t^{\ast}(0) \cdot \left( 1-\frac{1}{K} \right) \overline{\sigma}$ or $\kappa < \eta(0) \cdot ( \underline{\sigma} - s_0 )$ holds, then we obtain
\begin{equation}
\max_{\bm{\theta} \in \Theta(\kappa)} R(\bm{\theta},\hat{\bm{\delta}}^{\bm{w}^{\ast}(\kappa)}) \ < \ \max_{\bm{\theta} \in \Theta(\kappa)} R(\bm{\theta},\hat{\bm{\delta}}^{\mathrm{CES}}). \label{ineq_CES_correct}
\end{equation}
\end{Theorem}

The first result of Theorem \ref{thm:CES_correct} implies that the maximum regret of $\hat{\bm{\delta}}^{\bm{w}^{\ast}(\kappa)}$ is less than or equal to that of $\hat{\bm{\delta}}^{\mathrm{CES}}$ when $\kappa \geq t^{\ast}(0) \cdot (\overline{\sigma} - \underline{\sigma}) \simeq 0.75 (\overline{\sigma} - \underline{\sigma})$ or $\kappa \leq  \eta(0) \cdot ( \underline{\sigma} - s_0 ) \simeq  0.17 ( \underline{\sigma} - s_0 )$. This implies that the proposed shrinkage rule outperforms the CES rule under the homoscedasticity assumption $\sigma_1 = \cdots = \sigma_K = \sigma$. Even when the homoscedasticity assumption does not hold, the maximum regret of $\hat{\bm{\delta}}^{\bm{w}^{\ast}(\kappa)}$ is less than or equal to that of $\hat{\bm{\delta}}^{\mathrm{CES}}$ for sufficiently small or large $\kappa$. For example, if $\overline{\sigma} = 1.5  \underline{\sigma}$, then we have
\[
\max_{\bm{\theta} \in \Theta(\kappa)} R(\bm{\theta},\hat{\bm{\delta}}^{\bm{w}^{\ast}(\kappa)}) \ \leq \ \max_{\bm{\theta} \in \Theta(\kappa)} R(\bm{\theta},\hat{\bm{\delta}}^{\mathrm{CES}}) \ \ \ \text{when $0.38 \underline{\sigma} \leq \kappa$ or $\kappa \leq 0.17 (\underline{\sigma} - s_0)$.}
\] 

The second result of Theorem \ref{thm:CES_correct} implies that the maximum regret of $\hat{\bm{\delta}}^{\bm{w}^{\ast}(\kappa)}$ is less than that of $\hat{\bm{\delta}}^{\mathrm{CES}}$. As discussed in Section \ref{SubSec:asymptotic}, the proposed shrinkage rule becomes the CES rule when $\kappa$ is sufficiently large. Because condition $\kappa < t^{\ast}(0) \cdot \left( 1-\frac{1}{K} \right) \overline{\sigma}$ implies $\hat{\bm{\delta}}^{\bm{w}^{\ast}(\kappa)} \neq \hat{\bm{\delta}}^{\mathrm{CES}}$, the proposed shrinkage rule has a smaller maximum regret than the CES rule when $t^{\ast}(0) \cdot (\overline{\sigma} - \underline{\sigma}) \leq \kappa < t^{\ast}(0) \cdot \left( 1-\frac{1}{K} \right) \overline{\sigma}$. Hence, under the homoscedasticity assumption, the proposed shrinkage rule has a smaller maximum regret than the CES rule when $\kappa < t^{\ast}(0) \cdot \left( 1-\frac{1}{K} \right) \sigma \simeq 0.75 \left( 1-\frac{1}{K} \right) \sigma$. If the homoscedasticity assumption does not hold, it is unclear whether (\ref{ineq_CES_correct}) holds when $\kappa$ is between $0.75(\overline{\sigma} - \underline{\sigma})$ and $0.17(\underline{\sigma} - s_0)$. However, numerical simulations in Section \ref{Sec:example} show that (\ref{ineq_CES_correct}) holds for all $\kappa$ below a certain value in all correctly specified settings.

Next, we compare the maximum regrets of $\hat{\bm{\delta}}^{\bm{w}^{\ast}(\kappa)}(\hat{\bm{\theta}})$ and $\hat{\bm{\delta}}^{\text{pool}}(\hat{\bm{\theta}})$.

\begin{Theorem}\label{thm:pool_correct}
Suppose that Assumption \ref{ass:parameter_space} holds and $K$ is even. Then, we obtain
\[
\begin{cases}
    \max_{\bm{\theta} \in \Theta(\kappa)} R(\bm{\theta},\hat{\bm{\delta}}^{\bm{w}^{\ast}(\kappa)}) \ \leq \ 2 \cdot \max_{\bm{\theta} \in \Theta(\kappa)} R(\bm{\theta},\hat{\bm{\delta}}^{\mathrm{pool}}) & \text{if $\kappa > 0$,} \\
    \max_{\bm{\theta} \in \Theta(\kappa)} R(\bm{\theta},\hat{\bm{\delta}}^{\bm{w}^{\ast}(\kappa)}) \ \leq \ \max_{\bm{\theta} \in \Theta(\kappa)} R(\bm{\theta},\hat{\bm{\delta}}^{\mathrm{pool}}) & \text{if $\kappa=0$.}
\end{cases}
\]
In addition, if $s_0 \eta(\kappa/s_0) > 2\eta(0) \left( \sum_{k=1}^K p_k \sigma_k \right)$ holds, then we obtain
\begin{equation}
\max_{\bm{\theta} \in \Theta(\kappa)} R(\bm{\theta},\hat{\bm{\delta}}^{\bm{w}^{\ast}(\kappa)}) \ < \ \max_{\bm{\theta} \in \Theta(\kappa)} R(\bm{\theta},\hat{\bm{\delta}}^{\mathrm{pool}}). \label{ineq_pool_correct}
\end{equation}
\end{Theorem}

The first result of Theorem \ref{thm:pool_correct} implies that the maximum regret of $\hat{\bm{\delta}}^{\bm{w}^{\ast}(\kappa)}$ is less than or equal to twice that of $\hat{\delta}^{\mathrm{pool}}$ for any $\kappa > 0$. In addition, the maximum regret of $\hat{\bm{\delta}}^{\bm{w}^{\ast}(\kappa)}$ is less than or equal to that of $\hat{\delta}^{\mathrm{pool}}$ when $\kappa = 0$. Hence, the proposed shrinkage rule is not so inferior to the pooling rule when $\kappa$ is positive. If the CATEs are the same across the groups, that is, $\kappa = 0$, the proposed rule is superior to the pooling rule.

Because $\eta(\cdot)$ is strictly increasing, the second result of Theorem \ref{thm:pool_correct} implies that the proposed shrinkage rule has a smaller maximum regret than the pooling rule when $\kappa$ is sufficiently large. If the homoscedasticity assumption holds, for any $\kappa > 0$, the condition $s_0 \eta(\kappa/s_0) > 2\eta(0) \left( \sum_{k=1}^K p_k \sigma_k \right)$ is equivalent to
\begin{equation}
\left\{ \frac{\eta\left( \sqrt{K} \cdot (\kappa / \sigma) \right)}{\sqrt{K} \cdot (\kappa / \sigma)} \right\} \cdot (\kappa / \sigma) \ > \ 2 \eta(0) \ \simeq \ 0.34, \label{condition_pool_correct}
\end{equation}
where a function $\eta(a)/a$ satisfies $\eta(a)/a \geq 1/2$ and $\lim_{a \to \infty} \eta(a)/a = 1$. Hence, for any $K$, (\ref{ineq_pool_correct}) holds when $\kappa / \sigma$ is larger than $4 \eta(0) \simeq 0.68$. Furthermore, if $K$ goes to infinity, the condition (\ref{condition_pool_correct}) becomes $\kappa / \sigma > 2 \eta(0) \simeq 0.34$.

Combining Theorems \ref{thm:CES_correct} and \ref{thm:pool_correct} yields conditions under which the proposed shrinkage rule is superior to both the CES and pooling rules.

\begin{Corollary}\label{cor:dominate}
Suppose that Assumption \ref{ass:parameter_space} holds and $K$ is even. If
\begin{equation}
\overline{\sigma} \ < \ K \underline{\sigma} \ \ \ \text{and} \ \ \  \frac{\sum_{k=1}^K p_k \sigma_k}{\overline{\sigma}} \ < \ \frac{t^{\ast}(0) \left( 1- \frac{1}{K} \right)}{4 \eta(0)}, \label{cor_condition}
\end{equation}
then both (\ref{ineq_CES_correct}) and (\ref{ineq_pool_correct}) hold when $\kappa$ satisfies
\begin{equation}
\max \left\{ 4 \eta(0) \left( \sum_{k=1}^K p_k \sigma_k \right), \, t^{\ast}(0) \cdot \left( \overline{\sigma} - \underline{\sigma} \right)  \right\} \ < \ \kappa \ < \ t^{\ast}(0) \cdot \left( 1 - \frac{1}{K} \right) \overline{\sigma}. \label{range_dominate}
\end{equation}
\end{Corollary}

Corollary \ref{cor:dominate} implies that if $K$ is large, then the proposed shrinkage rule has a smaller maximum regret than both the CES and pooling rules when $\kappa$ is within a certain range. When $K$ is large, the first condition of (\ref{cor_condition}) is satisfied as long as the dispersion of standard errors is not too large. Because $\frac{\sum_{k=1}^K p_k \sigma_k}{\overline{\sigma}} \leq 1$ and $\frac{t^{\ast}(0) }{4 \eta(0)} \simeq 1.106$, the second condition of (\ref{cor_condition}) is satisfied when $K \geq 12$. Hence, both (\ref{ineq_CES_correct}) and (\ref{ineq_pool_correct}) hold for some $\kappa$ when $K$ is sufficiently large. Under the homoscedasticity assumption, the range (\ref{range_dominate}) becomes
\[
0.68 \sigma \ < \ \kappa \ < \ 0.75 \left( 1 - \frac{1}{K} \right) \sigma.
\]
The following remark shows that when the homoscedasticity assumption holds, the proposed shrinkage rule has a smaller maximum regret than both the CES and pooling rules over a wider range than (\ref{range_dominate}).

\begin{Remark}\label{rem:dominate}
As discussed above, if $\sigma_1 = \cdots = \sigma_K = \sigma$ holds, the condition $s_0 \eta(\kappa/s_0) > 2\eta(0) \left( \sum_{k=1}^K p_k \sigma_k \right)$ is equivalent to
\begin{equation*}
\frac{\eta\left( \sqrt{K} \cdot (\kappa / \sigma) \right)}{\sqrt{K}}  \ > \ 2 \eta(0) \ \simeq \ 0.34.
\end{equation*}
Hence, from Theorems \ref{thm:CES_correct} and \ref{thm:pool_correct}, both (\ref{ineq_CES_correct}) and (\ref{ineq_pool_correct}) hold when $\kappa / \sigma$ satisfies
\begin{equation}
\frac{\eta^{-1} \left( 2 \eta(0) \sqrt{K} \right)}{\sqrt{K}} \ < \ \frac{\kappa}{\sigma} \ <\ t^{\ast}(0) \left( 1 - \frac{1}{K}  \right), \label{range_kappa}
\end{equation}
where $\eta^{-1}$ is the inverse function of $\eta$. Because we have $\lim_{a \to \infty} \eta(a)/a = 1$, we obtain
\[
\frac{\eta\left( \sqrt{K} \cdot (\kappa / \sigma) \right)}{\sqrt{K}} \ = \ \frac{\eta\left( \sqrt{K} \cdot (\kappa / \sigma) \right)}{\sqrt{K}\cdot (\kappa / \sigma)} \cdot (\kappa / \sigma) \ \to \ \kappa / \sigma \ \ \ \text{as $K \to \infty$.}
\]
Therefore, the range (\ref{range_kappa}) approaches $(0.34, 0.75)$ as $K$ goes to infinity.

Figure \ref{fig:kappa_range} shows the above range of $\kappa / \sigma$ (\ref{range_kappa}) as a function of $K$. The region below the dashed line represents the region in which the proposed shrinkage rule has a smaller maximum regret than the CES rule. The region above the dotted line represents the region in which the proposed shrinkage rule has a smaller maximum regret than the pooling rule. Hence, Figure \ref{fig:kappa_range} implies that no value of $\kappa / \sigma$ satisfies the condition (\ref{range_kappa}) when $K \leq 6$ and the range (\ref{range_kappa}) becomes wider as $K$ increases. For example, when $K=20, \, 50, \, 100$, the ranges (\ref{range_kappa}) are $(0.60, 0.71)$, $(0.54, 0.74)$, and $(0.50,0.74)$, respectively. Therefore, if $K$ is large, then both (\ref{ineq_CES_correct}) and (\ref{ineq_pool_correct}) hold when the dispersion of CATEs is moderate. In Section \ref{Sec:example}, numerical simulations show that the actual range in which both (\ref{ineq_CES_correct}) and (\ref{ineq_pool_correct}) hold is wider than the range (\ref{range_kappa}).

\begin{figure}[h] 
\centering
\includegraphics[width=12cm]{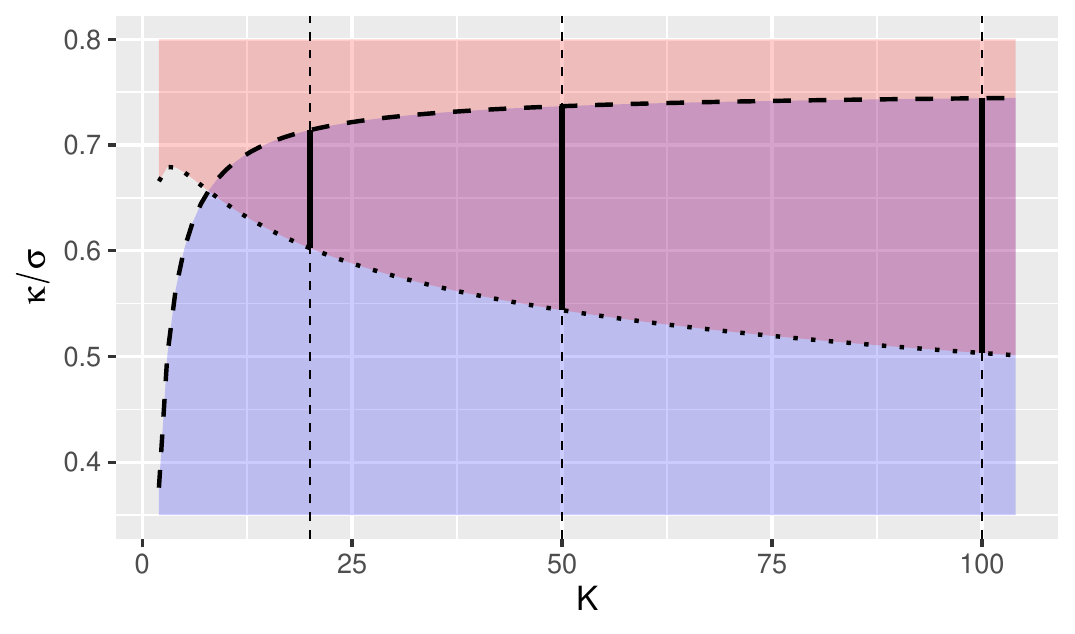}
\caption{The dashed and dotted lines denote $t^{\ast}(0) \left( 1 - \frac{1}{K} \right)$ and $K^{-1/2} \eta^{-1} \left( 2 \eta(0) \sqrt{K} \right)$, respectively. The solid lines denote the ranges (\ref{range_kappa}) for $K=20, \, 50, \, 100$. When $K=20, \, 50, \, 100$, the ranges (\ref{range_kappa}) are $(0.602, 0.714)$, $(0.544, 0.737)$, and $(0.503,0.744)$, respectively.} \label{fig:kappa_range}
\end{figure}
\end{Remark}

\begin{Remark}\label{rem:pool_correct}
In Theorem \ref{thm:pool_correct}, it is assumed that $K$ is even. This assumption is unnecessary to obtain the result similar to Theorem \ref{thm:pool_correct}. However, when $K$ is even, we can easily derive the lower bound of $\max_{\bm{\theta} \in \Theta(\kappa)}R(\bm{\theta},\hat{\bm{\delta}}^{\mathrm{pool}})$. Because we have $\bm{\theta}_t(\kappa) \equiv (t+\kappa, \ldots, t+\kappa, t-\kappa, \ldots, t-\kappa)' \in \Theta(\kappa)$ for all $t \in \mathbb{R}$ when $K$ is even, the following lower bound is obtained:
$$
\max_{\bm{\theta} \in \Theta(\kappa)}R(\bm{\theta},\hat{\bm{\delta}}^{\mathrm{pool}}) \ \geq \ \max_{t \in \mathbb{R}} R(\bm{\theta}_t(\kappa),\hat{\bm{\delta}}^{\mathrm{pool}}).
$$
The results of Theorem \ref{thm:pool_correct} are derived using this lower bound. Even if $K$ is odd, we can derive a similar lower bound because $(t+\kappa, \ldots, t+\kappa, t, t-\kappa, \ldots, t-\kappa)' \in \Theta(\kappa)$ holds for all $t \in \mathbb{R}$. 
\end{Remark}

\subsection{Comparison with the CES and pooling rules when $\kappa$ is misspecified}\label{SubSec:misspecified}
In the previous section, we assume that the space of CATEs $\Theta (\kappa)$ is known. However, in practice, it may be challenging to select a reasonable $\kappa$. In this section, we consider the case in which the researcher's choice of the space of CATEs $\Theta (\kappa')$ is different from the true space of CATEs $\Theta (\kappa)$, that is, $\kappa$ is misspecified.

First, we compare the maximum regret of the proposed shrinkage rule with that of the CES rule when $\kappa$ is misspecified.

\begin{Theorem}\label{thm:CES_mis}
Suppose that Assumption \ref{ass:parameter_space} holds. Let $c$ be a nonnegative constant. If $\kappa' = (1+c)\kappa$ holds, then it follows that
\begin{equation}
\max_{\bm{\theta} \in \Theta(\kappa)} R(\bm{\theta},\hat{\bm{\delta}}^{\bm{w}^{\ast}(\kappa')}) \ < \ \max_{\bm{\theta} \in \Theta(\kappa)} R(\bm{\theta},\hat{\bm{\delta}}^{\mathrm{CES}}) \label{ineq_CES_mis}
\end{equation}
when $t^{\ast}(0) \cdot (\overline{\sigma} - \underline{\sigma}) \leq \kappa < \frac{t^{\ast}(0) \cdot \left(1- \frac{1}{K}\right) \overline{\sigma}}{1+c}$ or $\kappa <  \frac{\eta(0) \cdot (\underline{\sigma} - s_0)}{1+c}$. If $\kappa = (1+c)\kappa'$, then (\ref{ineq_CES_mis}) holds when $\kappa < \eta(0) \cdot \left\{ \left( \frac{1+c}{1+2c}\right)  \underline{\sigma} - (1+c)s_0 \right\}$.
\end{Theorem}

In Theorem \ref{thm:CES_mis}, we use the parameter of CATEs $\Theta (\kappa')$, which may differ from the true space $\Theta(\kappa)$, to determine the shrinkage factors. The first result of Theorem \ref{thm:CES_mis} implies that if $\kappa \leq \kappa'$, the maximum regret of $\hat{\bm{\delta}}^{\bm{w}^{\ast}(\kappa')}$ is less than that of $\hat{\bm{\delta}}^{\mathrm{CES}}$ under conditions analogous to those in Theorem \ref{thm:CES_correct}. However, the range of $\kappa$ in which (\ref{ineq_CES_mis}) holds is narrower than the range of Theorem \ref{thm:CES_correct} due to the misspecification of $\kappa$. The second result of Theorem \ref{thm:CES_mis} implies that if $\kappa \geq \kappa'$, the maximum regret of $\hat{\bm{\delta}}^{\bm{w}^{\ast}(\kappa')}$ is less than or equal to that of $\hat{\bm{\delta}}^{\mathrm{CES}}$ when $\kappa$ is sufficiently small. According to the proof of Theorem \ref{thm:CES_mis}, if $\kappa \leq \kappa'$, the maximum regret of $\hat{\bm{\delta}}^{\bm{w}^{\ast}(\kappa')}$ is less than that of $\hat{\bm{\delta}}^{\mathrm{CES}}$ for all $\kappa$ under the homoscedasticity assumption. Conversely, if $\kappa > \kappa'$, the maximum regret of $\hat{\bm{\delta}}^{\bm{w}^{\ast}(\kappa')}$ can be larger than that of $\hat{\bm{\delta}}^{\mathrm{CES}}$ even when the homoscedasticity assumption holds. In fact, numerical simulations in Section \ref{Sec:example} show that the maximum regret of $\hat{\bm{\delta}}^{\bm{w}^{\ast}(\kappa')}$ is at most $1.7$ times less than that of $\hat{\bm{\delta}}^{\mathrm{CES}}$ for all $\kappa$ even when $\kappa$ is twice as large as $\kappa'$.

Next, we compare the maximum regrets of the proposed shrinkage rule with the pooling rule when $\kappa$ is misspecified.

\begin{Theorem}\label{thm:pool_mis}
Suppose that Assumption \ref{ass:parameter_space} holds and $K$ is even.  Let $c$ be a nonnegative constant. If $\kappa' = (1+c)\kappa$, then we obtain
\begin{equation}
\max_{\bm{\theta} \in \Theta(\kappa)} R(\bm{\theta},\hat{\bm{\delta}}^{\bm{w}^{\ast}(\kappa')}) \ < \ \max_{\bm{\theta} \in \Theta(\kappa)} R(\bm{\theta},\hat{\bm{\delta}}^{\mathrm{pool}}) \label{ineq_pool_mis}
\end{equation}
when $s_0 \eta(\kappa / s_0) > 2 \eta(0) \left( \sum_{k=1}^K p_k \sigma_k \right)$. If $\kappa = (1+c)\kappa'$, then (\ref{ineq_pool_mis}) holds when $s_0 \eta(\kappa / s_0) > 2 \eta(0) (1+2c) \left( \sum_{k=1}^K p_k \sigma_k \right)$.
In addition, if $\kappa' = (1+c)\kappa$ or $\kappa = (1+c)\kappa'$ holds, then we obtain
\begin{equation}
\max_{\bm{\theta} \in \Theta(\kappa)} R(\bm{\theta},\hat{\bm{\delta}}^{\bm{w}^{\ast}(\kappa')}) \ < \ 2(1+2c) \cdot \max_{\bm{\theta} \in \Theta(\kappa)} R(\bm{\theta},\hat{\bm{\delta}}^{\mathrm{pool}}) \ \ \ \text{for all $\kappa$.}\nonumber
\end{equation}
\end{Theorem}

The first result of Theorem \ref{thm:pool_mis} implies that if $\kappa \leq \kappa'$, the maximum regret of $\hat{\bm{\delta}}^{\bm{w}^{\ast}(\kappa')}$ is less than that of $\hat{\bm{\delta}}^{\mathrm{pool}}$ under the same condition as Theorem \ref{thm:pool_correct}. The second result of Theorem \ref{thm:pool_mis} also implies that (\ref{ineq_pool_mis}) holds when $\kappa$ is sufficiently large, but the range of $\kappa$ in which (\ref{ineq_pool_mis}) holds is narrower than the range of Theorem \ref{thm:pool_correct} due to the misspecification of $\kappa$. Similar to Theorem \ref{thm:pool_correct}, the third result of Theorem \ref{thm:pool_mis} implies that the proposed shrinkage rule is not so inferior to the pooling rule if the degree of misspecification is smaller.

If the homoscedasticity assumption holds and $\kappa' = (1+c) \kappa$, Theorems \ref{thm:CES_mis} and \ref{thm:pool_mis} imply that both (\ref{ineq_CES_mis}) and (\ref{ineq_pool_mis}) hold when $\kappa / \sigma$ satisfies
\[
\frac{\eta^{-1} \left( 2 \eta(0) \sqrt{K} \right)}{\sqrt{K}} \ < \ \frac{\kappa}{\sigma} \ < \ \frac{t^{\ast}(0) \left( 1 - \frac{1}{K}  \right)}{1+c}.
\]
Hence, although the above range is narrower than (\ref{range_kappa}), the proposed shrinkage rule has a smaller maximum regret than the CES and pooling rules for some $\kappa$ when $\kappa'$ is larger than $\kappa$. On the other hand, when $\kappa'$ is less than $\kappa$, Theorems \ref{thm:CES_mis} and \ref{thm:pool_mis} do not establish whether there exists $\kappa / \sigma$ such that both (\ref{ineq_CES_mis}) and (\ref{ineq_pool_mis}) hold. However, numerical simulations in Section \ref{Sec:example} show that both (\ref{ineq_CES_correct}) and (\ref{ineq_pool_correct}) hold for some $\kappa$ even when $\kappa'$ is half of $\kappa$.

This section provides sufficient conditions for the proposed shrinkage rule to have a smaller maximum regret than the CES and pooling rules. While these conditions are intuitive and easy to interpret, it is possible to improve these conditions. However, attempting to improve the result leads to sufficient conditions that are difficult to interpret. Therefore, in this section, we present less sharp but more interpretable conditions.

\section{Numerical examples}\label{Sec:example}

We present numerical examples to illustrate the results obtained in the previous sections. We demonstrate the relationship between $\kappa$ and $w_k^{\ast}(\kappa)$ under the homoscedasticity assumption. When homoscedasticity holds, that is, $\sigma_1 = \cdots = \sigma_K = 1$, we obtain $w_1^{\ast}(\kappa) = \cdots = w_K^{\ast}(\kappa) = w^{\ast}(\kappa)$. We calculate the shrinkage factor $w^{\ast}(\kappa)$ numerically for each $\kappa \in [0,1]$. Figure \ref{fig:w_opt} shows the relationship between $\kappa$ and $w^{\ast}(\kappa)$ for $K = 2, \, 5, \, 100$. The proposed shrinkage rule becomes the CES rule when $\kappa$ is sufficiently large and approaches the pooling rule when $\kappa$ approaches zero in all settings. This finding is consistent with the results presented in Section \ref{SubSec:asymptotic}. Furthermore, as seen in Proposition \ref{prop:K_asymptotics}, when the number of subgroups increases ($K=100$), the shrinkage factor becomes one if $\kappa$ is larger than $t^{\ast}(0) \simeq 0.752$.

\begin{figure}[h] 
\centering
\includegraphics[width=13cm]{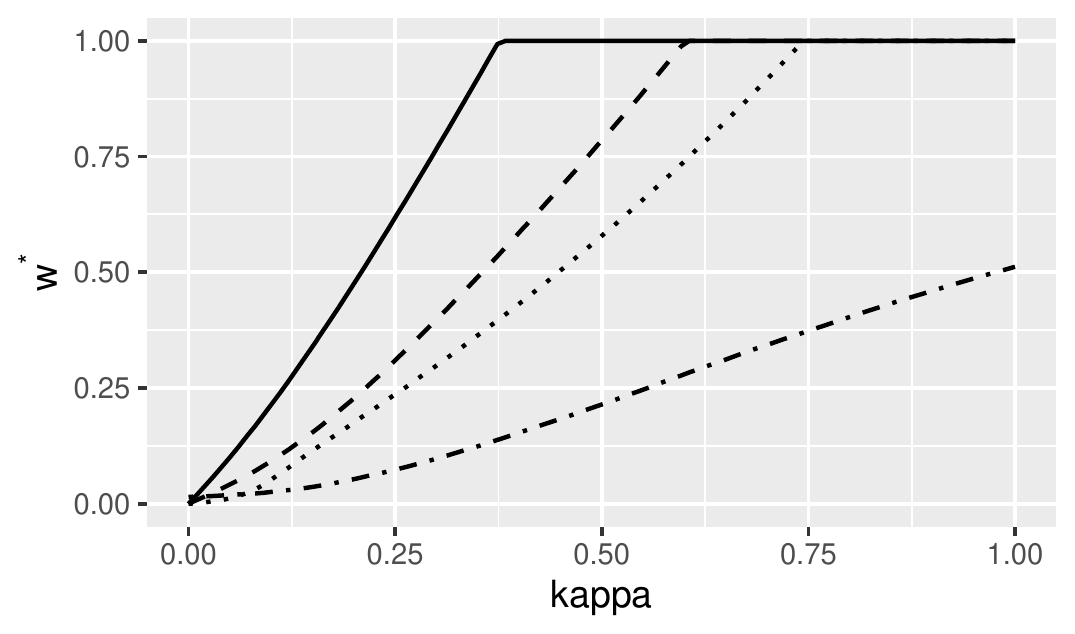}
\caption{The relationship between $\kappa$ and $w^{\ast}(\kappa)$ when $K=2, 5, 100$. The solid, dashed, dotted lines denote the shrinkage factors when $K=2, 5, 100$, respectively. The dot-dashed line denotes the median of $\hat{w}_{\mathrm{JS}}$ in Remark \ref{rem:shrinkage_est} when $K=100$ and $\hat{\theta}_k \sim N((-1)^k \kappa, 1)$.} \label{fig:w_opt}
\end{figure}

Figure \ref{fig:w_opt} also shows the shrinkage factor of the James-Stein-type estimator when $K=100$ and $\hat{\theta}_k \sim N((-1)^k \kappa, 1)$. Concretely, the dot-dashed line denotes the median of $\hat{w}_{\mathrm{JS}}$ in Remark \ref{rem:shrinkage_est}. In this setting, $\bm{\theta}$ is contained in $\Theta(\kappa)$ because $\theta_k = (-1)^k \kappa$ and $\overline{\theta} = 0$. Similar to $w^{\ast}(\kappa)$, the median of $\hat{w}_{\mathrm{JS}}$ increases as $\kappa$ increases. However, $\hat{w}_{\mathrm{JS}}$ tends to be smaller than $w_k^{\ast}(\kappa)$ for almost all $\kappa$. This implies that the proposed shrinkage rule places greater emphasis on the bias than on the variance compared with the James–Stein-type estimator.

We compare the maximum regrets of the shrinkage, CES, and pooling rules when $\kappa$ is correctly specified. We consider the following two cases:
\begin{eqnarray*}
    \text{Case (i):} & & \text{$\left( \sigma_1, \cdots, \sigma_K \right) = \left( 1, \ldots, 1 \right)$ and $\left( p_1, \cdots, p_K \right) = \left( \frac{1}{K}, \ldots, \frac{1}{K} \right)$} \\
    \text{Case (ii):} & & \text{$\left( \sigma_1, \cdots, \sigma_{K/2}, \sigma_{K/2+1}, \cdots, \sigma_K \right) = \left( 1, \ldots, 1, \sqrt{2}, \ldots, \sqrt{2} \right)$} \\
    & & \hspace{0.4in} \text{and $\left( p_1, \cdots, p_{K/2}, p_{K/2+1}, \cdots, p_K \right) = \left( \frac{4}{3K}, \ldots, \frac{4}{3K}, \frac{2}{3K}, \ldots, \frac{2}{3K} \right)$}
\end{eqnarray*}
Case (i) corresponds to the balanced design, where the standard errors and group sizes are the same for all groups. Case (ii) corresponds to the unbalanced design, where the group sizes of one half of the groups are larger than those of the other half and the standard errors for the larger groups are smaller than those for the remaining groups. Figure \ref{fig:max_regret_K10_100} shows the maximum regrets of the shrinkage, CES, and pooling rules for $\kappa = 0, \, 0.1, \, \ldots, \, 0.9, \, 1.0$ when $K=10$ and $100$.\footnote{Similar to Figure \ref{fig:true_K20}, we calculate the maximum regrets using the Monte Carlo approximation.} In Case (i), we also calculate the treatment rule based on the James-Stein-type estimator $\hat{\bm{\delta}}^{\mathrm{JS}}(\hat{\bm{\theta}}) \equiv \left( \hat{\delta}^{\mathrm{JS}}_1(\hat{\bm{\theta}}), \ldots , \hat{\delta}^{\mathrm{JS}}_K(\hat{\bm{\theta}}) \right)'$ defined as
\[
\hat{\delta}^{\mathrm{JS}}_k(\hat{\bm{\theta}}) \ \equiv \ 1 \left\{ \hat{w}_{\mathrm{JS}} \cdot \hat{\theta}_k + (1-\hat{w}_{\mathrm{JS}}) \cdot \mathrm{ave}(\hat{\bm{\theta}}) \geq 0 \right\}.
\]
Because this James-Stein-type estimator is not theoretically justified under heteroscedasticity, we do not report $\hat{\bm{\delta}}^{\mathrm{JS}}$ in Case (ii).

\begin{figure}[htbp]
    \begin{tabular}{cc}
      \begin{minipage}[t]{0.45\hsize}
        \centering
        \includegraphics[width=7cm]{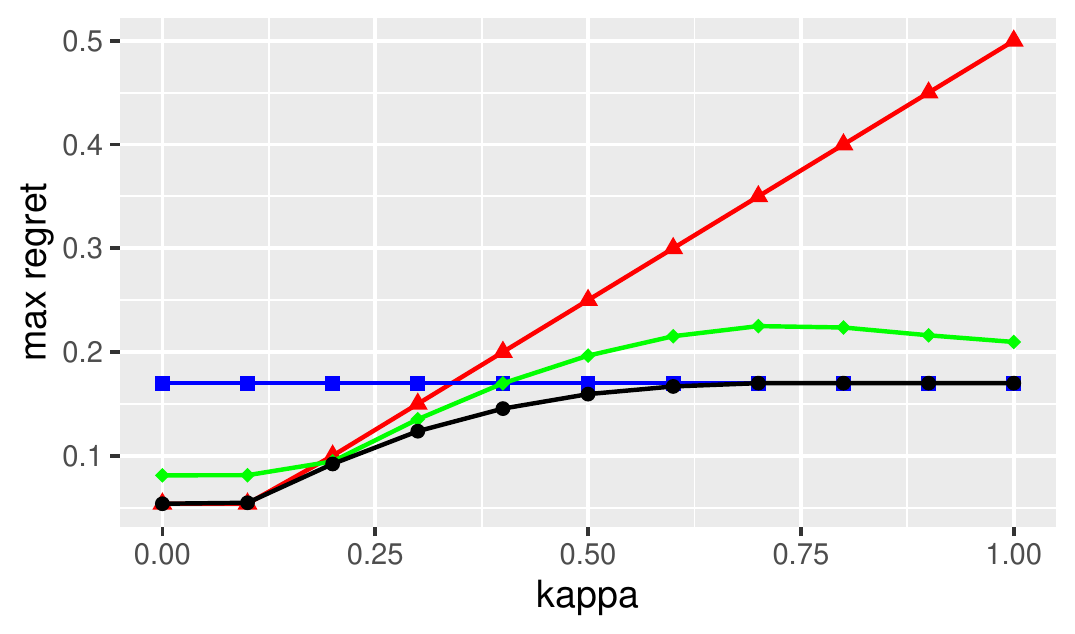}
        \subcaption{Case (i), $K=10$.}
        \label{fig:max_regret_K10_balance}
      \end{minipage} &
      \begin{minipage}[t]{0.45\hsize}
        \centering
        \includegraphics[width=7cm]{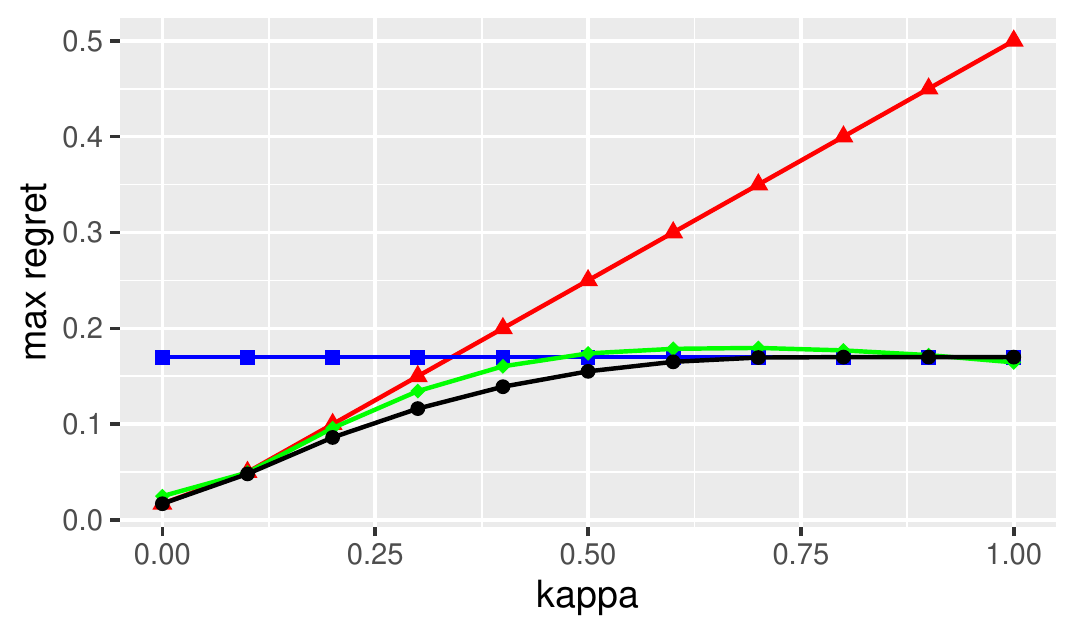}
        \subcaption{Case (i), $K=100$.}
        \label{fig:max_regret_K100_balance}
      \end{minipage} \\
   
      \begin{minipage}[t]{0.45\hsize}
        \centering
        \includegraphics[width=7cm]{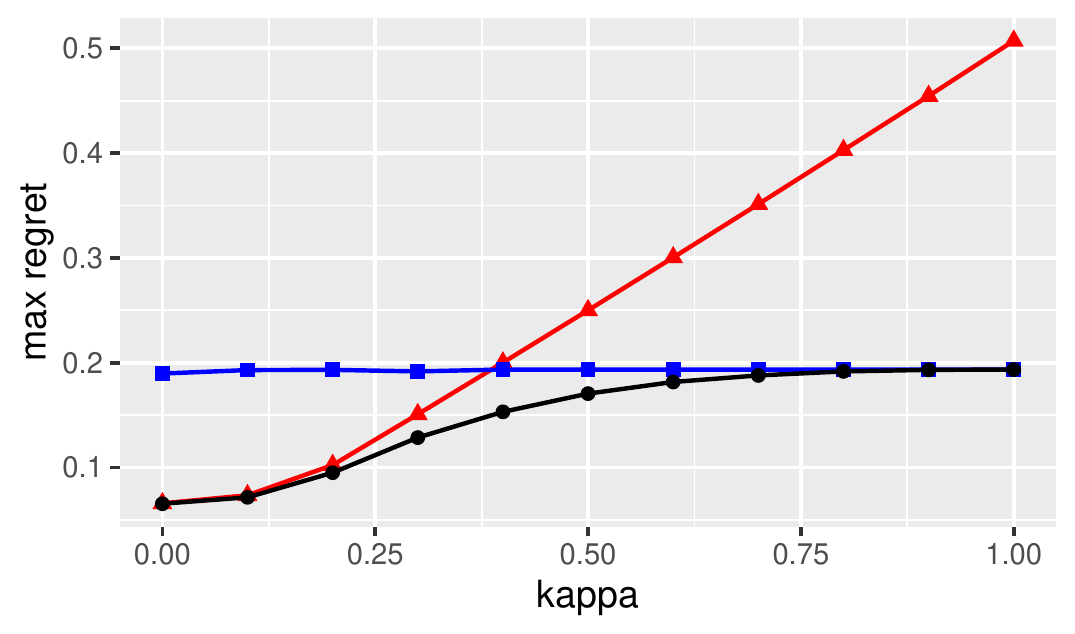}
        \subcaption{Case (ii), $K=10$.}
        \label{fig:max_regret_K10_unbalance}
      \end{minipage} &
      \begin{minipage}[t]{0.45\hsize}
        \centering
        \includegraphics[width=7cm]{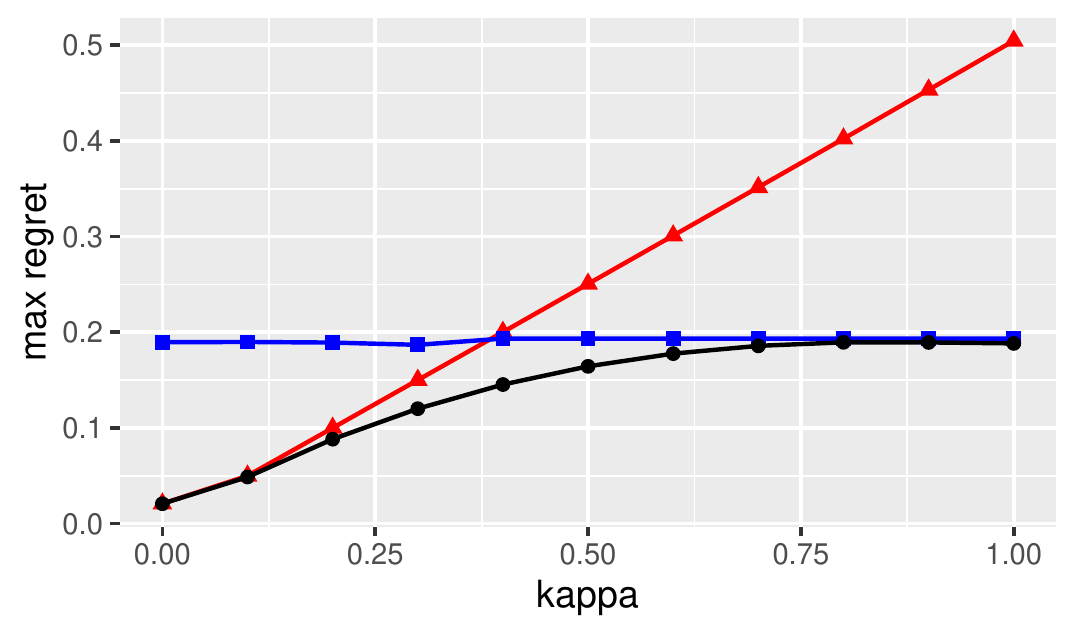}
        \subcaption{Case (ii), $K=100$.}
        \label{fig:max_regret_K100_unbalance}
      \end{minipage} 
    \end{tabular}
    \caption{The black circles, blue squares, and red triangles denote the maximum regrets of the shrinkage, CES, and pooling rules. The green diamonds denote the maximum regrets of the treatment rule based on the James-Stein-type estimator.}\label{fig:max_regret_K10_100}
\end{figure}

Figures \ref{fig:max_regret_K10_balance} and \ref{fig:max_regret_K100_balance} show that the maximum regret of the proposed shrinkage rule is less than or equal to that of the CES rule for all $\kappa$ in the balanced design. This result is consistent with Theorem \ref{thm:CES_correct}. In the balanced design, the proposed shrinkage rule is slightly inferior to the pooling rule when $K=10$ and $\kappa = 0.1$, but the maximum regrets of $\hat{\bm{\delta}}^{\bm{w}^{\ast}(\kappa)}$ are less than or equal to those of $\hat{\bm{\delta}}^{\mathrm{pool}}$ in all other settings. In addition, the maximum regrets of $\hat{\bm{\delta}}^{\bm{w}^{\ast}(\kappa)}$ are less than that of $\hat{\bm{\delta}}^{\mathrm{JS}}$ in almost all settings.\footnote{When calculating the maximum regret of $\hat{\bm{\delta}}^{\mathrm{JS}}$, we also calculate the regret $R(\bm{\theta},\hat{\bm{\delta}}^{\mathrm{JS}})$ using the Monte Carlo approximation. In Figure \ref{fig:max_regret_K100_balance}, the maximum regret of $\hat{\bm{\delta}}^{\mathrm{JS}}$ is less than that of $\hat{\bm{\delta}}^{\bm{w}^{\ast}(\kappa)}$ for $\kappa = 1$. However, this is probably due to approximation errors.} However, the difference between $\hat{\bm{\delta}}^{\bm{w}^{\ast}(\kappa)}$ and $\hat{\bm{\delta}}^{\mathrm{JS}}$ decreases as $K$ increases. In the unbalanced design, similar results are obtained for the shrinkage, CES, and pooling rules. In Appendix B, we provide additional results for $K=2$ and $4$.

Next, we calculate the maximum regret of the shrinkage rule when $\kappa$ is misspecified, that is, we calculate $\max_{\bm{\theta} \in \Theta(\kappa)} R(\bm{\theta},\hat{\bm{\delta}}^{\bm{w}^{\ast}(\kappa')})$. We consider the four misspecification situations, $\kappa' = 0.5 \kappa$, $\kappa' = 0.8 \kappa$, $\kappa' = 1.2 \kappa$, and $\kappa' = 2 \kappa$. Figure \ref{fig:max_regret_K10_mis} shows the maximum regrets of $\hat{\bm{\delta}}^{\bm{w}^{\ast}(\kappa)}$, $\hat{\bm{\delta}}^{\mathrm{CES}}$, $\hat{\bm{\delta}}^{\mathrm{pool}}$, and $\hat{\bm{\delta}}^{\mathrm{JS}}$ in the balanced design when $K=10$ and $\kappa$ is misspecified. Because $\hat{\bm{\delta}}^{\mathrm{CES}}$, $\hat{\bm{\delta}}^{\mathrm{pool}}$, and $\hat{\bm{\delta}}^{\mathrm{JS}}$ is not affected by misspecification, the maximum regrets of $\hat{\bm{\delta}}^{\mathrm{CES}}$, $\hat{\bm{\delta}}^{\mathrm{pool}}$, and $\hat{\bm{\delta}}^{\mathrm{JS}}$ are the same as those in Figure \ref{fig:max_regret_K10_balance}.

\begin{figure}[htbp]
    \begin{tabular}{cc}
      \begin{minipage}[t]{0.45\hsize}
        \centering
        \includegraphics[width=7cm]{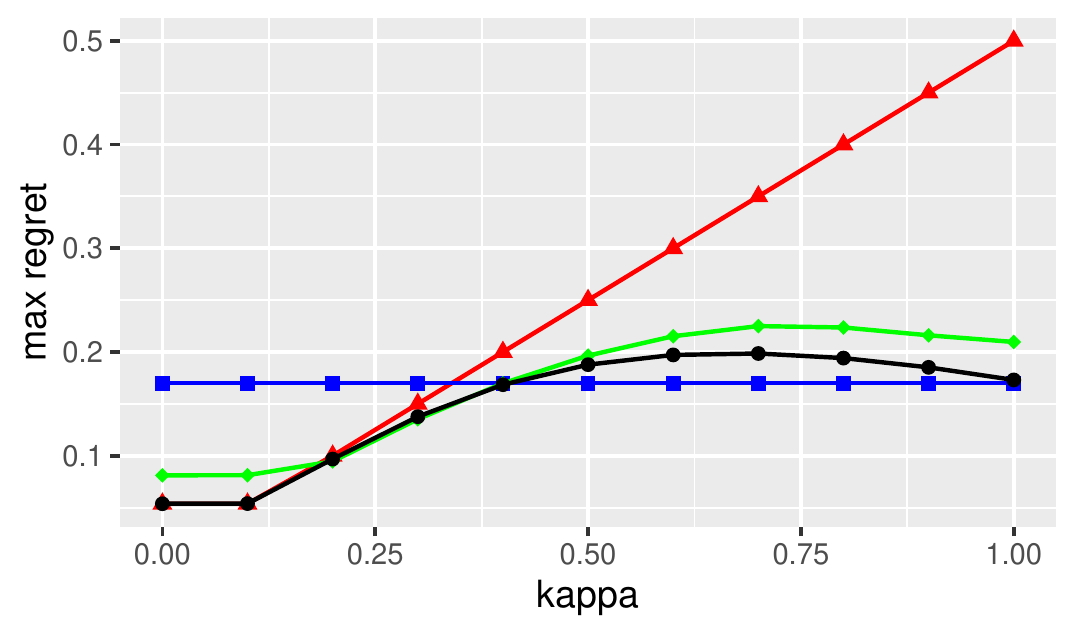}
        \subcaption{$\kappa' = 0.5 \kappa$.}
        \label{fig:max_regret_K10_mis_a05}
      \end{minipage} &
      \begin{minipage}[t]{0.45\hsize}
        \centering
        \includegraphics[width=7cm]{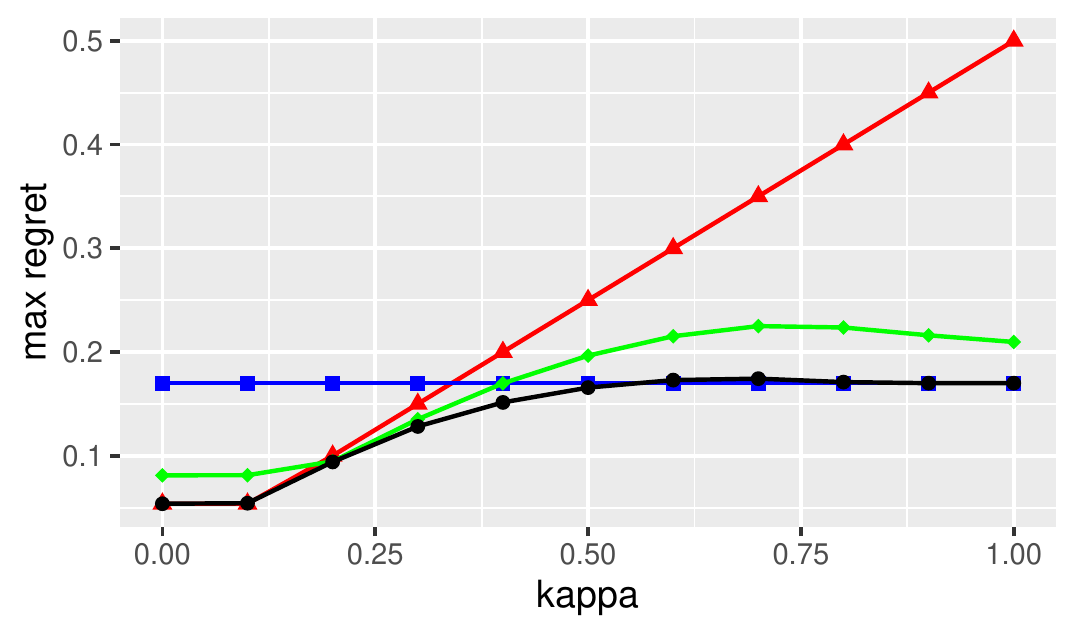}
        \subcaption{$\kappa' = 0.8 \kappa$.}
        \label{fig:max_regret_K10_mis_a08}
      \end{minipage} \\
   
      \begin{minipage}[t]{0.45\hsize}
        \centering
        \includegraphics[width=7cm]{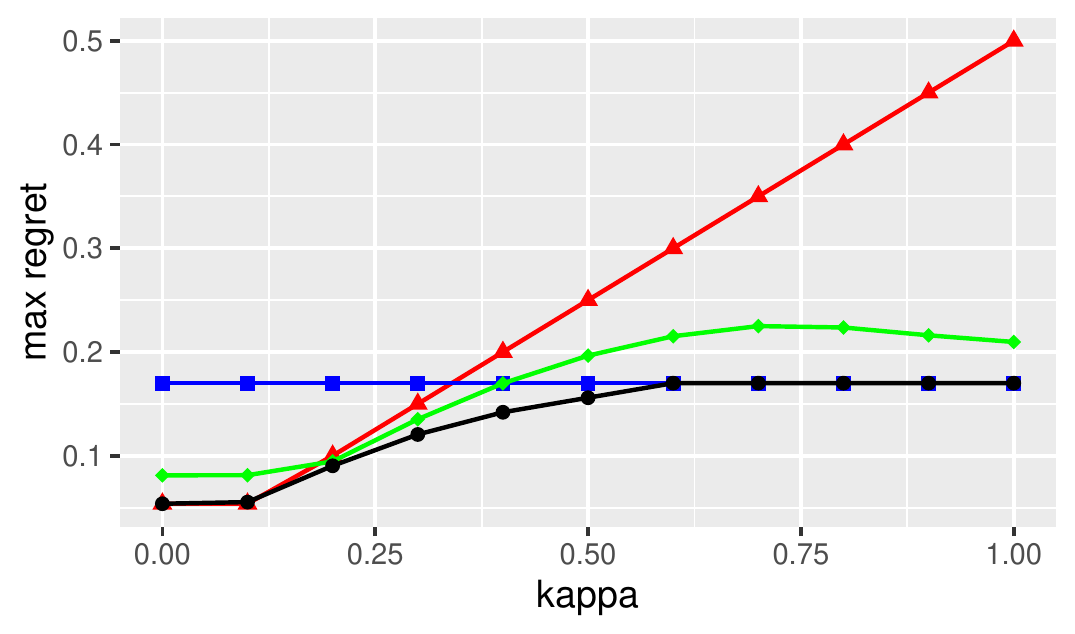}
        \subcaption{$\kappa' = 1.2 \kappa$.}
        \label{fig:max_regret_K10_mis_a12}
      \end{minipage} &
      \begin{minipage}[t]{0.45\hsize}
        \centering
        \includegraphics[width=7cm]{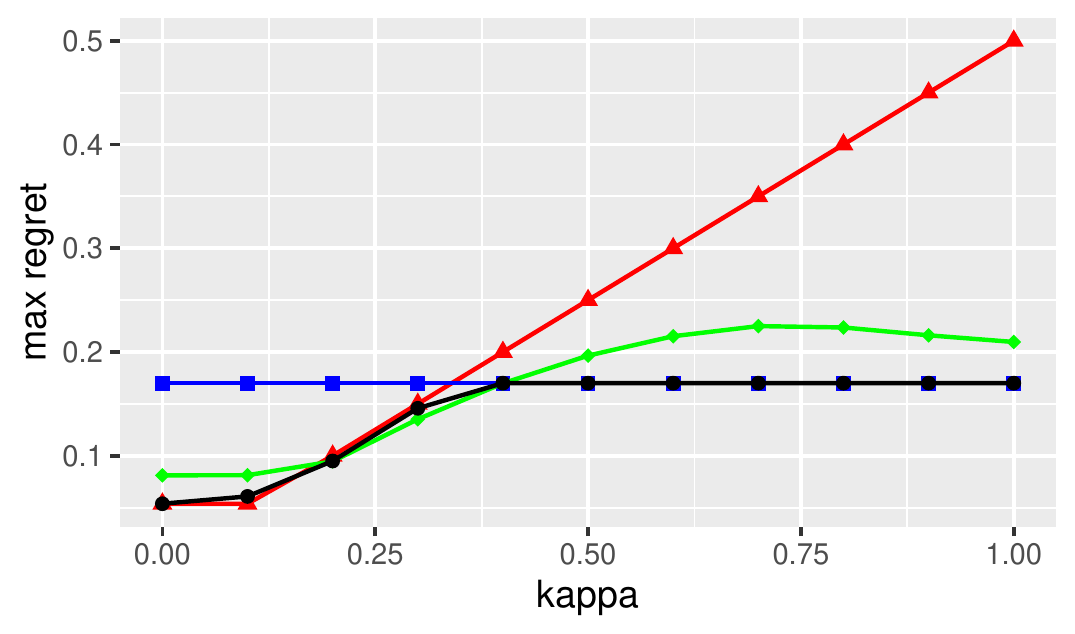}
        \subcaption{$\kappa' = 2 \kappa$.}
        \label{fig:max_regret_K10_mis_a20}
      \end{minipage} 
    \end{tabular}
    \caption{The black circles denotes $\max_{\bm{\theta} \in \Theta(\kappa)} R(\bm{\theta},\hat{\bm{\delta}}^{\bm{w}^{\ast}(\kappa')})$ when $K=10$, $\left( \sigma_1, \cdots, \sigma_K \right) = \left( 1, \ldots, 1 \right)$, and $\left( p_1, \cdots, p_K \right) = \left( \frac{1}{K}, \ldots, \frac{1}{K} \right)$. The blue squares, red triangles, and green diamonds are the same as those in Figure \ref{fig:max_regret_K10_balance}.}\label{fig:max_regret_K10_mis}
\end{figure}

When $\kappa \leq \kappa'$, as discussed in Section \ref{SubSec:misspecified}, the maximum regrets of $\hat{\bm{\delta}}^{\bm{w}^{\ast}(\kappa')}$ are less than or equal to those of $\hat{\bm{\delta}}^{\mathrm{CES}}$ in all settings. However, the range of $\kappa$ in which $\hat{\bm{\delta}}^{\bm{w}^{\ast}(\kappa')}$ has a smaller maximum regret than $\hat{\bm{\delta}}^{\mathrm{CES}}$ is narrower than that of Figure\ref{fig:max_regret_K10_balance} due to the misspecification of $\kappa$. In addition, the proposed shrinkage rule has a smaller maximum regret than the CES and pooling rules for some $\kappa$ when the degree of misspecification is not so large.

When $\kappa \geq \kappa'$, as discussed in Section \ref{SubSec:misspecified}, the maximum regrets of $\hat{\bm{\delta}}^{\bm{w}^{\ast}(\kappa')}$ are less than those of $\hat{\bm{\delta}}^{\mathrm{pool}}$ when $\kappa$ is sufficiently large. In contrast to the case where $\kappa \leq \kappa'$, the maximum regrets of $\hat{\bm{\delta}}^{\bm{w}^{\ast}(\kappa')}$ are larger than those of $\hat{\bm{\delta}}^{\mathrm{CES}}$ for some $\kappa$. However, the amount by which $\hat{\bm{\delta}}^{\bm{w}^{\ast}(\kappa')}$ exceeds $\hat{\bm{\delta}}^{\mathrm{CES}}$ decreases as the degree of misspecification decreases. In summary, we obtain results similar to the correctly specified case when the degree of misspecification is small, whereas $\hat{\bm{\delta}}^{\bm{w}^{\ast}(\kappa')}$ may be inferior to other rules when $\kappa'$ is significantly smaller than $\kappa$.

\section{Empirical application}\label{Sec:real}
We illustrate the proposed method by applying it to experimental data from the National Job Training Partnership Act (JTPA) Study and using the dataset in \cite{abadie2002instrumental}. The JTPA study is a randomized controlled trial whose purpose was to measure the impact of a training program on earnings. It also collected background information on the applicants prior to the random assignment and obtained data on their earnings for 30 months following the assignment.

We construct 24 subgroups using the following characteristics: race (Black, Hispanic, or other), sex (male or female), marital status (married or unmarried), and working status prior to random assignment (worked for at least 12 weeks in the 12 months preceding random assignment or not). Table \ref{table:group_number} shows the corresponding group numbers. For each subgroup, we calculate the treatment effect of the training program on earnings for 30 months following the assignment and its standard error. Following \cite{kitagawa2018should}, we set the treatment cost as \$774. Hence, we use the CATE estimate minus \$774 as $\hat{\theta}_k$. Figure \ref{fig:CI} shows that the benefits of the training program vary across subgroups but are statistically insignificant for all subgroups.

\begin{table}[h]
\begin{center}
\begin{tabular}{r|cccccc} \hline
    race & \multicolumn{2}{c}{Black} & \multicolumn{2}{c}{Hispanic} & \multicolumn{2}{c}{other} \\ \hline
    working status & $> 12$ & $\leq 12$ & $> 12$ & $\leq 12$ & $> 12$ & $\leq 12$ \\ \hline
   female, unmarried & 1 & 4 & 2 & 5 & 3 & 6 \\
   female, married & 7 & 10 & 8 & 11 & 9 & 12 \\
   male, unmarried & 13 & 16 & 14 & 17 & 15 & 18 \\
   male, married & 19 & 22 & 20 & 23 & 21 & 24 \\ \hline
\end{tabular}
\caption{This table shows the corresponding group numbers. For example, group 10 corresponds to the group of married Black women who worked at least 12 weeks.} \label{table:group_number}
\end{center}
\end{table}

\begin{figure}[h] 
\centering
\includegraphics[width=14cm]{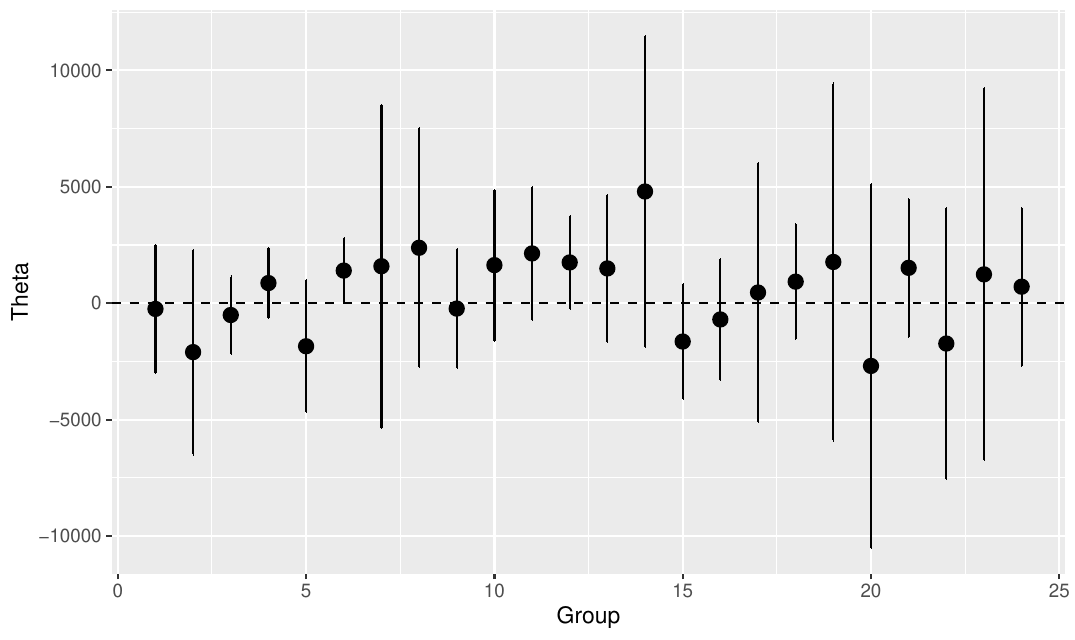}
\caption{The black dots denote the CATE estimates minus \$774 and the bars denote the $95\%$ confidence intervals.} \label{fig:CI}
\end{figure}

As the average of $\hat{\bm{\theta}}$ is approximately \$541, the pooling rule determines to treat the individuals in all subgroups. However, Figure \ref{fig:CI} indicates that the decisions of the CES rule vary across the subgroups. Hence, the shrinkage rule may differ from the CES rule depending on the values of the shrinkage factors.

In this empirical application, we consider two James-Stein-type rules. The first is a treatment rule based on the following James–Stein-type shrinkage estimator:
\begin{equation}
    \hat{\theta}_k^{\mathrm{JS}} \ \equiv \ \hat{w}_{k} \cdot \hat{\theta}_k + (1-\hat{w}_{k}) \cdot \mathrm{ave} (\hat{\bm{\theta}}), \label{JS-positive-estimator}
\end{equation}
where
\[
\hat{w}_k \ \equiv \ \max \left\{ 1 - \frac{(K-3) \sigma_k^2}{\sum_{j=1}^K \left( \hat{\theta}_k - \mathrm{ave}(\hat{\bm{\theta}}) \right)^2} , 0\right\}.
\]
In contrast to the James–Stein-type shrinkage estimator introduced in Remark \ref{rem:shrinkage_est}, this estimator is specifically designed to ensure that the shrinkage factor remains non-negative. Using (\ref{JS-positive-estimator}), we define the following James-Stein-type rule:
\begin{equation}
\hat{\bm{\delta}}^{\mathrm{JS}}(\hat{\bm{\theta}}) \ \equiv \ \left( \hat{\delta}^{\mathrm{JS}}_1(\hat{\bm{\theta}}), \ldots, \hat{\delta}^{\mathrm{JS}}_K(\hat{\bm{\theta}}) \right)', \ \ \text{where $\hat{\delta}^{\mathrm{JS}}_k(\hat{\bm{\theta}}) \ \equiv 1\left\{ \hat{\theta}_k^{\mathrm{JS}} \geq 0 \right\}$.} \label{JS-positive-rule} \nonumber
\end{equation}
Because the shrinkage estimator (\ref{JS-positive-estimator}) is not theoretically justified under heteroscedasticity, we consider an alternative James-Stein-type rule that accommodates heteroscedasticity. \cite{xie2012sure} propose the following James–Stein-type shrinkage estimator in the heteroscedastic normal means model:
\begin{equation}
    \hat{\theta}^{\mathrm{XKB}}_k \ \equiv \ \tilde{w}_{k} \cdot \hat{\theta}_k + (1-\tilde{w}_{k}) \cdot \mathrm{ave} (\hat{\bm{\theta}}), \label{XKB-estimator}
\end{equation}
where $\tilde{w}_{k}  \ \equiv \ \frac{\hat{\lambda}}{\sigma_k^2 + \hat{\lambda}}$ and
\[
\hat{\lambda} \ \equiv \ \mathrm{arg} \min_{\lambda \geq 0} \left[ \frac{1}{K} \sum_{k=1}^K \left( \frac{\sigma_k^2}{\sigma_k^2 + \lambda} \right)^2 \left\{\hat{\theta}_k - \mathrm{ave}(\hat{\bm{\theta}})\right\}^2 + \frac{1}{K} \sum_{k=1}^K\left( \frac{\sigma_k^2}{\sigma_k^2 + \lambda} \right) \left( \lambda - \sigma_k^2 + \frac{2}{K} \sigma_k^2 \right) \right].
\]
\cite{xie2012sure} establish the asymptotic optimality property for this shrinkage estimator when $K \to \infty$. Using this shrinkage estimator (\ref{XKB-estimator}), we define the following James-Stein-type rule:
\begin{equation}
\hat{\bm{\delta}}^{\mathrm{XKB}}(\hat{\bm{\theta}}) \ \equiv \ \left( \hat{\delta}^{\mathrm{XKB}}_1(\hat{\bm{\theta}}), \ldots, \hat{\delta}^{\mathrm{XKB}}_K(\hat{\bm{\theta}}) \right)', \ \ \text{where $\hat{\delta}^{\mathrm{XKB}}_k(\hat{\bm{\theta}}) \ \equiv 1\left\{ \hat{\theta}^{\mathrm{XKB}}_k \geq 0 \right\}$.} \label{XKB rule} \nonumber
\end{equation}

We calculate the shrinkage factors $\bm{w}^{\ast}(\kappa)$ for $\kappa = 500, \, 1{,}000$. Figure \ref{fig:se} shows the relationship between the shrinkage factor and standard error. As expected, the shrinkage factor decreases as the standard error increases. For $\kappa = 1{,}000$, the shrinkage rule becomes the CES rule when $\sigma_k$ is less than around $1{,}300$. Proposition \ref{prop:K_asymptotics} indicates that $w^{\ast}_k (\kappa)$ approaches $1$ asymptotically when $\kappa / \sigma_k \leq t^{\ast}(0) \simeq 0.75$, which is equivalent to $\sigma_k < \kappa / t^{\ast}(0) \simeq 1.33 \kappa$. Hence, this result indicates that the approximation in Proposition \ref{prop:K_asymptotics} is useful. 

\begin{figure}[htbp]
    \begin{tabular}{cc}
      \begin{minipage}[t]{0.45\hsize}
        \centering
        \includegraphics[width=7.5cm]{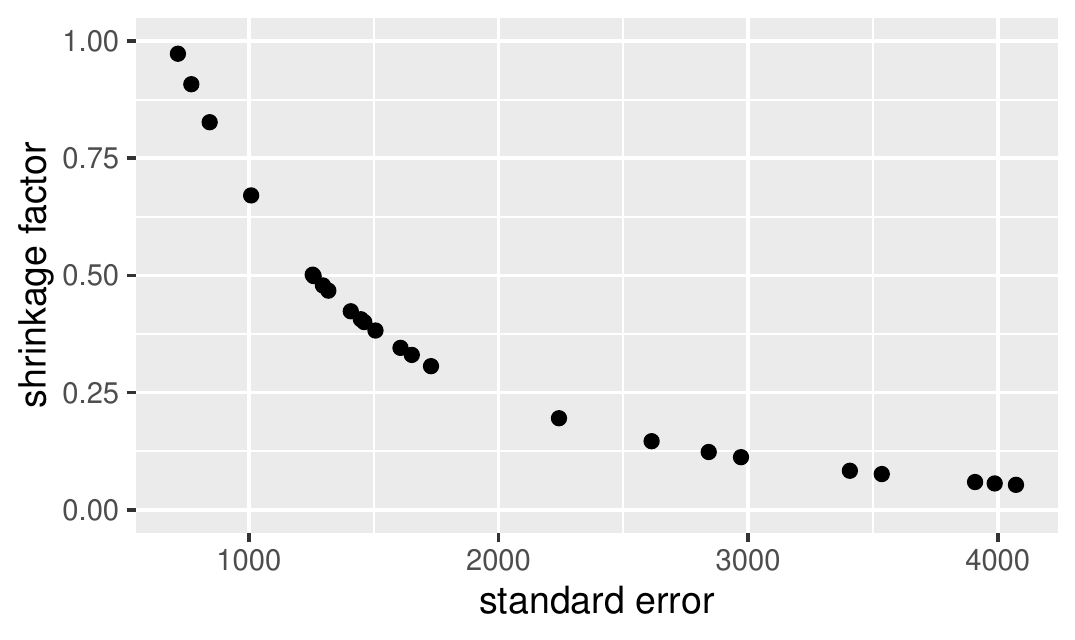}
        \subcaption{$\kappa = 500$.}
        \label{fig:se500}
      \end{minipage} &
      \begin{minipage}[t]{0.45\hsize}
        \centering
        \includegraphics[width=7.5cm]{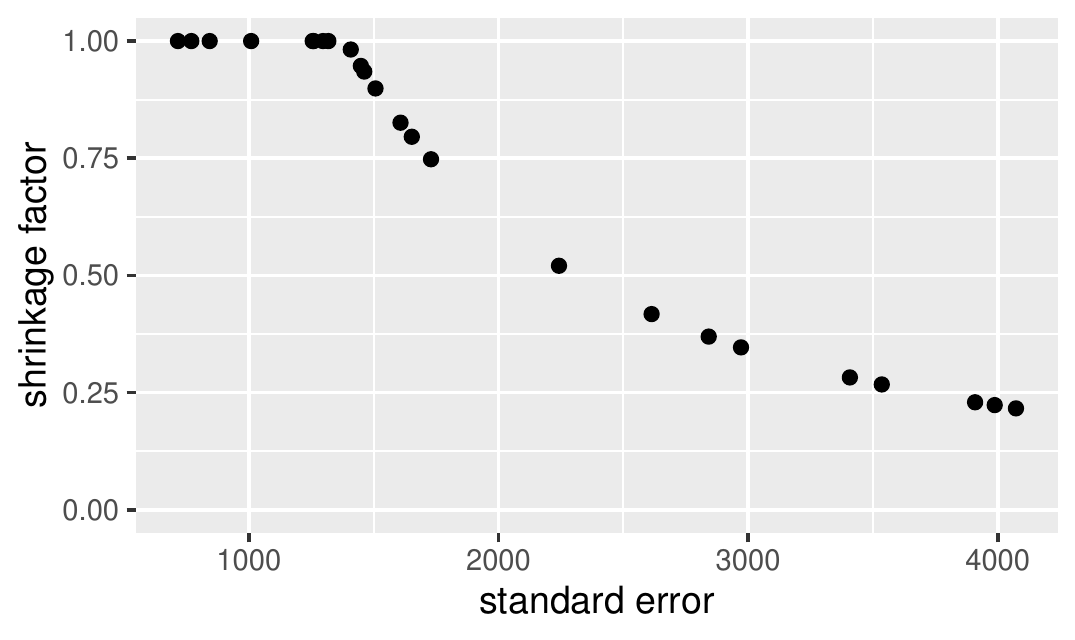}
        \subcaption{$\kappa = 1{,}000$.}
        \label{fig:se1000}
      \end{minipage}
    \end{tabular}
    \caption{The shrinkage factor when $\kappa = 500, 1{,}000$. The horizontal and vertical axes represent the values of $\sigma_k$ and $w_k^{\ast}(\kappa)$.} \label{fig:se}
\end{figure}

Figures \ref{fig:res500} and \ref{fig:res1000} illustrate the shrinkage estimates $w_k^{\ast}(\kappa) \cdot \hat{\theta}_k + (1-w_k^{\ast}(\kappa)) \cdot \text{ave}(\hat{\bm{\theta}})$ for $\kappa = 500, \, 1{,}000$. As the red line denotes the average of $\hat{\bm{\theta}}$, the shrinkage estimate (white circle) is closer to the red line than $\hat{\theta}_k$ (black circle). As the average of $\hat{\bm{\theta}}$ is positive, the decision of the shrinkage rule is the same as that of the CES rule when $\hat{\theta}_k$ is positive. Whereas, the decision regarding the shrinkage rule can differ from that  regarding the CES rule when $\hat{\theta}_k$ is negative. Figure \ref{fig:res500} shows that the decision of the shrinkage rule is not identical to that of the CES rule in certain subgroups. However, the shrinkage rule makes the same decisions as the CES rule when $\kappa = 1{,}000$.

\begin{figure}[h] 
\centering
\includegraphics[width=14cm]{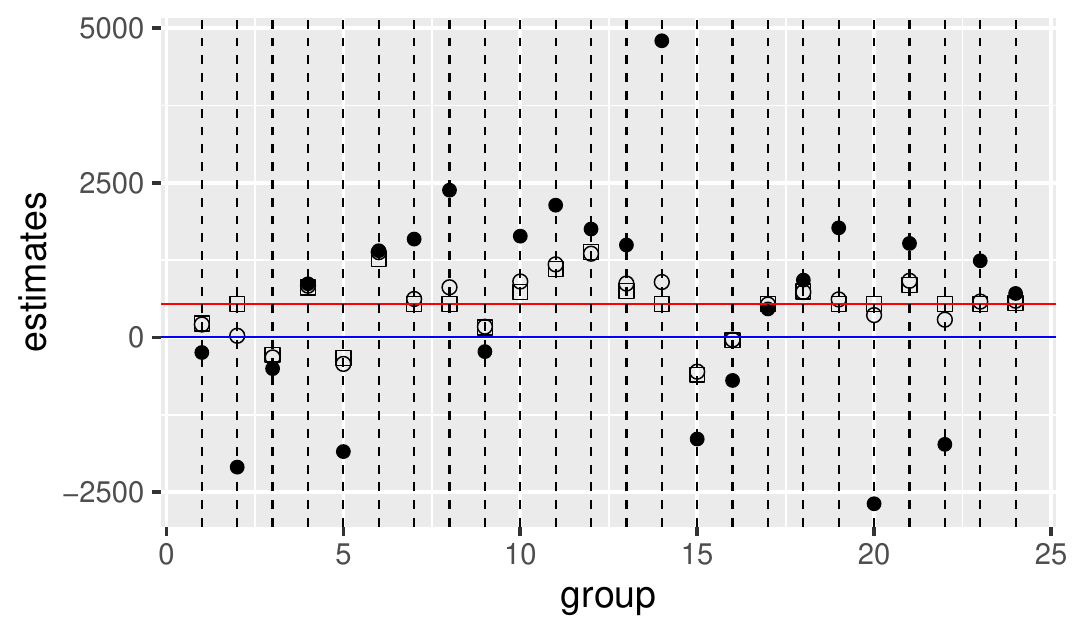}
\caption{The shrinkage rule when $\kappa = 500$. The black circles denote $\hat{\theta}_k$, the white circles denote the shrinkage estimates $w_k^{\ast}(\kappa) \cdot \hat{\theta}_k + (1-w_k^{\ast}(\kappa)) \cdot \text{ave}(\hat{\bm{\theta}})$, the white squares denote the James–Stein-type estimates $\hat{\theta}_k^{\mathrm{JS}}$, and the blue and red lines denote $0$ and the average of $\hat{\bm{\theta}}$, respectively.} \label{fig:res500}
\end{figure}

\begin{figure}[h] 
\centering
\includegraphics[width=14cm]{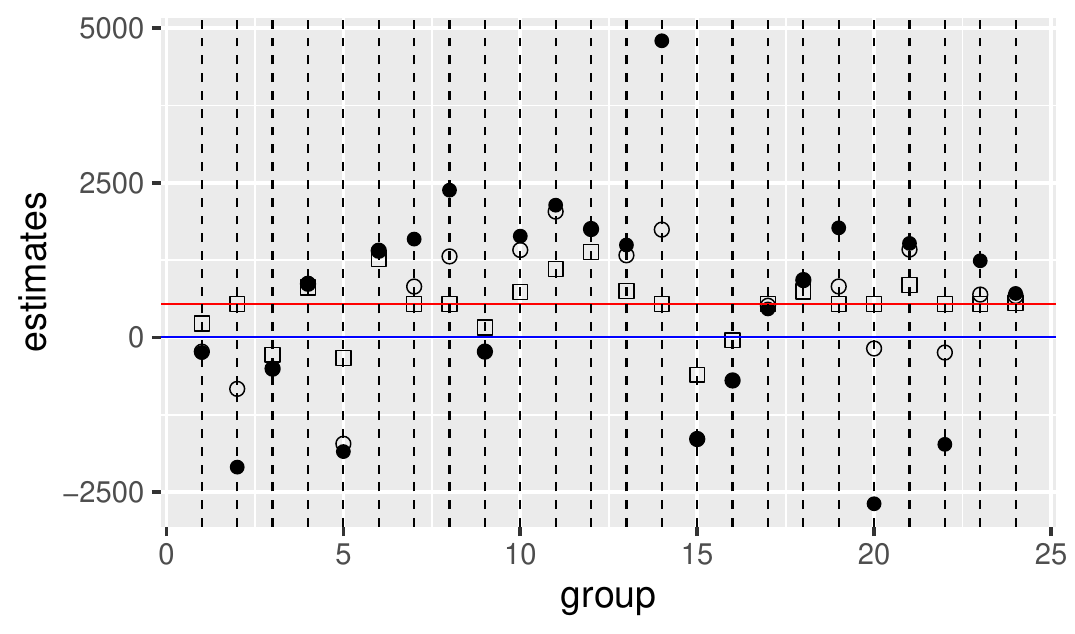}
\caption{The shrinkage rule when $\kappa = 1{,}000$. The black circles denote $\hat{\theta}_k$, the white circles denote the shrinkage estimates $w_k^{\ast}(\kappa) \cdot \hat{\theta}_k + (1-w_k^{\ast}(\kappa)) \cdot \text{ave}(\hat{\bm{\theta}})$, the white squares denote the James–Stein-type estimates $\hat{\theta}_k^{\mathrm{JS}}$, and the blue and red lines denote $0$ and the average of $\hat{\bm{\theta}}$, respectively.} \label{fig:res1000}
\end{figure}

In addition, we also compute the James-Stein-type rules $\hat{\bm{\delta}}^{\mathrm{JS}}$ and $\hat{\bm{\delta}}^{\mathrm{XKB}}$. Figures \ref{fig:res500} and \ref{fig:res1000} illustrate the James-Stein-type estimates $\hat{\theta}^{\mathrm{JS}}_k$. These estimates are close to our shrinkage estimates $w_k^{\ast}(\kappa) \cdot \hat{\theta}_k + (1-w_k^{\ast}(\kappa)) \cdot \text{ave}(\hat{\bm{\theta}})$ when $\kappa = 500$. Moreover, $\hat{\bm{\delta}}^{\mathrm{JS}}$ makes the same decisions as the shrinkage rule $\hat{\bm{\delta}}^{\bm{w}^{\ast}(\kappa)}$ for $\kappa = 500$. In contrast, all shrinkage factors of (\ref{XKB-estimator}) are approximately zero. This implies that $\hat{\bm{\delta}}^{\mathrm{XKB}}$ makes the same decision as the pooling rule. The pooling rule is nearly optimal when $\kappa$ is small and, as we will show later, the hypothesis $H_0:\bm{\theta} \in \Theta(0)$ is not rejected. Hence, although this result may appear extreme, it is not necessarily unreasonable. We note that $\hat{\theta}^{\mathrm{XKB}}_k$ is asymptotically justified but may be unreliable when $K$ is small.

This analysis focuses on the treatment choice problem when $\kappa = 500, \, 1{,}000$. As $w^{\ast}_k(\kappa)$ is increasing with respect to $\kappa$, the shrinkage rule makes the same decisions as the CES rule when $\kappa$ is greater than $1{,}000$. Hence, the decision regarding the shrinkage rule differs from that regarding the CES rule only when $\kappa$ is small. However, the choice of $\kappa = 500$ is not unrealistic. As discussed in Remark \ref{rem:choice_kappa}, we can assess whether a given value of $\kappa$ is too small. Letting $Z_k \sim N(0,\sigma_k^2)$, the $95 \%$ quantile of $\max_{1\leq k \leq K} |Z_k - \overline{Z}|$ is about $9{,}500$. Whereas, the realized value of $\max_{1 \leq k \leq K} \left| \hat{\theta}_k - \mathrm{ave}(\hat{\bm{\theta}}) \right|$ is $4{,}256$. This implies that the hypothesis $H_0:\bm{\theta} \in \Theta(0)$ is not rejected and any value of $\kappa$ is consistent with the actual data. Therefore, this empirical application shows that the decision of the shrinkage rule can differ from that of the CES rule even if $\kappa$ is a realistic value.

Finally, we compare the proposed shrinkage rule with the empirical welfare maximization method. As discussed in Remark \ref{rem:EWM}, the EWM rule is equivalent to the CES rule if there are no restrictions on the class of candidate treatment rules. In this analysis, we consider the constraints that decisions are not changed based on race and that if men with the same characteristics receive treatment, then women also receive treatment. Under the above constraints, the EWM rule chooses to give treatment to all groups except groups 1, 2, 3, 13, 14, and 15. The EWM rule determines not to treat unmarried Black or Hispanic individuals who worked for more than 12 weeks. On the contrary, the proposed shrinkage rule determines to treat such individuals when $\kappa=500$.

\section{Conclusion}\label{Sec:conclusion}
This study examined the problem of determining whether to treat individuals based on observed covariates. Particularly, we proposed a computationally tractable shrinkage rule that selects the shrinkage factor by minimizing the upper bound of the maximum regret. We also provided upper bounds of the ratio of the maximum regret of the shrinkage rule to those of the CES and pooling rules when the space of CATEs was correctly specified or misspecified. The theoretical and numerical results show that our shrinkage rule performs better than the CES and pooling rules in many cases when the space of CATEs is correctly specified. In addition, the results were robust to the misspecifications of the space of CATEs. Particularly, we found that the maximum regret of the shrinkage rule can be strictly smaller than that of the CES and pooling rules when the dispersion of the CATEs is moderate. Finally, we applied our method to experimental data from the JTPA study and showed that the decision of the shrinkage rule can differ from that of the CES rule even if $\kappa$ is a realistic value.

\clearpage

\renewcommand{\theequation}{A.\arabic{equation}}
\setcounter{equation}{0}
\section*{Appendix A. Proofs and lemmas}\label{Sec:appendix_proof}

\begin{Lemma}\label{lem:eta_derivative}
For any $a \geq 0$, we have
$$
\eta'(a) \ \equiv \ \frac{d}{da} \eta(a) \ = \ \Phi(-t^*(a) + a).
$$
In addition, $-t^*(a) + a$ is strictly increasing in $a$.
\end{Lemma}
\begin{proof}
This lemma follows from the proof of Lemma 2 in \cite{ishihara2021evidence}.
\end{proof}\vspace{0.1in}

\begin{Lemma}\label{lem:eta}
For any $a \geq 0$ and $v \in \mathbb{R}$, we have
\begin{eqnarray*}
&  v \Phi(-v) + \Phi(-v) a \ \leq \ \eta(a) \ \leq \ \eta(0) + a, \\
&  \eta(0) \sqrt{1+a^2} \ \leq \ \eta(a) \ \leq \ \sqrt{1+a^2}.
\end{eqnarray*}
\end{Lemma}
\begin{proof}
This lemma follows from Lemmas 1 and 2 in \cite{ishihara2021evidence}.
\end{proof}\vspace{0.1in}

\begin{Lemma}\label{lem:prop1}
Suppose that $\sigma_1 = \cdots = \sigma_K$ and $p_1 = \cdots = p_K$. Then, it follows that
\[
\overline{R}_{\mathrm{true}}(w) \ \geq \ L_{\mathrm{true}}(w) \ \equiv \ \begin{cases}
\frac{1}{2} \overline{R}_{\mathrm{upper}}(w) + \frac{1}{2} \overline{r}_1(w), & \text{if $K$ is even} \\
\frac{K-1}{2K} \overline{R}_{\mathrm{upper}}(w) + \frac{K-1}{2K} \overline{r}_1(w) + \frac{1}{K} \overline{r}_2(w), & \text{if $K$ is odd}
\end{cases}
\]
holds for any $w \in [0,1]$, where
\begin{eqnarray*}
\overline{r}_{1}(w) & \equiv & s(w) \cdot \left| t^{\ast}\left( \frac{(1-w)\kappa}{s(w)} \right) - \frac{2\kappa}{s(w)} \right| \\
& & \times  \Phi \left( -  \left| t^{\ast}\left( \frac{(1-w)\kappa}{s(w)} \right) - \frac{2\kappa}{s(w)} \right| - \mathrm{sgn} \left( t^{\ast}\left( \frac{(1-w)\kappa}{s(w)} \right) - \frac{2\kappa}{s(w)} \right) \cdot \frac{(1-w)\kappa}{s(w)} \right), \\
\overline{r}_{2}(w) & \equiv & s(w) \cdot \left| t^{\ast}\left( \frac{(1-w)\kappa}{s(w)} \right) - \frac{\kappa}{s(w)} \right| \cdot \Phi \left( - \left| t^{\ast}\left( \frac{(1-w)\kappa}{s(w)} \right) - \frac{\kappa}{s(w)} \right| \right).
\end{eqnarray*}
In addition, $L_{\mathrm{true}}(w) = \overline{R}_{\mathrm{upper}}(w)$ holds when $\kappa=0$.
\end{Lemma}
\begin{proof}
First, we consider the case where $K$ is even. Because $\sigma_1 = \cdots = \sigma_K$, we have $s(w) \equiv s_1(w) = \cdots = s_K(w)$ for any $w \in [0,1]$. For any $t \in \mathbf{R}$, we define 
$$
\bm{\theta}_{t,\kappa} \ \equiv \ (\underbrace{t+\kappa, \ldots, t+\kappa}_{\text{$K/2$ elements}}, \underbrace{t-\kappa, \ldots, t-\kappa}_{\text{$K/2$ elements}})'.
$$
Since $\bm{\theta}_{t,\kappa} \in \Theta(\kappa)$ for all $t \in \mathbb{R}$, we obtain
\begin{eqnarray*}
\overline{R}_{\mathrm{true}}(w) \ = \ \max_{\bm{\theta} \in \Theta(\kappa)} R(\bm{\theta},\hat{\bm{\delta}}^{(w,\ldots,w)}) & \geq & \max_{t \geq -\kappa} R(\bm{\theta}_{t,\kappa},\hat{\bm{\delta}}^{(w,\ldots,w)}).
\end{eqnarray*}
This implies that for any $t \geq -\kappa$ we obtain
\begin{eqnarray}
\overline{R}_{\mathrm{true}}(w) & \geq & \frac{1}{K} \sum_{k=1}^{K/2} (t+\kappa) \cdot \Phi \left( - \frac{(t+\kappa) - (1-w) \kappa}{s(w)} \right)  \nonumber \\
& &  \hspace{0.3in} + \frac{1}{K} \sum_{k=1}^{K/2} |t-\kappa| \cdot \Phi \left( - \frac{|t-\kappa| + (1-w) \cdot \mathrm{sgn}(t-\kappa)  \kappa}{s(w)} \right) \nonumber \\
&=& \frac{1}{2}  (t+\kappa) \cdot \Phi \left( - \frac{(t+\kappa) - (1-w) \kappa}{s(w)} \right) + \frac{1}{2} r_1(t;w,\kappa), \nonumber
\end{eqnarray}
where $r_1(t;w,\kappa) \equiv |t-\kappa| \cdot \Phi \left( - \frac{|t-\kappa| + (1-w) \cdot \mathrm{sgn}(t-\kappa)  \kappa}{s(w)} \right)$. Substituting $t = s(w) \cdot t^{\ast}\left( \frac{(1-w)\kappa}{s(w)} \right) - \kappa$, we obtain
\begin{eqnarray*}
\overline{R}_{\mathrm{true}}(w) & \geq & \frac{1}{2}  s(w)  \eta\left( \frac{(1-w) \kappa}{s(w)}  \right)  + \frac{1}{2} r_1\left( s(w) \cdot t^{\ast}\left( \frac{(1-w)\kappa}{s(w)} \right) - \kappa ;  w, \kappa \right)  \\
&=& \frac{1}{2} \overline{R}_{\mathrm{upper}}(w) + \frac{1}{2} \overline{r}_1(w).
\end{eqnarray*}

Next, we consider the case where $K$ is odd. Similar to $\bm{\theta}_{t,\kappa}$, we define 
$$
\tilde{\bm{\theta}}_{t,\kappa} \ \equiv \ (\underbrace{t+\kappa, \ldots, t+\kappa}_{\text{$(K-1)/2$ elements}}, t, \underbrace{t-\kappa, \ldots, t-\kappa}_{\text{$(K-1)/2$ elements}})'.
$$
From a similar argument as above, for any $t \geq - \kappa$ we obtain
\begin{eqnarray*}
\overline{R}_{\mathrm{true}}(w) & \geq & R(\tilde{\bm{\theta}}_{t,\kappa},\hat{\bm{\delta}}^{(w,\ldots, w)}) \\
& \geq & \frac{1}{K} \sum_{k=1}^{(K-1)/2} (t+\kappa) \cdot \Phi \left( - \frac{(t+\kappa) - (1-w) \kappa}{s(w)} \right) + \frac{1}{K} |t| \cdot \Phi \left( - \frac{|t|}{s(w)} \right)  \nonumber \\
& &  \hspace{0.3in} + \frac{1}{K} \sum_{k=1}^{(K-1)/2} |t-\kappa| \cdot \Phi \left( - \frac{|t-\kappa| + (1-w) \cdot \mathrm{sgn}(t-\kappa)  \kappa}{s(w)} \right) \nonumber \\
&=& \frac{K-1}{2K} (t+\kappa) \cdot \Phi \left( - \frac{(t+\kappa) - (1-w) \kappa}{s(w)} \right) + \frac{K-1}{2K} r_1(t;w,\kappa) + \frac{1}{K} r_2(t;w),
\end{eqnarray*}
where $r_2(t;w) \equiv |t| \cdot \Phi \left( - \frac{|t|}{s(w)} \right)$. Substituting $t = s(w) \cdot t^{\ast}\left( \frac{(1-w)\kappa}{s(w)} \right) - \kappa$, we obtain
\begin{eqnarray*}
\overline{R}_{\mathrm{true}}(w) & \geq & \frac{K-1}{2K} \overline{R}_{\mathrm{upper}}(w) + \frac{K-1}{2K} \overline{r}_1(w) + \frac{1}{K} \overline{r}_2(w).
\end{eqnarray*}

Finally, we show that $L_{\mathrm{true}}(w) = \overline{R}_{\mathrm{upper}}(w)$ holds when $\kappa=0$. When $\kappa=0$, we have
\begin{eqnarray*}
\overline{r}_1(w) \ = \ \overline{r}_2(w) \ = \ s(w) t^{\ast}(0) \Phi \left( -t^{\ast}(0) \right) \ = \ s(w) \eta(0).
\end{eqnarray*}
In addition, $\overline{R}_{\mathrm{upper}}(w) = s(w) \eta(0)$ holds when $\kappa=0$. Therefore, we obtain $L_{\mathrm{true}}(w) = \overline{R}_{\mathrm{upper}}(w)$ when $\kappa=0$.
\end{proof}\vspace{0.1in}

\begin{Lemma}\label{lem:CES_mis}
Suppose that Assumption \ref{ass:parameter_space} holds. Letting $\mathcal{K}(\kappa) \equiv \{k : \sigma_k \leq \underline{\sigma} + \kappa / t^{\ast}(0)\}$, we have
\begin{eqnarray}
\frac{\max_{\bm{\theta} \in \Theta(\kappa)} R(\bm{\theta},\hat{\bm{\delta}}^{\bm{w}^{\ast}(\kappa')})}{\max_{\bm{\theta} \in \Theta(\kappa)} R(\bm{\theta},\hat{\bm{\delta}}^{\mathrm{CES}})} & \leq & \frac{\sum_{k=1}^K p_k \cdot \psi_k\left( w_k^{\ast}(\kappa') ;\kappa' \right)}{\eta(0) \left( \sum_{k \in \mathcal{K}(\kappa)} p_k \sigma_k + \sum_{k \not\in \mathcal{K}(\kappa)} p_k \underline{\sigma} \right)} \nonumber \\
& &\times \left[ 1 + \max_{k} \left\{ H \left( \frac{(1-w_k^{\ast}(\kappa')) \cdot \kappa' }{s_k \left( w_k^{\ast}(\kappa') \right)} \right) \right\} \cdot \left( \frac{|\kappa - \kappa'|_{+}}{\kappa'} \right) \right], \label{Thm_CES_mis}
\end{eqnarray}
where $H(a) \equiv a/\eta(a)$. In particular, if $\sigma_1, \ldots, \sigma_K$ and $\kappa$ satisfy $\overline{\sigma} - \underline{\sigma} \leq \kappa / t^{\ast}(0)$, then
\[
\frac{\max_{\bm{\theta} \in \Theta(\kappa)} R(\bm{\theta},\hat{\bm{\delta}}^{\bm{w}^{\ast}(\kappa')})}{\max_{\bm{\theta} \in \Theta(\kappa)} R(\bm{\theta},\hat{\bm{\delta}}^{\mathrm{CES}})} \ \leq \ 1 + \max_{k} \left\{ H \left( \frac{(1-w_k^{\ast}(\kappa')) \cdot \kappa' }{s_k \left( w_k^{\ast}(\kappa') \right)} \right) \right\} \cdot \left( \frac{|\kappa - \kappa'|_{+}}{\kappa'} \right).
\]
\end{Lemma}
\begin{proof}
Because $\hat{\bm{\delta}}^{\bm{1}}(\hat{\bm{\theta}}) = \hat{\bm{\delta}}^{\text{CES}}(\hat{\bm{\theta}})$, it follows from (\ref{regret}) that we have
\begin{eqnarray*}
R(\bm{\theta},\hat{\bm{\delta}}^{\text{CES}}) &=&  \sum_{k=1}^K p_k \cdot \left\{ |\theta_k| \cdot  \Phi \left( -   \frac{|\theta_k|}{\sigma_k} \right) \right\} \\
&=& \sum_{k=1}^K p_k \cdot \left\{ \sigma_k \cdot \left( \frac{|\theta_k|}{\sigma_k} \right) \cdot  \Phi \left( -   \frac{|\theta_k|}{\sigma_k} \right) \right\}.
\end{eqnarray*}
For any $t$, a hyper-rectangle $[t,t+\kappa]^K$ is included in $\Theta(\kappa)$. Hence, for any $t>0$, we obtain
\begin{eqnarray*}
\max_{\bm{\theta} \in \Theta(\kappa)} R(\bm{\theta},\hat{\bm{\delta}}^{\text{CES}}) & \geq & \max_{\bm{\theta} \in [t,t+\kappa]^K} R(\bm{\theta},\hat{\bm{\delta}}^{\text{CES}}) \ = \ \sum_{k=1}^K p_k \sigma_k \cdot \max_{\theta_k \in [t,t+\kappa]} \left\{ (\theta_k / \sigma_k) \Phi(-\theta_k / \sigma_k) \right\} \\
&=&  \sum_{k=1}^K p_k \sigma_k \cdot \max_{s \in [t/\sigma_k,(t+\kappa)/\sigma_k]} \left\{ s \Phi(-s) \right\}.
\end{eqnarray*}
When $t = t^{\ast}(0) \cdot \underline{\sigma}$, we have $t/\sigma_k \leq t^{\ast}(0)$. Because $s \Phi(-s)$ is single-peaked and maximized at $s = t^{\ast}(0)$, for $t = t^{\ast}(0) \cdot \underline{\sigma}$ we obtain
\begin{eqnarray*}
\max_{s \in [t/\sigma_k,(t+\kappa)/\sigma_k]} \left\{ s \Phi(-s) \right\} &=& \min \left\{ t^{\ast}(0), \frac{t^{\ast}(0)\underline{\sigma}+\kappa}{\sigma_k} \right\} \cdot \Phi \left( - \min \left\{ t^{\ast}(0), \frac{t^{\ast}(0)\underline{\sigma}+\kappa}{\sigma_k} \right\} \right) \\
& \geq &  \min \left\{ t^{\ast}(0), \frac{t^{\ast}(0)\underline{\sigma}+\kappa}{\sigma_k} \right\} \cdot \Phi(-t^{\ast}(0)) \\
& \geq & \min \left\{ \eta(0), \frac{\eta(0) \underline{\sigma} + \kappa \Phi(-t^{\ast}(0))}{\sigma_k} \right\}.
\end{eqnarray*}
Because $\eta(0) = t^{\ast}(0) \Phi(-t^{\ast}(0))$, we obtain
\begin{eqnarray}
\max_{\bm{\theta} \in \Theta(\kappa)} R(\bm{\theta},\hat{\bm{\delta}}^{\text{CES}}) & \geq & \sum_{k=1}^K p_k \cdot \min \left\{ \eta(0) \sigma_k,  \eta(0) \underline{\sigma} + \kappa \frac{\eta(0)}{t^{\ast}(0)} \right\} \nonumber \\
&=& \eta(0) \cdot \left[ \sum_{k=1}^K p_k \cdot \left\{ \sigma_k - \left| (\sigma_k - \underline{\sigma}) - \frac{\kappa}{t^{\ast}(0)} \right|_{+} \right\}  \right] \nonumber \\
& \geq & \eta(0) \cdot \left( \sum_{k \in \mathcal{K}(\kappa)} p_k \sigma_k + \sum_{k \not\in \mathcal{K}(\kappa)} p_k  \underline{\sigma}  \right), \label{CES_lower}
\end{eqnarray}
where $|a|_{+} \equiv \max\{0,a\}$. Hence, if $\overline{\sigma} - \underline{\sigma} \leq \kappa / t^{\ast}(0)$, the lower bound of the maximum regret of the CES rule becomes $\eta(0) \left( \sum_{k=1}^K p_k \sigma_k \right)$.

Next, we derive the upper bound of $\max_{\bm{\theta} \in \Theta(\kappa)} R(\bm{\theta},\hat{\bm{\delta}}^{\bm{w}^*(\kappa')})$. We observe that
\begin{eqnarray*}
\max_{\bm{\theta} \in \Theta(\kappa)} R(\bm{\theta},\hat{\bm{\delta}}^{\bm{w}^*(\kappa')}) &\leq & \sum_{k=1}^K p_k \cdot \psi_k\left( w_k^{\ast}(\kappa') ;\kappa \right)  \\
&=& \sum_{k=1}^K p_k \cdot \psi_k\left( w_k^{\ast}(\kappa') ;\kappa' \right) \cdot \left( \frac{\psi_k\left( w_k^{\ast}(\kappa') ;\kappa \right)}{\psi_k\left( w_k^{\ast}(\kappa') ;\kappa' \right)} \right) \\
&=& \sum_{k=1}^K p_k \cdot \psi_k\left( w_k^{\ast}(\kappa') ;\kappa' \right) \cdot \left\{ \frac{\eta \left( \frac{(1-w_k^{\ast}(\kappa')) \cdot \kappa}{s_k \left( w_k^{\ast}(\kappa') \right)} \right)}{\eta \left( \frac{(1-w_k^{\ast}(\kappa')) \cdot \kappa'}{s_k \left( w_k^{\ast}(\kappa') \right)} \right)} \right\}.
\end{eqnarray*}
This implies that $\max_{\bm{\theta} \in \Theta(\kappa)} R(\bm{\theta},\hat{\bm{\delta}}^{\bm{w}^*(\kappa')}) \leq \sum_{k=1}^K p_k \cdot \psi_k\left( w_k^{\ast}(\kappa') ;\kappa' \right)$ when $\kappa \leq \kappa'$. Because $\eta'(a) \leq 1$ for all $a \geq 0$, we have $\eta(a) \leq \eta(a') + (a-a')$ for $a \geq a'$. Hence, when $\kappa \geq \kappa'$, we obtain
\begin{eqnarray}
& & \max_{\bm{\theta} \in \Theta(\kappa)} R(\bm{\theta},\hat{\bm{\delta}}^{\bm{w}^*(\kappa')}) \nonumber \\
&\leq & \sum_{k=1}^K p_k \cdot \psi_k\left( w_k^{\ast}(\kappa') ;\kappa' \right) \cdot \left\{ 1 + \frac{ \left( \frac{(1-w_k^{\ast}(\kappa'))\cdot \kappa' }{s_k \left( w_k^{\ast}(\kappa') \right)} \right) \cdot \frac{\kappa - \kappa'}{\kappa'}}{\eta \left( \frac{(1-w_k^{\ast}(\kappa')) \cdot \kappa'}{s_k \left( w_k^{\ast}(\kappa') \right)} \right)} \right\} \nonumber \\
&=& \sum_{k=1}^K p_k \cdot \psi_k\left( w_k^{\ast}(\kappa') ;\kappa' \right) \cdot \left\{ 1 + H \left( \frac{(1-w_k^{\ast}(\kappa')) \cdot \kappa' }{s_k \left( w_k^{\ast}(\kappa') \right)} \right) \cdot \frac{\kappa - \kappa'}{\kappa'} \right\}, \nonumber
\end{eqnarray}
where $H(a) \equiv a/\eta(a)$. Hence, we obtain the following upper bound:
\begin{eqnarray}
& & \max_{\bm{\theta} \in \Theta(\kappa)} R(\bm{\theta},\hat{\bm{\delta}}^{\bm{w}^*(\kappa')}) \nonumber \\
&\leq & \left\{ \sum_{k=1}^K p_k \cdot \psi_k\left( w_k^{\ast}(\kappa') ;\kappa' \right) \right\} \cdot \left[ 1 + \max_{k} \left\{ H \left( \frac{(1-w_k^{\ast}(\kappa')) \cdot \kappa' }{s_k \left( w_k^{\ast}(\kappa') \right)} \right) \right\} \cdot \frac{|\kappa - \kappa'|_{+}}{\kappa'} \right]. \label{shrinkage_upper}
\end{eqnarray}
From (\ref{CES_lower}) and (\ref{shrinkage_upper}), we obtain the upper bound of (\ref{Thm_CES_mis}).
\end{proof}\vspace{0.1in}

\begin{Lemma}\label{lem:pool_mis}
Suppose that Assumption \ref{ass:parameter_space} holds and $K$ is even. Then, we have
\begin{eqnarray}
\frac{\max_{\bm{\theta} \in \Theta(\kappa)} R(\bm{\theta},\hat{\bm{\delta}}^{\bm{w}^{\ast}(\kappa')})}{\max_{\bm{\theta} \in \Theta(\kappa)} R(\bm{\theta},\hat{\bm{\delta}}^{\mathrm{pool}})} & \leq & \frac{\sum_{k=1}^K p_k \cdot \psi_k\left( w_k^{\ast}(\kappa') ;\kappa' \right)}{\max \left\{ \frac{1}{2} s_0 \eta ( \kappa / s_0 ), s_0 \eta ( \kappa / s_0 ) - \kappa \right\}} \nonumber \\
& & \times  \left[ 1 + \max_{k} \left\{ H \left( \frac{(1-w_k^{\ast}(\kappa')) \cdot \kappa' }{s_k \left( w_k^{\ast}(\kappa') \right)} \right) \right\} \cdot \left( \frac{|\kappa - \kappa'|_{+}}{\kappa'} \right) \right]. \label{Thm_pool_mis}
\end{eqnarray}
\end{Lemma}
\begin{proof}
Without loss of generality, we assume $p_1 \geq \cdots \geq p_K$. Because $\hat{\bm{\delta}}^{\mathrm{pool}}(\hat{\bm{\theta}}) = \hat{\bm{\delta}}^{\bm{0}}(\hat{\bm{\theta}})$ and $(t+\kappa, \ldots, t+\kappa, t-\kappa, \ldots, t-\kappa)' \in \Theta(\kappa)$ for all $t \in \mathbb{R}$, we observe that
\begin{eqnarray*}
\max_{\bm{\theta} \in \Theta(\kappa)} R(\bm{\theta},\hat{\bm{\delta}}^{\mathrm{pool}}) &=& \max_{\bm{\theta} \in \Theta(\kappa)} \sum_{k=1}^K p_k \cdot \left\{ |\theta_k| \cdot \Phi \left( - \text{sgn}(\theta_k) \cdot \frac{\overline{\theta}}{s_0} \right) \right\} \\
& \geq & \max_{t \geq -\kappa} \left[ \sum_{k=1}^{K/2} p_k \cdot \left\{ (t+\kappa) \cdot \Phi \left( - \frac{(t+\kappa)-\kappa}{s_0} \right) \right\} \right. \\
& & \hspace{0.5in} \left. + \sum_{k=K/2 + 1}^{K} p_k \cdot \left\{ |t-\kappa| \cdot \Phi \left( - \text{sgn}(t-\kappa) \cdot \frac{t}{s_0} \right) \right\} \right].
\end{eqnarray*}
Substituting $t = s_0 \cdot t^{\ast}(\kappa / s_0) - \kappa$, we obtain
\begin{equation}
\max_{\bm{\theta} \in \Theta(\kappa)} R(\bm{\theta},\hat{\bm{\delta}}^{\mathrm{pool}}) \ \geq \ \left( \sum_{k=1}^{K/2} p_k \right) \cdot s_0 \eta \left( \kappa / s_0 \right) \ \geq \ \frac{1}{2} s_0 \eta \left( \kappa / s_0 \right), \label{pool_lower1}
\end{equation}
where $\sum_{k=1}^{K/2} p_k \geq 1/2$ because $p_1 \geq \cdots \geq p_K$.

Next, we derive another lower bound of $\max_{\bm{\theta} \in \Theta(\kappa)} R(\bm{\theta},\hat{\bm{\delta}}^{\mathrm{pool}})$. Similar to the above argument, we observe that
\begin{eqnarray*}
\max_{\bm{\theta} \in \Theta(\kappa)} R(\bm{\theta},\hat{\bm{\delta}}^{\mathrm{pool}}) & \geq & \max_{t \geq -\kappa} \left[ \sum_{k=1}^{K/2} p_k \cdot \left\{ (t+\kappa) \cdot \Phi \left( - \frac{(t+\kappa)-\kappa}{s_0} \right) \right\} \right. \\
& & \hspace{0.5in} \left. + \sum_{k=K/2 + 1}^{K} p_k \cdot \left\{ |t-\kappa| \cdot \Phi \left( - \text{sgn}(t-\kappa) \cdot \frac{t}{s_0} \right) \right\} \right] \\
& \geq & \max_{t \geq -\kappa} \left[ \sum_{k=1}^{K/2} p_k \cdot \left\{ (t+\kappa) \cdot \Phi \left( - \frac{(t+\kappa)-\kappa}{s_0} \right) \right\} \right. \\
& & \hspace{1.2in} \left. + \sum_{k=K/2 + 1}^{K} p_k \cdot \left\{ (t-\kappa) \cdot \Phi \left( - \frac{t}{s_0} \right) \right\} \right].
\end{eqnarray*}
Substituting $t = s_0 \cdot t^{\ast}(\kappa / s_0) - \kappa$, we obtain
\begin{eqnarray}
& & \max_{\bm{\theta} \in \Theta(\kappa)} R(\bm{\theta},\hat{\bm{\delta}}^{\mathrm{pool}}) \nonumber \\
& \geq &  \sum_{k=1}^{K/2} p_k \cdot s_0 \eta \left( \frac{\kappa}{s_0} \right) + \sum_{k=K/2 + 1}^{K} p_k \cdot s_0 \left\{ \left( t^{\ast}\left( \frac{\kappa}{s_0} \right)  - \frac{2 \kappa}{s_0} \right) \cdot \Phi \left( -  t^{\ast}\left( \frac{\kappa}{s_0} \right) + \frac{\kappa}{s_0}  \right) \right\} \nonumber \\
& \geq & s_0 \eta \left( \frac{\kappa}{s_0} \right) - \left( \sum_{k=K/2 + 1}^{K} p_k \right) \cdot 2 \kappa \cdot \Phi \left( -  t^{\ast}\left( \frac{\kappa}{s_0} \right) + \frac{\kappa}{s_0}  \right) \ \geq \ s_0 \eta(\kappa / s_0) - \kappa. \label{pool_lower2}
\end{eqnarray}
From (\ref{shrinkage_upper}), (\ref{pool_lower1}), and (\ref{pool_lower2}), we obtain the upper bound of (\ref{Thm_pool_mis}).
\end{proof}\vspace{0.1in}

\begin{proof}[Proof of Proposition \ref{prop:tightness}]
Because the lower bound is trivial, we consider the upper bound of (\ref{tight_upper}). Lemma \ref{lem:prop1} implies that we obtain
\begin{equation*}
\frac{\overline{R}_{\mathrm{upper}}(w)}{\overline{R}_{\mathrm{true}}(w)} \ \leq \ \frac{\overline{R}_{\mathrm{upper}}(w)}{L_{\mathrm{true}}(w)}.
\end{equation*}
Because $\overline{r}_1(w)$ and $\overline{r}_2(w)$ are nonnegative, we have $L_{\mathrm{true}}(w) \geq \frac{1}{2} \overline{R}_{\mathrm{upper}}(w)$ when $K$ is even and $L_{\mathrm{true}}(w) \geq \frac{K-1}{2K} \overline{R}_{\mathrm{upper}}(w)$ when $K$ is odd. Hence, we obtain (\ref{tight_upper}). Furthermore, Lemma \ref{lem:prop1} shows that $L_{\mathrm{true}}(w) = \overline{R}_{\mathrm{upper}}(w)$ holds when $\kappa = 0$. This implies that $\overline{R}_{\mathrm{upper}}(w) = \overline{R}_{\mathrm{true}}(w)$ when $\kappa = 0$.
\end{proof}\vspace{0.1in}

\begin{proof}[Proof of Proposition \ref{prop:K_asymptotics}]
Using Lemma \ref{lem:eta_derivative}, we obtain
\begin{eqnarray}
\frac{d}{d w_k} \tilde{\psi}_k(w_k;\kappa) &=& \sigma_k \eta \left( (w_k^{-1} - 1) \cdot (\kappa/\sigma_k) \right) \nonumber \\
& & - \kappa w_k^{-1} \Phi\left(-t^{\ast}((w_k^{-1} - 1) \cdot (\kappa/\sigma_k) ) + (w_k^{-1} - 1) \cdot (\kappa / \sigma_k ) \right). \nonumber
\end{eqnarray}
Because we have $\eta(a) = t^{\ast}(a) \cdot \Phi(-t^{\ast}(a)+a)$, we have
\begin{eqnarray*}
\frac{d}{d w_k} \tilde{\psi}_k(w_k;\kappa) &=& \left\{ \sigma_k  - \frac{\kappa w_k^{-1}}{t^{\ast}((w_k^{-1} - 1) \cdot (\kappa/\sigma_k) )} \right\} \eta \left( (w_k^{-1} - 1) \cdot (\kappa/\sigma_k) \right)
\end{eqnarray*}
Hence, we obtain
\begin{equation}
\mathrm{sgn} \left( \frac{d}{d w_k} \tilde{\psi}_k(w_k ; \kappa) \right) \ = \ \mathrm{sgn} \left( t^{\ast}\left( (w_k^{-1} - 1) \cdot (\kappa/\sigma_k) \right) - (\kappa/\sigma_k) w_k^{-1} \right).\label{sign_tilde_psi}
\end{equation}
Because $t^{\ast}(0) > 0$, the sign of the derivative becomes strictly positive when $\kappa=0$. Hence, we obtain $\tilde{w}_k^{\ast}(0) = 0$.

Because $-t^{\ast}(a) + a + t^{\ast}(0)$ is strictly positive for $a > 0$ by Lemma \ref{lem:eta_derivative}, we have
$$
t^{\ast}\left( (w_k^{-1} - 1) \cdot (\kappa/\sigma_k) \right) - (w_k^{-1} - 1) \cdot (\kappa/\sigma_k) - t^{\ast}(0) \ < \ 0 \ \ \text{for all $w_k \in (0,1)$.}
$$
This implies that for all $w_k \in (0,1)$, we have
\begin{eqnarray*}
 t^{\ast}\left( (w_k^{-1} - 1) \cdot (\kappa/\sigma_k) \right) - (\kappa/\sigma_k) w_k^{-1} & < & t^{\ast}(0) - \kappa/ \sigma_k.
\end{eqnarray*}
Therefore, if $\kappa/\sigma_k > t^{\ast}(0)$, then $\frac{d}{d w_k} \tilde{\psi}_k(w_k)$ is negative for all $w_k \in (0,1)$. As a result, if $\kappa/\sigma_k > t^{\ast}(0)$, then we have $\tilde{w}_k^{\ast}(\kappa) = 1$.
\end{proof}\vspace{0.1in}

\begin{proof}[Proof of Theorem \ref{thm:CES_correct}]
From Lemma \ref{lem:CES_mis}, we observe that
\begin{equation}
\frac{\max_{\bm{\theta} \in \Theta(\kappa)} R(\bm{\theta},\hat{\bm{\delta}}^{\bm{w}^{\ast}(\kappa)})}{\max_{\bm{\theta} \in \Theta(\kappa)} R(\bm{\theta},\hat{\bm{\delta}}^{\mathrm{CES}})} \ \leq \ \frac{\sum_{k=1}^K p_k \cdot \psi_k\left( w_k^{\ast}(\kappa) ;\kappa \right)}{\eta(0) \left( \sum_{k \in \mathcal{K}(\kappa)} p_k \sigma_k + \sum_{k \not\in \mathcal{K}(\kappa)} p_k \underline{\sigma} \right)}. \label{Thm_CES_correct}
\end{equation}
When $t^{\ast}(0) \cdot (\overline{\sigma} - \underline{\sigma}) \leq \kappa$ holds, the denominator on the right-hand side of (\ref{Thm_CES_correct}) becomes $\eta(0) \left( \sum_{k=1}^K p_k \sigma_k \right)$. Because $\psi_k\left( w_k^{\ast}(\kappa) ;\kappa \right) \leq \psi_k (1;\kappa) = \eta(0) \sigma_k$, the right-hand side of (\ref{Thm_CES_correct}) is not larger than $1$ when $t^{\ast}(0) \cdot (\overline{\sigma} - \underline{\sigma}) \leq \kappa$.

As discussed in Remark \ref{rem:shrinkage_simple}, we observe that
\begin{eqnarray*}
\psi_k'(w_k;\kappa) &=& s_k'(w_k) \eta \left( \frac{(1-w_k) \cdot \kappa}{s_k(w_k)} \right) - \left\{ \frac{\kappa s_k(w_k) + \kappa (1-w_k) s_k'(w_k)}{s_k(w_k)} \right\} \eta'\left( \frac{(1-w_k) \cdot \kappa}{s_k(w_k)} \right),
\end{eqnarray*}
which implies that $\psi_k'(1;\kappa) = s_k'(w_k) \eta(0) - \kappa \eta'(0)$. Hence, if $\kappa < \frac{s_k'(w_k) \eta(0)}{\eta'(0)}$ holds, we have $\psi_k\left( w_k^{\ast}(\kappa) ;\kappa \right) < \psi_k (1;\kappa)$. Lemma \ref{lem:eta_derivative} implies $\frac{\eta(0)}{\eta'(0)} = \frac{t^{\ast}(0) \cdot \Phi(-t^{\ast}(0))}{\Phi(-t^{\ast}(0))} = t^{\ast}(0)$ and we observe that
\begin{eqnarray*}
s_k'(1) &=& \left. \left\{ \sqrt{s_k^2(w_k)} \right\}' \right|_{w_k=1} \ = \ \left. \frac{1}{2} \frac{\{s_k^2(w_k)\}'}{s_k(w_k)} \right|_{w_k=1}\\
&=& \left. \frac{\left\{ 2 w_k + \frac{2}{K}(1-2w_k) \right\}\sigma_k^2 - 2(1-w_k) \left( \frac{1}{K^2} \sum_{k=1}^K \sigma_k^2 \right) }{2 s_k(w_k)} \right|_{w_k=1} \ = \ \left( 1 - \frac{1}{K} \right) \sigma_k.
\end{eqnarray*}
Hence, when $t^{\ast}(0) \cdot (\overline{\sigma} - \underline{\sigma}) \leq \kappa < t^{\ast}(0) \cdot \left( 1 - \frac{1}{K} \right) \sigma_k$ holds for some $k$, the right-hand side of (\ref{Thm_CES_correct}) is strictly smaller than $1$. Therefore, if $t^{\ast}(0) \cdot (\overline{\sigma} - \underline{\sigma}) \leq \kappa < t^{\ast}(0) \cdot \left( 1-\frac{1}{K} \right) \overline{\sigma}$ holds, then we obtain $\max_{\bm{\theta} \in \Theta(\kappa)} R(\bm{\theta},\hat{\bm{\delta}}^{\bm{w}^{\ast}(\kappa)}) < \max_{\bm{\theta} \in \Theta(\kappa)} R(\bm{\theta},\hat{\bm{\delta}}^{\mathrm{CES}})$.

Because $\psi_k\left( w_k^{\ast}(\kappa) ;\kappa \right) \leq \psi_k\left( 0 ;\kappa \right) = s_0 \eta(\kappa / s_0)$, the right-hand side of (\ref{Thm_CES_correct}) is bounded by
\begin{eqnarray*}
\frac{s_0 \cdot \eta \left( \kappa / s_0 \right)}{\eta(0) \underline{\sigma}} \ \leq \ \frac{s_0 \{ \eta(0) + \frac{\kappa}{s_0} \}}{\eta(0) \underline{\sigma}} \ = \ \frac{ \eta(0) s_0 + \kappa}{\eta(0) \underline{\sigma}},
\end{eqnarray*}
where the first inequality follows from Lemma \ref{lem:eta}. Hence, $\kappa \leq \eta(0) \cdot (\underline{\sigma} - s_0)$ implies 
$$
\max_{\bm{\theta} \in \Theta(\kappa)} R(\bm{\theta},\hat{\bm{\delta}}^{\bm{w}^{\ast}(\kappa)}) \leq \max_{\bm{\theta} \in \Theta(\kappa)} R(\bm{\theta},\hat{\bm{\delta}}^{\mathrm{CES}})
$$ and $\kappa < \eta(0) \cdot (\underline{\sigma} - s_0)$ implies 
$$
\max_{\bm{\theta} \in \Theta(\kappa)} R(\bm{\theta},\hat{\bm{\delta}}^{\bm{w}^{\ast}(\kappa)}) < \max_{\bm{\theta} \in \Theta(\kappa)} R(\bm{\theta},\hat{\bm{\delta}}^{\mathrm{CES}}).
$$
As a result, we obtain the desired results.
\end{proof}\vspace{0.1in}

\begin{proof}[Proof of Theorem \ref{thm:pool_correct}]
From the proof of Lemma \ref{lem:pool_mis}, we obtain
\begin{eqnarray}
\frac{\max_{\bm{\theta} \in \Theta(\kappa)} R(\bm{\theta},\hat{\bm{\delta}}^{\bm{w}^{\ast}(\kappa)})}{\max_{\bm{\theta} \in \Theta(\kappa)} R(\bm{\theta},\hat{\bm{\delta}}^{\mathrm{pool}})} & \leq & \frac{\sum_{k=1}^K p_k \cdot \psi_k\left( w_k^{\ast}(\kappa) ;\kappa \right)}{\max \left\{ \frac{1}{2} s_0 \eta ( \kappa / s_0 ), s_0 \eta ( \kappa / s_0 ) - \kappa \right\}} \nonumber \\
& \leq & \min \left\{ 2, \, \frac{s_0 \eta(\kappa / s_0)}{s_0 \eta(\kappa / s_0) - \kappa}, \, \frac{2\eta(0) \left( \sum_{k=1}^K p_k \sigma_k \right)}{s_0 \eta(\kappa / s_0)} \right\}, \label{Thm_pool_correct}
\end{eqnarray}
where the second inequality follows from 
\[
\psi_k\left( w_k^{\ast}(\kappa) ;\kappa \right) \leq \psi_k\left(0 ;\kappa \right) = s_0 \eta(\kappa / s_0) \ \ \text{and} \ \ \psi_k\left( w_k^{\ast}(\kappa) ;\kappa \right) \leq \psi_k\left(1 ;\kappa \right) = \sigma_k \eta(0).
\]
Because the right-hand side of (\ref{Thm_pool_correct}) is bounded by $2$, we have
\[
\max_{\bm{\theta} \in \Theta(\kappa)} R(\bm{\theta},\hat{\bm{\delta}}^{\bm{w}^{\ast}(\kappa)}) \ \leq \ 2 \cdot \max_{\bm{\theta} \in \Theta(\kappa)} R(\bm{\theta},\hat{\bm{\delta}}^{\mathrm{pool}}) \ \ \ \text{for all $\kappa$}.
\]
When $\kappa = 0$, we have
\[
\frac{s_0 \eta(\kappa / s_0)}{s_0 \eta(\kappa / s_0) - \kappa} \ = \ \frac{s_0 \eta(0)}{s_0 \eta(0)} \ = \ 1.
\]
Hence, we have $\max_{\bm{\theta} \in \Theta(\kappa)} R(\bm{\theta},\hat{\bm{\delta}}^{\bm{w}^{\ast}(\kappa)}) \leq \max_{\bm{\theta} \in \Theta(\kappa)} R(\bm{\theta},\hat{\bm{\delta}}^{\mathrm{pool}})$ for $\kappa=0$. In addition, if $s_0 \eta(\kappa/s_0) > 2\eta(0) \left( \sum_{k=1}^K p_k \sigma_k \right)$ holds, the right-hand side of (\ref{Thm_pool_correct}) is smaller than $1$. Hence, we obtain $\max_{\bm{\theta} \in \Theta(\kappa)} R(\bm{\theta},\hat{\bm{\delta}}^{\bm{w}^{\ast}(\kappa)}) < \max_{\bm{\theta} \in \Theta(\kappa)} R(\bm{\theta},\hat{\bm{\delta}}^{\mathrm{pool}})$.
\end{proof}\vspace{0.1in}

\begin{proof}[Proof of Corollary \ref{cor:dominate}]
Because $\overline{\sigma} < K \underline{\sigma}$, we have $t^{\ast}(0) \cdot (\overline{\sigma} - \underline{\sigma}) \ < \ t^{\ast}(0) \cdot \left( 1-\frac{1}{K} \right) \overline{\sigma}$. From Lemma \ref{lem:eta}, we have $s_0 \eta(\kappa / s_0) \geq \kappa /2$. This implies that the condition 
$$
s_0 \eta(\kappa/s_0) \ > \ 2\eta(0) \left( \sum_{k=1}^K p_k \sigma_k \right)
$$
holds when
\begin{equation}
\kappa \ > \ 4 \eta(0) \left( \sum_{k=1}^K p_k \sigma_k \right). \nonumber
\end{equation}
We observe that
\begin{eqnarray*}
\frac{\sum_{k=1}^K p_k \sigma_k}{\underline{\sigma}} \ < \ \frac{t^{\ast}(0) \left( 1- \frac{1}{K} \right)}{4 \eta(0)} & \Leftrightarrow & 4 \eta(0) \left( \sum_{k=1}^K p_k \sigma_k \right) \ < \ t^{\ast}(0) \cdot \left( 1-\frac{1}{K} \right).
\end{eqnarray*}
Therefore, it follows from Theorems \ref{thm:CES_correct} and \ref{thm:pool_correct} that the proposed shrinkage rule strictly dominates the CES and pooling rules when $\kappa$ satisfies (\ref{range_dominate}).
\end{proof}\vspace{0.1in}

\begin{proof}[Proof of Theorem \ref{thm:CES_mis}]
First, we consider the case where $\kappa' =(1+c)\kappa$. Because $\kappa \leq \kappa'$, it follows from Lemma \ref{lem:CES_mis} that we have
\begin{eqnarray}
\frac{\max_{\bm{\theta} \in \Theta(\kappa)} R(\bm{\theta},\hat{\bm{\delta}}^{\bm{w}^{\ast}(\kappa')})}{\max_{\bm{\theta} \in \Theta(\kappa)} R(\bm{\theta},\hat{\bm{\delta}}^{\mathrm{CES}})} & \leq & \frac{\sum_{k=1}^K p_k \cdot \psi_k\left( w_k^{\ast}(\kappa') ;\kappa' \right)}{\eta(0) \left( \sum_{k \in \mathcal{K}(\kappa)} p_k \sigma_k + \sum_{k \not\in \mathcal{K}(\kappa)} p_k \underline{\sigma} \right)}. \nonumber
\end{eqnarray}
As discussed in the proof of Theorem \ref{thm:CES_correct}, the right-hand side is less than $1$ when $\kappa' = (1+c)\kappa < \eta(0) \cdot (\underline{\sigma} - s_0)$. Similarly, the proof of Theorem \ref{thm:CES_correct} implies that we have
\[
\sum_{k=1}^K p_k \cdot \psi_k\left( w_k^{\ast}(\kappa') ;\kappa' \right) \ < \ \sum_{k=1}^K p_k \cdot \psi_k\left( 1 ;\kappa' \right) \ = \ \eta(0) \left( \sum_{k=1}^K p_k \sigma_k \right)
\]
when $\kappa' = (1+c) \kappa < t^{\ast}(0) \cdot \left( 1 - \frac{1}{K} \right) \overline{\sigma}$. If $t^{\ast}(0) \cdot (\overline{\sigma} - \underline{\sigma}) \leq \kappa$, then we have
\[
\eta(0) \left( \sum_{k \in \mathcal{K}(\kappa)} p_k \sigma_k + \sum_{k \not\in \mathcal{K}(\kappa)} p_k \underline{\sigma} \right) \ = \ \eta(0) \left( \sum_{k=1}^K p_k \sigma_k \right).
\]
As a result, (\ref{ineq_CES_mis}) holds when $t^{\ast}(0) \cdot (\overline{\sigma} - \underline{\sigma}) \leq \kappa < \frac{t^{\ast}(0) \cdot \left(1- \frac{1}{K}\right) \overline{\sigma}}{1+c}$ or $\kappa <  \frac{\eta(0) \cdot (\underline{\sigma} - s_0)}{1+c}$. 

Next, we consider the case where $\kappa =(1+c)\kappa'$. Because Lemma \ref{lem:eta} implies $\eta(a) \geq \frac{a}{2}$, we have $H(a) \leq 2$ for all $a>0$. From (\ref{Thm_CES_mis}) and the proof of Theorem \ref{thm:CES_correct}, we obtain
\begin{eqnarray}
\frac{\max_{\bm{\theta} \in \Theta(\kappa)} R(\bm{\theta},\hat{\bm{\delta}}^{\bm{w}^{\ast}(\kappa')})}{\max_{\bm{\theta} \in \Theta(\kappa)} R(\bm{\theta},\hat{\bm{\delta}}^{\mathrm{CES}})} & \leq & \frac{\sum_{k=1}^K p_k \cdot \psi_k\left( w_k^{\ast}(\kappa') ;\kappa' \right)}{\eta(0) \left( \sum_{k \in \mathcal{K}(\kappa)} p_k \sigma_k + \sum_{k \not\in \mathcal{K}(\kappa)} p_k \underline{\sigma} \right)} \times (1+2c). \nonumber \\
& \leq & \frac{s_0 \eta( \kappa' / s_0 ) (1+2c)}{\eta(0) \underline{\sigma}} \ \leq \ \frac{(\eta(0)s_0 + \kappa')(1+2c)}{\eta(0) \underline{\sigma}}.
\end{eqnarray}
Hence, the right-hand side is less than $1$ when $\kappa' = \frac{\kappa}{1+c} < \eta(0) \cdot \left( \frac{\underline{\sigma}}{1+2c} - s_0 \right)$. Therefore, (\ref{ineq_CES_mis}) holds when $\kappa < \eta(0) \cdot \left\{ \left( \frac{1+c}{1+2c}\right)  \underline{\sigma} - (1+c)s_0 \right\}$.
\end{proof}\vspace{0.1in}

\begin{proof}[Proof of Theorem \ref{thm:pool_mis}]
First, we consider the case where $\kappa' =(1+c)\kappa$. Because $\kappa \leq \kappa'$, it follows from Lemma \ref{lem:pool_mis} that we have
\begin{eqnarray}
\frac{\max_{\bm{\theta} \in \Theta(\kappa)} R(\bm{\theta},\hat{\bm{\delta}}^{\bm{w}^{\ast}(\kappa')})}{\max_{\bm{\theta} \in \Theta(\kappa)} R(\bm{\theta},\hat{\bm{\delta}}^{\mathrm{pool}})} & \leq & \frac{\sum_{k=1}^K p_k \cdot \psi_k\left( w_k^{\ast}(\kappa') ;\kappa' \right)}{\max \left\{ \frac{1}{2} s_0 \eta ( \kappa / s_0 ), s_0 \eta ( \kappa / s_0 ) - \kappa \right\}} \nonumber \\
& \leq & \min \left\{ \frac{2 \eta(\kappa' / s_0)}{\eta(\kappa / s_0)}, \, \frac{s_0 \eta(\kappa' / s_0)}{s_0 \eta(\kappa / s_0) - \kappa}, \, \frac{2\eta(0) \left( \sum_{k=1}^K p_k \sigma_k \right)}{s_0 \eta(\kappa / s_0)} \right\}, \nonumber
\end{eqnarray}
where the second inequality follows from $\psi_k\left( w_k^{\ast}(\kappa') ;\kappa' \right) \leq \min \left\{ s_0 \eta(\kappa' / s_0), \, \eta(0) \sigma_k \right\}$. Hence, we obtain (\ref{ineq_CES_mis}) when $s_0 \eta(\kappa / s_0) > 2\eta(0) \left( \sum_{k=1}^K p_k \sigma_k \right)$. In addition, because $\eta(a) \leq \eta(a') + (a-a')$ for $a \geq  a'$, we obtain
\begin{eqnarray*}
\frac{2 \eta(\kappa' / s_0)}{\eta(\kappa / s_0)} & \leq & \frac{2 \left\{ \eta(\kappa / s_0) + (\kappa' - \kappa)/s_0 \right\} }{\eta(\kappa / s_0)} \ \leq \ 2 \left\{ 1 + H(\kappa / s_0) \left( \frac{\kappa' - \kappa}{\kappa} \right) \right\} \ \leq \ 2(1+2c),
\end{eqnarray*}
where the last inequality follows from $H(a) \leq 2$. Therefore, we obtain
\[
\max_{\bm{\theta} \in \Theta(\kappa)} R(\bm{\theta},\hat{\bm{\delta}}^{\bm{w}^{\ast}(\kappa')}) \ \leq \ 2(1+2c) \cdot \max_{\bm{\theta} \in \Theta(\kappa)} R(\bm{\theta},\hat{\bm{\delta}}^{\mathrm{pool}}).
\]

Next, we consider the case where $\kappa =(1+c)\kappa'$. Because $H(a) \leq 2$, 
it follows from Lemma \ref{lem:pool_mis} that we have
\begin{eqnarray}
\frac{\max_{\bm{\theta} \in \Theta(\kappa)} R(\bm{\theta},\hat{\bm{\delta}}^{\bm{w}^{\ast}(\kappa')})}{\max_{\bm{\theta} \in \Theta(\kappa)} R(\bm{\theta},\hat{\bm{\delta}}^{\mathrm{pool}})} & \leq & \frac{\sum_{k=1}^K p_k \cdot \psi_k\left( w_k^{\ast}(\kappa') ;\kappa' \right)}{\max \left\{ \frac{1}{2} s_0 \eta ( \kappa / s_0 ), s_0 \eta ( \kappa / s_0 ) - \kappa \right\}} \times (1+2c). \nonumber 
\end{eqnarray}
Because $w_k^{\ast}(\kappa') = \text{arg} \min_{w_k} \psi_k\left( w_k ;\kappa' \right)$ and $\kappa \geq \kappa'$, we obtain
\begin{eqnarray*}
\psi_k\left( w_k^{\ast}(\kappa') ;\kappa' \right) \ \leq \ \psi_k\left( w_k^{\ast}(\kappa) ;\kappa' \right) \ \leq \ \psi_k\left( w_k^{\ast}(\kappa) ;\kappa \right) \ \leq \ \min \left\{ s_0 \eta(\kappa / s_0), \, \eta(0) \sigma_k \right\}.
\end{eqnarray*}
Hence, we obtain
\begin{eqnarray}
\frac{\max_{\bm{\theta} \in \Theta(\kappa)} R(\bm{\theta},\hat{\bm{\delta}}^{\bm{w}^{\ast}(\kappa')})}{\max_{\bm{\theta} \in \Theta(\kappa)} R(\bm{\theta},\hat{\bm{\delta}}^{\mathrm{pool}})} & \leq &  \min \left\{ 2(1+c), \, \frac{\eta(0) (1+2c) \left( \sum_{k=1}^K p_k \sigma_k \right)}{s_0 \eta(\kappa / s_0)} \right\}. \nonumber 
\end{eqnarray}
Therefore, we obtain the desired results.
\end{proof}

\section*{Appendix B. Additional figures}

\begin{figure}[htbp]
    \begin{tabular}{cc}
      \begin{minipage}[t]{0.45\hsize}
        \centering
        \includegraphics[width=7cm]{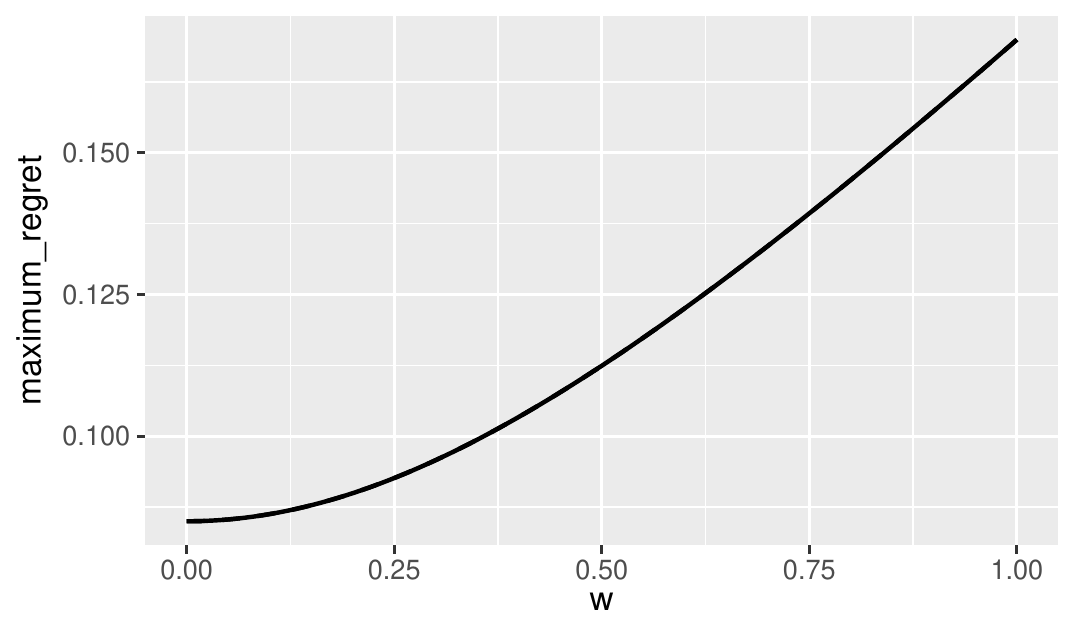}
        \subcaption{$K=4$ and $\kappa = 0$. The minimum values of $\overline{R}_{\mathrm{true}}(w)$ and $\overline{R}_{\mathrm{upper}}(w)$ are $0.085$ and $0.085$, respectively.}
        \label{fig:true_K4_kappa0}
      \end{minipage} &
      \begin{minipage}[t]{0.45\hsize}
        \centering
        \includegraphics[width=7cm]{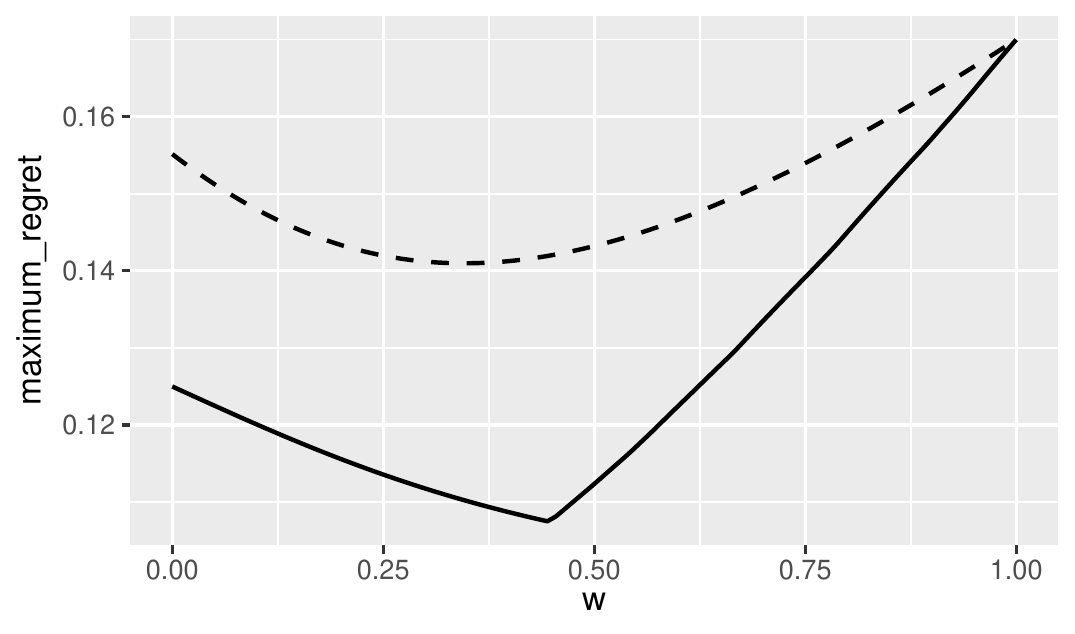}
        \subcaption{$K=4$ and $\kappa = 0.25$. The minimum values of $\overline{R}_{\mathrm{true}}(w)$ and $\overline{R}_{\mathrm{upper}}(w)$ are $0.107$ and $0.141$, respectively.}
        \label{fig:true_K4_kappa025}
      \end{minipage} \\
   
      \begin{minipage}[t]{0.45\hsize}
        \centering
        \includegraphics[width=7cm]{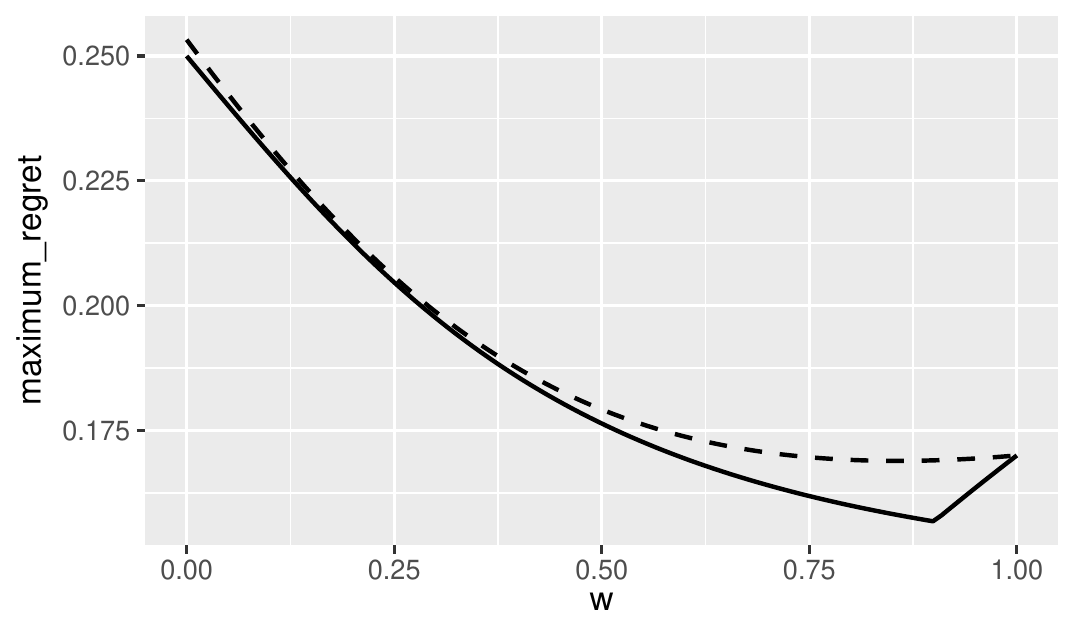}
        \subcaption{$K=4$ and $\kappa = 0.5$. The minimum values of $\overline{R}_{\mathrm{true}}(w)$ and $\overline{R}_{\mathrm{upper}}(w)$ are $0.157$ and $0.169$, respectively.}
        \label{fig:true_K4_kappa050}
      \end{minipage} &
      \begin{minipage}[t]{0.45\hsize}
        \centering
        \includegraphics[width=7cm]{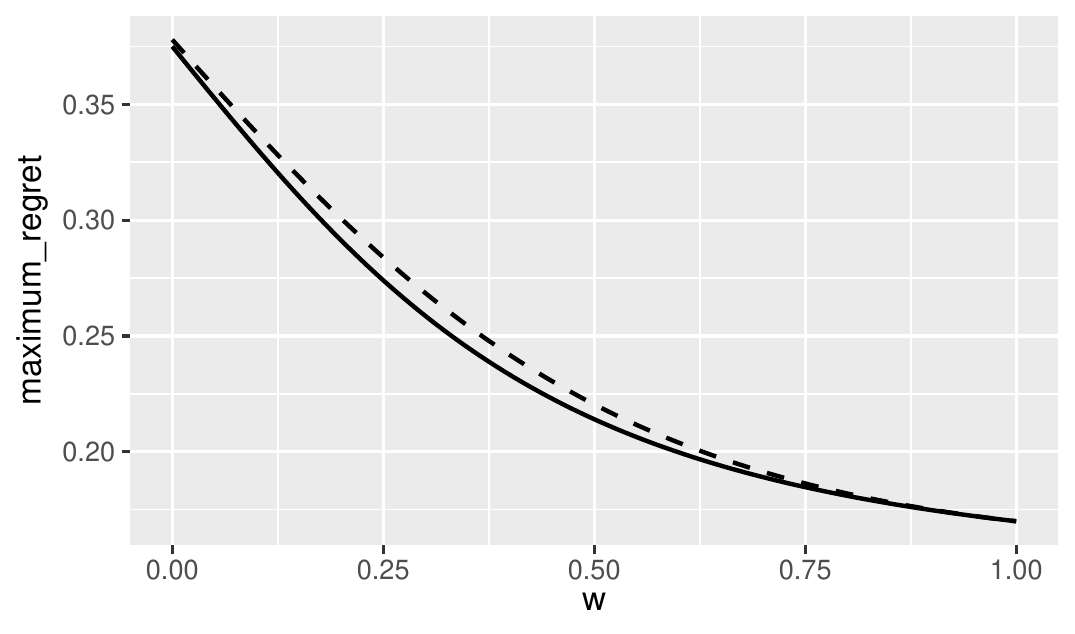}
        \subcaption{$K=4$ and $\kappa = 0.75$. The minimum values of $\overline{R}_{\mathrm{true}}(w)$ and $\overline{R}_{\mathrm{upper}}(w)$ are $0.170$ and $0.170$, respectively.}
        \label{fig:true_K4_kappa075}
      \end{minipage} 
    \end{tabular}
    \caption{The solid and dashed lines denote $\overline{R}_{\mathrm{true}}(w)$ and $\overline{R}_{\mathrm{upper}}(w)$, respectively.}\label{fig:true_K4}
\end{figure}

\begin{figure}[htbp]
    \begin{tabular}{cc}
      \begin{minipage}[t]{0.45\hsize}
        \centering
        \includegraphics[width=7cm]{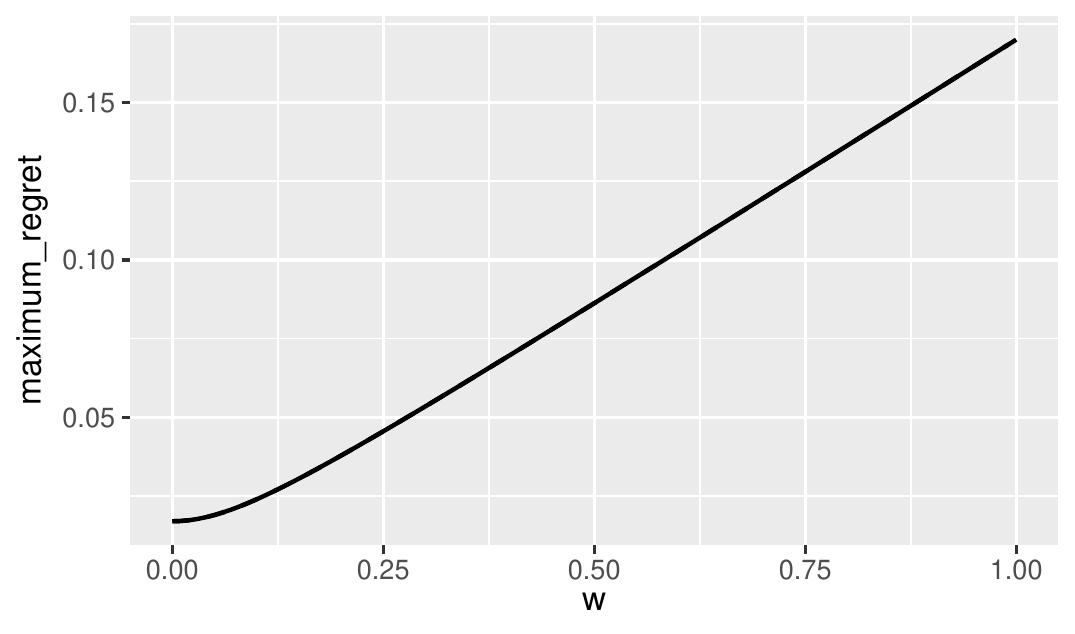}
        \subcaption{$K=100$ and $\kappa = 0$. The minimum values of $\overline{R}_{\mathrm{true}}(w)$ and $\overline{R}_{\mathrm{upper}}(w)$ are $0.017$ and $0.017$, respectively.}
        \label{fig:true_K100_kappa0}
      \end{minipage} &
      \begin{minipage}[t]{0.45\hsize}
        \centering
        \includegraphics[width=7cm]{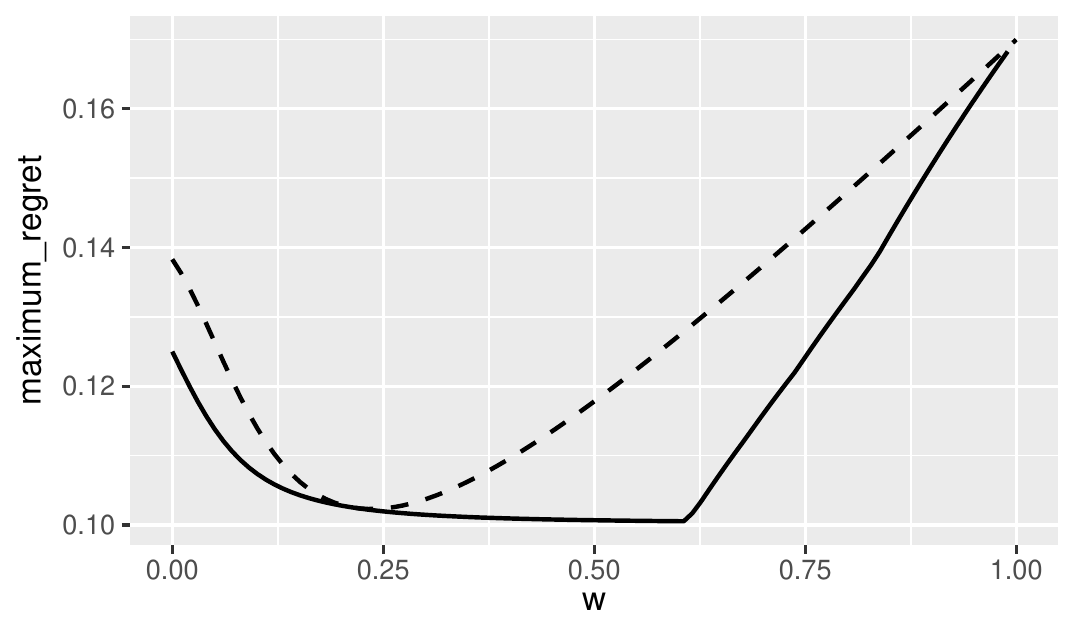}
        \subcaption{$K=100$ and $\kappa = 0.25$. The minimum values of $\overline{R}_{\mathrm{true}}(w)$ and $\overline{R}_{\mathrm{upper}}(w)$ are $0.101$ and $0.102$, respectively.}
        \label{fig:true_K100_kappa025}
      \end{minipage} \\
   
      \begin{minipage}[t]{0.45\hsize}
        \centering
        \includegraphics[width=7cm]{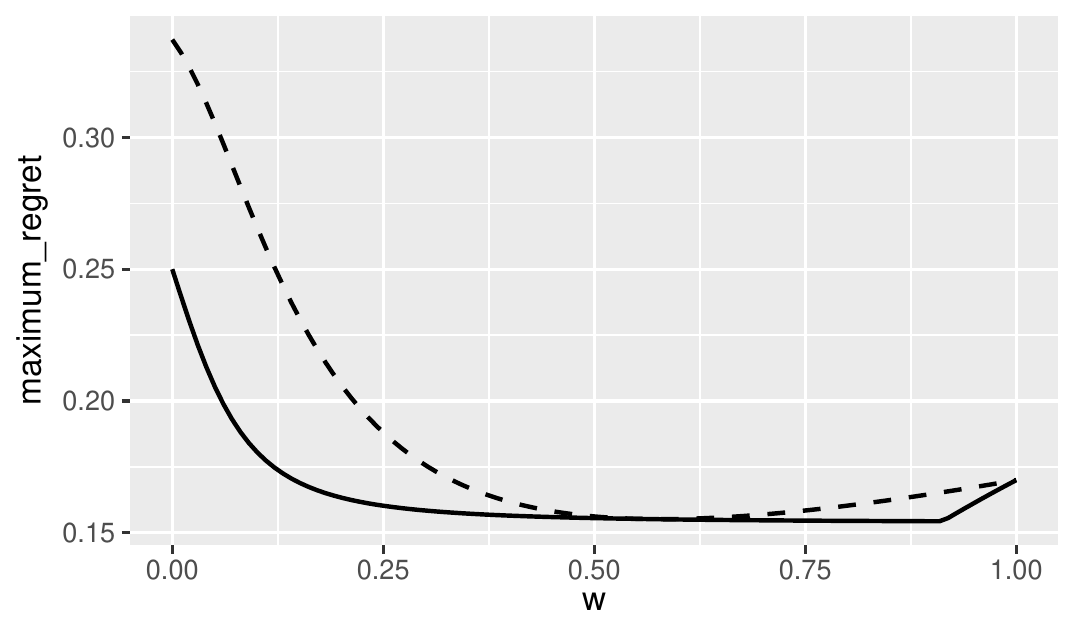}
        \subcaption{$K=100$ and $\kappa = 0.5$. The minimum values of $\overline{R}_{\mathrm{true}}(w)$ and $\overline{R}_{\mathrm{upper}}(w)$ are $0.154$ and $0.155$, respectively.}
        \label{fig:true_K100_kappa050}
      \end{minipage} &
      \begin{minipage}[t]{0.45\hsize}
        \centering
        \includegraphics[width=7cm]{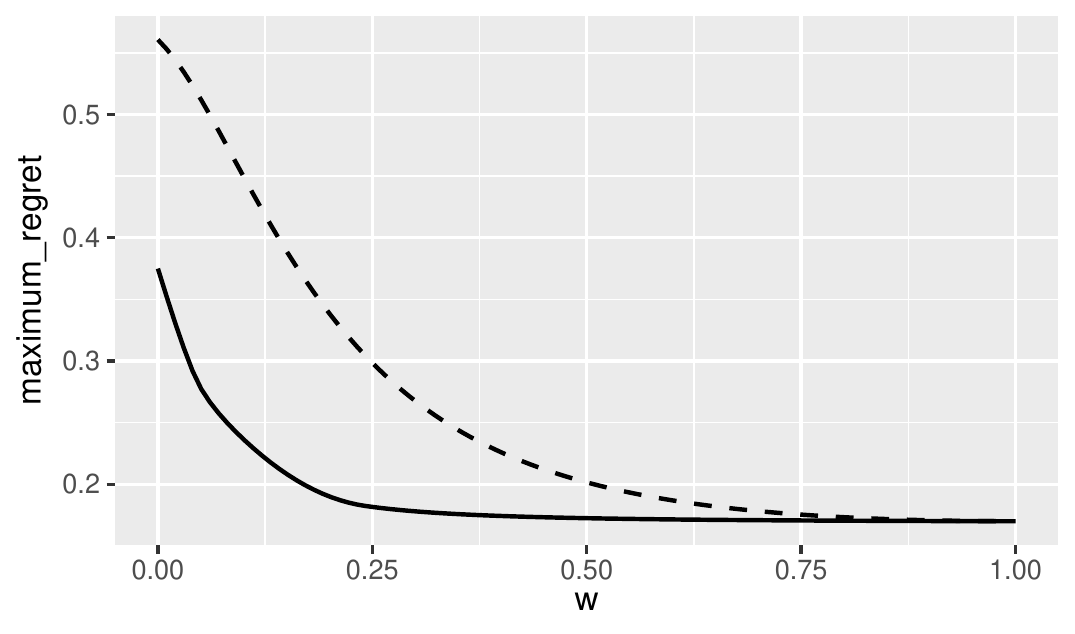}
        \subcaption{$K=100$ and $\kappa = 0.75$. The minimum values of $\overline{R}_{\mathrm{true}}(w)$ and $\overline{R}_{\mathrm{upper}}(w)$ are $0.170$ and $0.170$, respectively.}
        \label{fig:true_K100_kappa075}
      \end{minipage} 
    \end{tabular}
    \caption{The solid and dashed lines denote $\overline{R}_{\mathrm{true}}(w)$ and $\overline{R}_{\mathrm{upper}}(w)$, respectively.}\label{fig:true_K100}
\end{figure}

\begin{figure}[htbp]
    \begin{tabular}{cc}
      \begin{minipage}[t]{0.45\hsize}
        \centering
        \includegraphics[width=7cm]{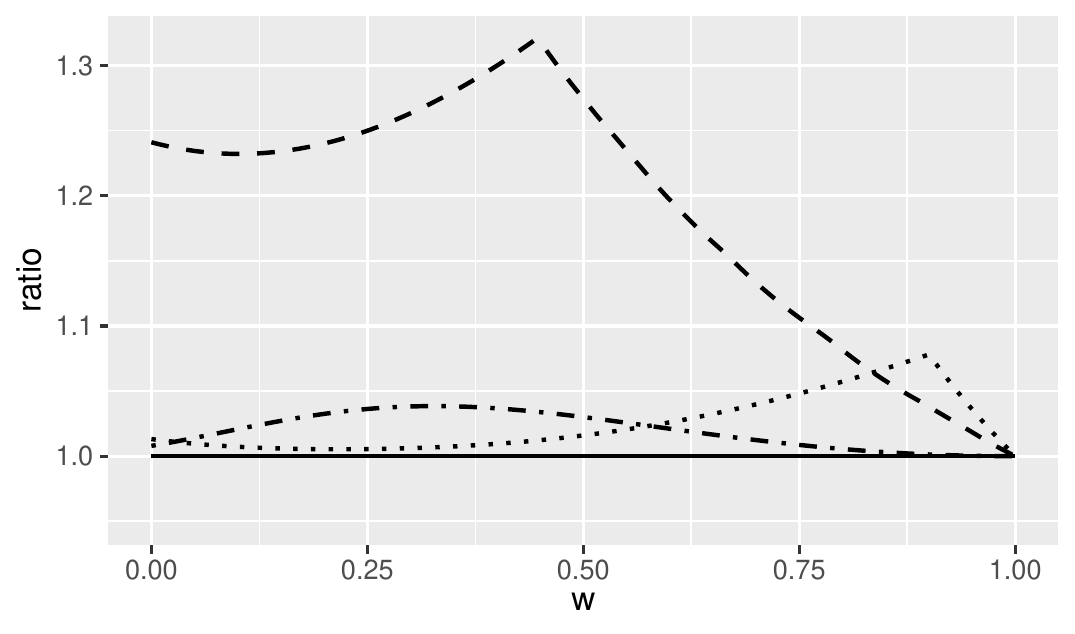}
        \subcaption{$K=4$}
        \label{fig:ratio_true_upper_K4}
      \end{minipage} &
      \begin{minipage}[t]{0.45\hsize}
        \centering
        \includegraphics[width=7cm]{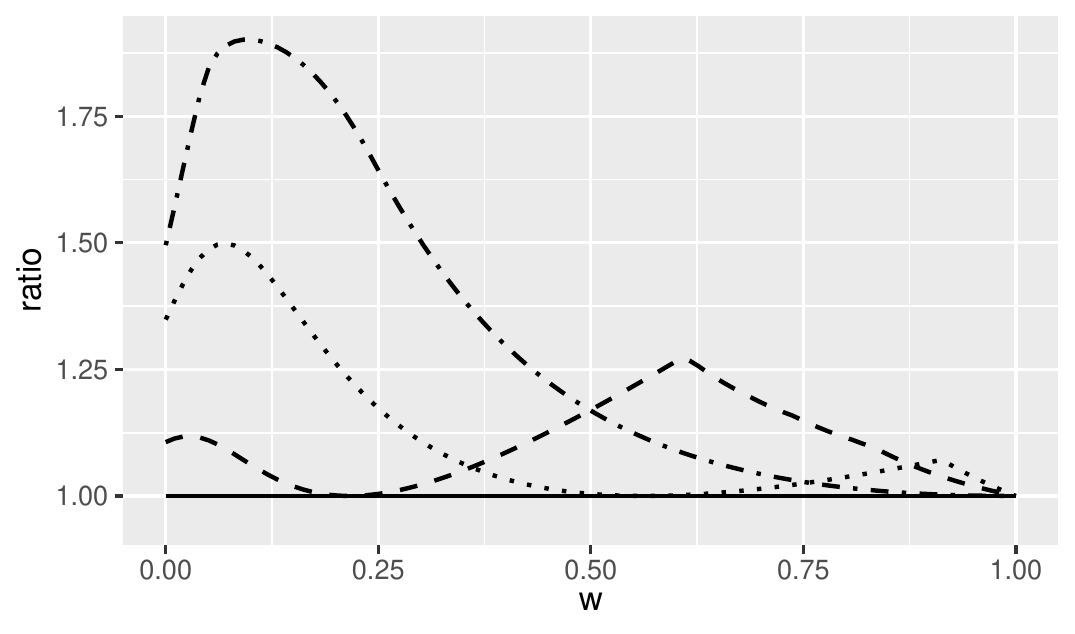}
        \subcaption{$K=100$}
        \label{fig:ratio_true_upper_K100}
      \end{minipage} \\
    \end{tabular}
    \caption{The solid, dashed, dotted, and dot-dashed lines denote $\overline{R}_{\mathrm{upper}}(w) / \overline{R}_{\mathrm{true}}(w)$ for $\kappa = 0, \, 0.25, \, 0.5$, and $0.75$.}\label{fig:ratio_K4_K100}
\end{figure}

\begin{figure}[htbp]
    \begin{tabular}{cc}
      \begin{minipage}[t]{0.45\hsize}
        \centering
        \includegraphics[width=7cm]{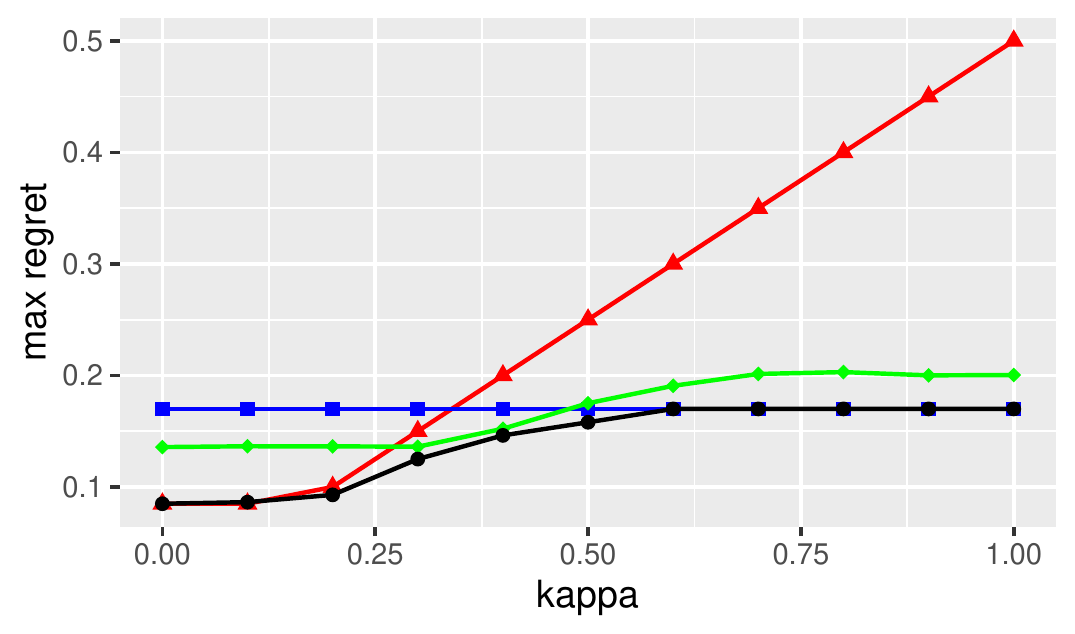}
        \subcaption{Case (i), $K=4$.}
        \label{fig:max_regret_K4_balance}
      \end{minipage} &
      \begin{minipage}[t]{0.45\hsize}
        \centering
        \includegraphics[width=7cm]{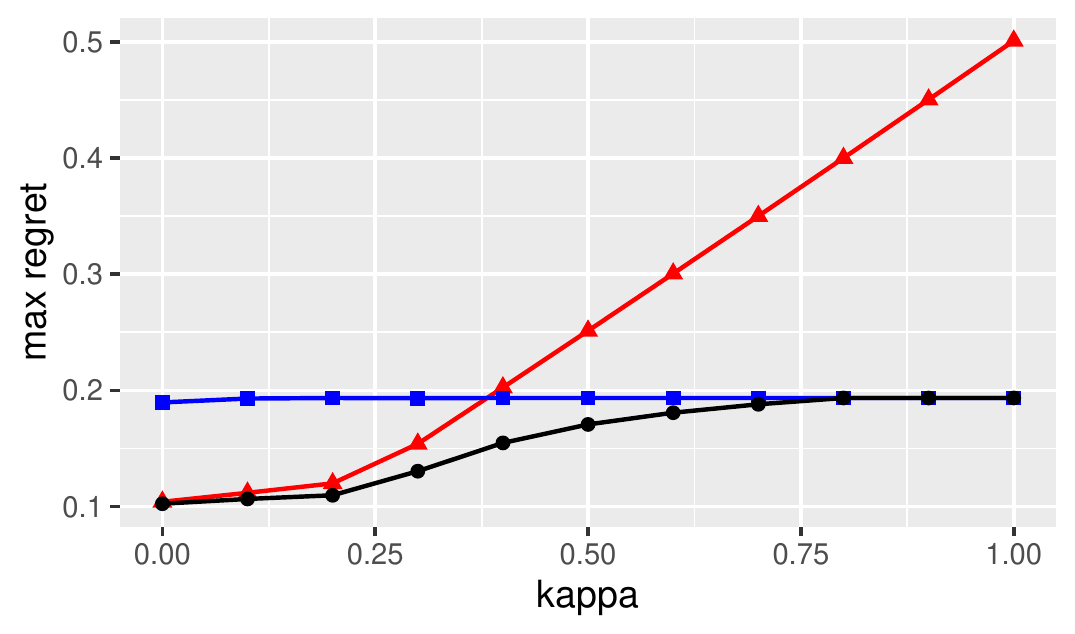}
        \subcaption{Case (ii), $K=4$.}        \label{fig:max_regret_K4_unbalance}
      \end{minipage} \\

      \begin{minipage}[t]{0.45\hsize}
        \centering
        \includegraphics[width=7cm]{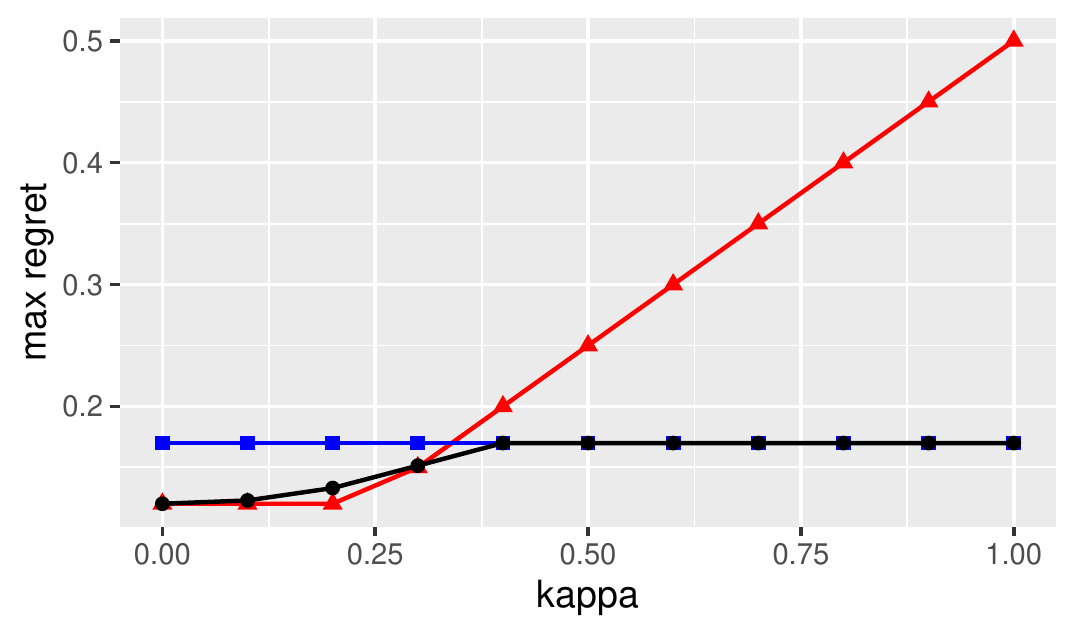}
        \subcaption{Case (i), $K=2$.}
        \label{fig:max_regret_K2_balance}
      \end{minipage} &
      \begin{minipage}[t]{0.45\hsize}
        \centering
        \includegraphics[width=7cm]{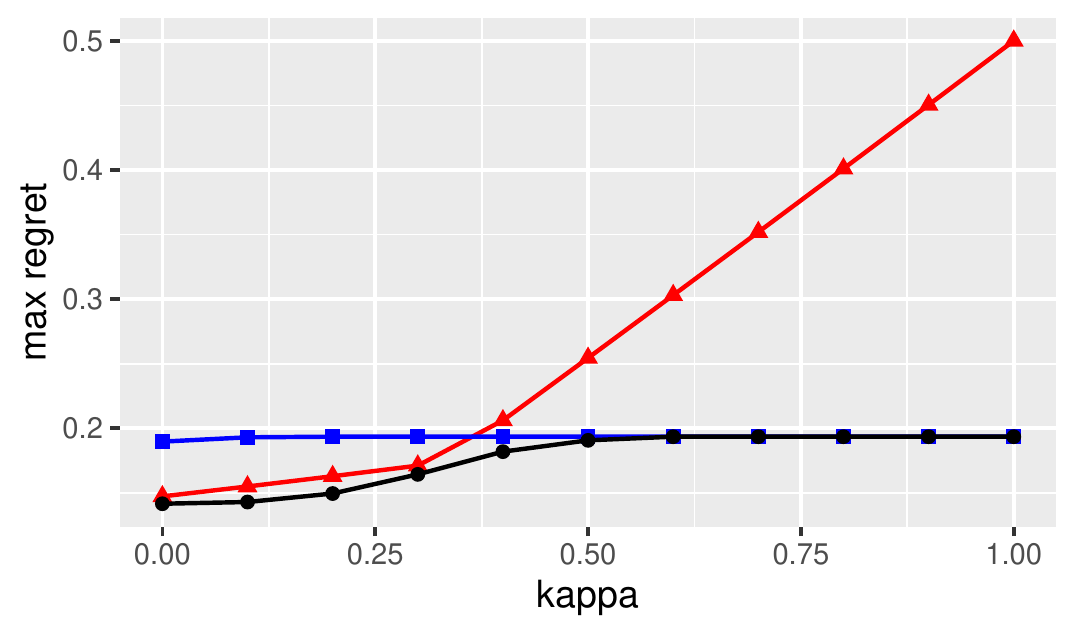}
        \subcaption{Case (ii), $K=2$.}        \label{fig:max_regret_K2_unbalance}
      \end{minipage} \\
    \end{tabular}
    \caption{The black circles, blue squares, and red triangles denote the maximum regrets of the shrinkage, CES, and pooling rules. The green diamonds denote the maximum regrets of the treatment rule based on the James-Stein-type estimator.}\label{fig:max_regret_K4}
\end{figure}

\clearpage
\bibliographystyle{ecta}
\bibliography{meta-analysis}

\end{document}